\documentclass[11pt]{article} 

\newif\ifgamma
\gammafalse

\usepackage{hyperref}
\usepackage{url}
\usepackage{natbib,graphicx,amsmath,amssymb,fullpage,vmargin}
\usepackage{amsfonts}
\usepackage{fancybox}
\usepackage{algorithmic}
\usepackage{algorithm}
\usepackage{comment}
\usepackage{color}
\usepackage{multirow,ctable}
\usepackage[normalem]{ulem}
\usepackage{tikz,enumerate}
\usepackage[framemethod=TikZ]{mdframed}
\usetikzlibrary{arrows}
\newdimen\arrowsize
\pgfarrowsdeclare{arcsq}{arcsq}
{
  \arrowsize=0.2pt
  \advance\arrowsize by .5\pgflinewidth
  \pgfarrowsleftextend{-4\arrowsize-.5\pgflinewidth}
  \pgfarrowsrightextend{.5\pgflinewidth}
}
{
  \arrowsize=1.5pt
  \advance\arrowsize by .5\pgflinewidth
  \pgfsetdash{}{0pt} 
  \pgfsetroundjoin   
  \pgfsetroundcap    
  \pgfpathmoveto{\pgfpoint{0\arrowsize}{0\arrowsize}}
  \pgfpatharc{-90}{-140}{4\arrowsize}
  \pgfusepathqstroke
  \pgfpathmoveto{\pgfpointorigin}
  \pgfpatharc{90}{140}{4\arrowsize}
  \pgfusepathqstroke
}

\usepackage{paralist,fullpage}
\newcommand{\independent}{\mbox{${}\perp\mkern-11mu\perp{}$}}

\newcommand{\iids}{\overset{\text{iid}}{\sim}}
\newcommand{\equald}{\overset{d}{=}}

\newcommand{\C}[1]{\mathcal{#1}}
\newcommand{\B}[1]{\mathbf{#1}}

\newcommand{\II}{{U}}
\newcommand{\I}{{\mathcal J}}
\newcommand{\E}{{\mathcal{E}}}
\renewcommand{\i}{{j}}
\newcommand{\e}{{e}}

\newcommand{\mean}{{E}}  
\newcommand{\var}{{\mathrm{Var}}}  

\newcommand{\pa}[1]{\B{PA}({#1})}
\newcommand{\PA}[1]{{\B{PA}}{(#1)}}

\newcommand{\DE}[1]{{\B{DE}}{(#1)}}
\newcommand{\ND}[1]{{\B{ND}}{(#1)}}
\newcommand{\AN}[1]{{\B{AN}}{(#1)}}

\newcommand{\given}{\,|\,}

\newcommand{\caus}{*}
\newcommand{\pred}{\mathrm{pred}}

\newtheorem{theorem}{Theorem}
\newtheorem{lemma}{Lemma}
\newtheorem{remark}{Remark}
\newtheorem{prop}{Proposition}

\newtheorem{defin}{Definition}
\newtheorem{assumption}{Assumption}

\newcommand{\iid}{i.i.d.\ }

\newcommand\Nicolai[1]{{\color{red}Nicolai: ``#1''}}

\newcommand\Jonas[1]{{\color{blue}Jonas: ``#1''}}

\graphicspath{{./}}

\newenvironment{proof}[1][. ]{{\bf Proof#1}}{\hfill$\square$\vskip\baselineskip}

\title{Causal inference using invariant prediction: \\ identification and confidence intervals}

\author{Jonas Peters$ ^{\flat, \sharp}$, Peter B\"uhlmann$ ^{\sharp}$ and Nicolai Meinshausen$ ^{\sharp}$ \vspace{0.1cm}\\
$ ^{\flat}$MPI for Intelligent Systems, T\"ubingen, Germany\\ 
$ ^{\sharp}$Seminar f\"ur Statistik, ETH Z\"urich, Switzerland \vspace{0.1cm}\\ \{peters,buhlmann,meinshausen\}@stat.math.ethz.ch}

\begin{document}

\maketitle

\begin{abstract}
What is the difference of a prediction that is made with a
  causal model and a non-causal model? Suppose we intervene on the
  predictor variables or
  change the whole environment. The predictions from a causal model
  will in general work as well under interventions as for observational data.  In
  contrast, predictions from
  a non-causal model can potentially be very wrong if we actively
  intervene on variables.
Here,
we propose to exploit this  invariance of a prediction under a causal model  for causal inference: given
different experimental settings (for example various interventions) we
collect all models that do show invariance in their predictive
accuracy across settings and interventions. The
causal model will be a member of this set of models with high probability. This approach
yields valid confidence intervals for the causal relationships in
quite general scenarios. 
We examine 
the example of 
structural equation models in more detail
and provide sufficient assumptions under which 
the set of causal predictors becomes identifiable.
We further investigate robustness properties of our approach under model misspecification and discuss possible extensions.
 The empirical properties are studied for various data sets,
including large-scale gene perturbation experiments.
\end{abstract}

\section{Introduction}
Inferring cause-effect relationships between variables is a primary goal in
many applications. Such causal inference has its roots in different fields
and various concepts have contributed to its understanding and
quantification. Among them are the framework of potential outcomes and
counterfactuals \citep[cf.][]{dawid2000causal,rubin2005causal}; or structural equation modelling
\citep[cf.][]{Bollen1989,robins2000marginal,Pearl2009} and 
graphical modeling \citep[cf.][]{lauritzen1988local,greenland1999causal, Spirtes2000},
where the book by 
\citet{Pearl2009} provides a nice overview. 
\citet{richardson2013single} make a 
connection between the frameworks using single-world intervention
graphs. 

A typical approach for
causal discovery, in the context of unknown
causal structure, is to characterise the Markov equivalence class of structures
(or graphs) \citep{Verma1991,ander97,Tian2001, Hauser2012}, estimate
the correct 
Markov equivalence class based on
observational or interventional data  
\citep[cf.]{Spirtes2000,Chickering2002,castelo2003inclusion,Kalisch2007,hegeng08,hapb14}, 
and finally infer the identifiable causal effects or provide some bounds
\citep[cf.]{Maathuis2009,vanderweele2010signed}. More recently, within the
framework of structural equation models, interesting work has been done for
fully identifiable structures exploiting additional restrictions such as 
non-Gaussianity \citep{Shimizu2006}, nonlinearity \citep{Hoyer2008,Peters2014JMLR} or equal error variances \citep{Peters2012}. \citet{Janzing2012} exploit an independence between causal mechanisms. 

We propose here a new method for causal discovery. 
The approach of the paper is to note that  if we consider all ``direct causes'' of a target variable of
interest, then the conditional distribution of the target given the
the direct causes will not
change when we interfere experimentally with all other variables in the model except the
target itself. 
This does not necessarily hold, however, if some of the direct causes are ignored in the conditioning.\footnote{We thank a referee for suggesting this succinct
  description of the main idea.}
We exploit, in other words, that 
the conditional
distribution of the target variable
of interest (often also termed ``response variable''), given the
complete set of 
corresponding direct causal  
predictors, has to remain identical under interventions on variables other
than the target variable.
This invariance idea is closely linked to causality and has been discussed,
for example, under the term ``autonomy'' and ``modularity''
\citep{Haavelmo1944, Aldrich1989, Hoover1990, Pearl2009, Schoelkopf2012} or
also ``stability'' \citep{Dawid2010Didilez2010} \citep[][Sec. 1.3.2]{Pearl2009}. While it is well-known that causal models have an invariance property, we try to
exploit this fact for inference. 
Our proposed procedure gathers all submodels that are statistically
invariant across environments in a suitable sense. The causal submodel
consisting of the set of variables with a direct causal effect on the
target variable will be one of 
these invariant submodels, with controlled high probability, and this
allows to control the probability of making false causal discoveries.

Our method is tailored for (but not restricted to) the setting where we
have data 
from different experimental settings or regimes \citep{didelez2006direct}. For example, two
different interventional data samples, or a combination of
observational and interventional data \citep[cf.][]{hegeng08} belong to such
a scenario. 
For known intervention targets, \citet{cooper1999causal} incorporate the intervention effects as mechanism changes \citep{Tian2001} into a Bayesian framework and \citet{hapb14} modify the greedy equivalence search \citep{Chickering2002} for perfect interventions. 
Our framework does not require to know the location of interventions. For
this setting, \citet{Eaton2007} use
intervention nodes with unknown children and \citet{Tian2001} consider
changes in marginal distributions, while 
\citet{dawid2012decision,dawid2015statistical} make use of different
regimes for a decision-theoretic approach.
In contrast to these approaches, our framework does not require the 
fitting of graphical, structural equation or potential outcome
models and comes with statistical guarantees. 
Further advantages are indicated below in
Section~\ref{subsec.ownc}.

We primarily consider the situation with no hidden
(confounder) variables that influence the target variable.
A rigorous treatment with hidden variables would be
more involved \citep[see][for graphical language]{richardson2002ancestral} but we provide an 
example with instrumental variables in Section \ref{subsec.instrvar} to
illustrate that the method could also work 
more generally in the context of hidden variables. We do not touch very much on the
framework of feedback models
\citep[cf.]{lauritzen2002chain, Mooij2011, Hyttinen2012},
although a constrained form of feedback is allowed. It is an
open question whether our approach could be generalised 
to include general feedback models.  

\subsection{Data from multiple environments or experimental settings}\label{subsec.data}

We consider the setting where we have different experimental conditions $\e
\in \E$ and have an \iid sample of $(X^\e,Y^\e)$ in each environment, where
$X^\e\in \mathbb{R}^p$ is a predictor variable and $Y^\e \in \mathbb{R}$ a target
variable of interest. While the environments $e \in \E$ can be created by precise
experimental design for~$X^\e$ (for example by randomising some or all
elements of~$X^\e$), we are more interested in settings where such
careful experimentation is not possible and the different
distributions of $X^\e$ in the environments are generated by unknown and
not precisely controlled  interventions. 
 If a subset $S^*\subseteq\{1,\ldots,p\}$ is 
causal for the prediction of a response $Y$, we assume that
\begin{align}  
\mbox{for all } \e \in \E:\quad  & X^e \mbox{ has an arbitrary
  distribution and }  \label{eq:fullmodel} \\  
&  Y^\e = g(X^\e_{S^*},  \varepsilon^\e), \quad \varepsilon^\e
\sim F_{\varepsilon} \mbox{  and } \varepsilon^\e \independent X^\e_{S^*}, \label{invariance-nonlin}
\end{align}
where  $g:\mathbb{R}^{|S^*|}\times \mathbb{R} \rightarrow \mathbb{R}$ is a real-valued function in a suitable function class, $X_{S^*}^\e$ is the vector of predictors $X^\e$ with indices in a set 
$S^*$ and both the error distribution $\varepsilon^\e \sim F_\varepsilon$ and the function $g$ are
assumed to be the same for all
the experimental settings. 
Equations~\eqref{eq:fullmodel} and~\eqref{invariance-nonlin} can also be interpreted as requiring that the conditionals $Y^\e \given X^\e_{S^*}$ and $Y^f \given X^f_{S^*}$ are identical for all environments $e,f \in \E$ (this equivalence is proved in Section~\ref{sec:nonlin}).

\definecolor{newblue}{RGB}{62,76,209}
\definecolor{newblueLight}{RGB}{184,189,238}

\if0
\hfsetbordercolor{gray!80}
\hfsetfillcolor{newblue!30}
\hfsetfillcolor{newblueLight!100}

\begin{figure}
\begin{center}
\tikzmarkin{gglw3}(-0.25,-2.4)(0.4,2.8)
\hfsetbordercolor{gray!80}
\hfsetfillcolor{gray!40}
\begin{minipage}[][][c]{0.28\textwidth}
\begin{center}
{\small
environment $e=1$:} \vspace{0.3cm}\\

$ $ \hspace{-0.37cm}\begin{tikzpicture}[scale=1, line width=0.5pt, minimum size=0.58cm, inner sep=0.3mm, shorten >=1pt, shorten <=1pt]
    \small
    \draw (1,3) node(x5) [circle, draw] {$X_5$};
    \draw (0,2) node(x2) [circle, draw] {$X_2$};
    \draw (2,2) node(x4) [circle, draw] {$X_4$};
    \draw (1,1) node(y) [circle, draw] {$Y$};
    \draw (1,0) node(x3) [circle, draw] {$X_3$};
    \draw (0,-0.38) node(x) [circle] {$ $};
    \draw[-arcsq] (x5) -- (x4);
    \draw[-arcsq] (x2) -- (y);
    \draw[-arcsq] (x4) -- (y);
    \draw[-arcsq] (y) -- (x3);
   \end{tikzpicture}
   \end{center}
\end{minipage}
\tikzmarkend{gglw3}
\definecolor{neworange}{RGB}{221,120,63}
\definecolor{neworangeLight}{RGB}{242,206,185}
\hfsetfillcolor{neworange!30}
\hfsetfillcolor{neworangeLight!100}
\tikzmarkin{gglw5}(-0.14,-2.4)(-0.15,2.8)
\begin{minipage}[][][c]{.28\textwidth}
\begin{center}
{\small
environment $e=2$:} \vspace{0.3cm}\\

$ $\hspace{-0.37cm}\begin{tikzpicture}[scale=1, line width=0.5pt, minimum size=0.58cm, inner sep=0.3mm, shorten >=1pt, shorten <=1pt]
    \small
    \draw (1,3) node(x5) [circle, draw] {$X_5$};
    \draw (0,2) node(x2) [circle, draw] {$X_2$};
    \draw (2,2) node(x4) [circle, draw] {$X_4$};
    \draw (1,1) node(y) [circle, draw] {$Y$};
    \draw (1,0) node(x3) [circle, draw] {$X_3$};
    \draw (2.7,0) node(x) [circle] {$ $};
    \draw[-arcsq] (x5) -- (x4);
    \draw[-arcsq] (x2) -- (y);
    \draw[-arcsq] (x4) -- (y);
    \draw[-arcsq] (y) -- (x3);
    \node (myfirstpic2) at (-1,2) {\includegraphics[angle = 60, scale = 0.02]{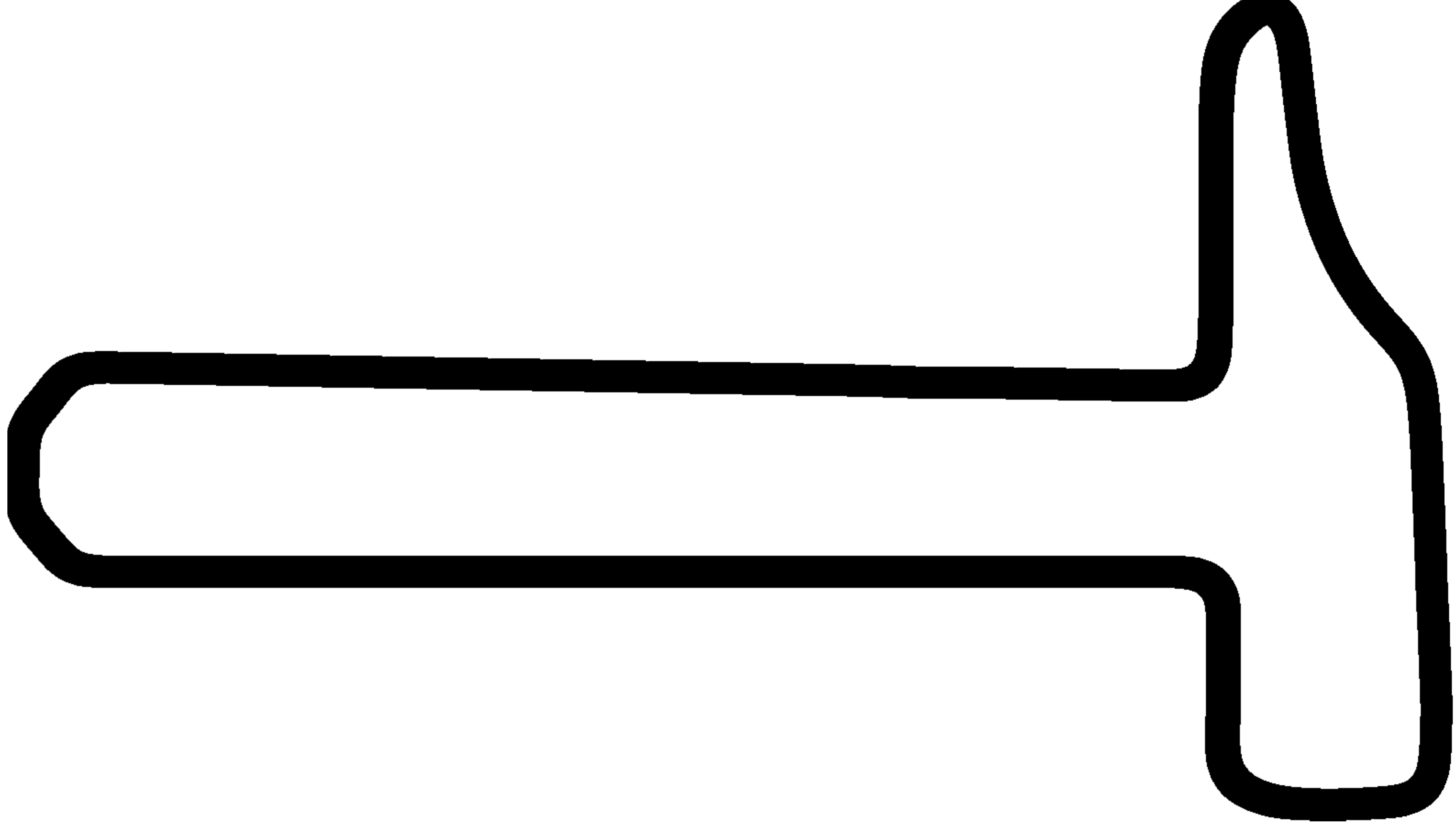}};
    \node (myfirstpic) at (2,0) {\includegraphics[angle = -60, scale = 0.02]{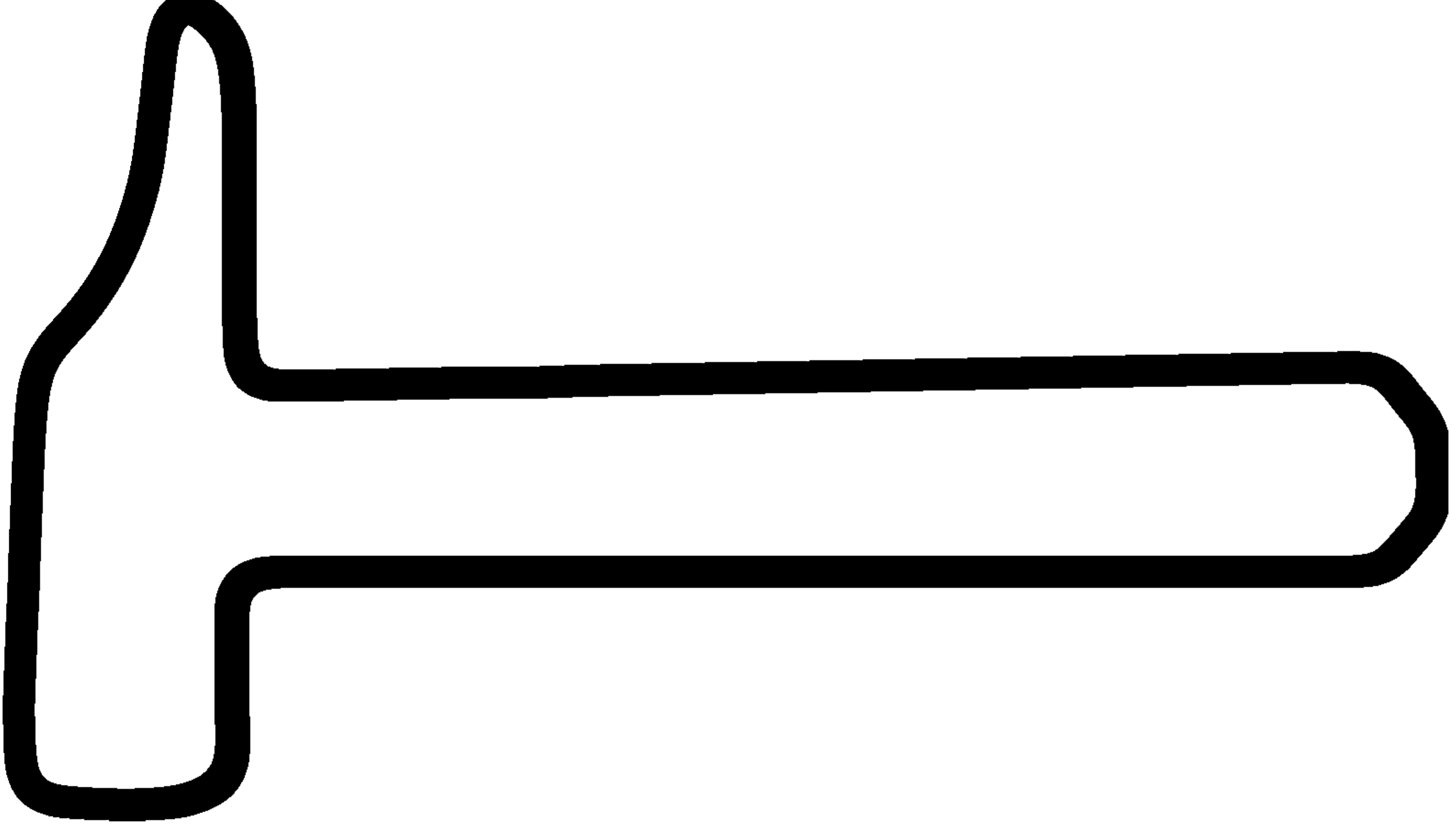}};
   \end{tikzpicture}
   \end{center}
\end{minipage}
\tikzmarkend{gglw5}
\hspace{-0.3cm}
\definecolor{neworange}{RGB}{221,120,63}
\definecolor{neworangeLight}{RGB}{242,206,185}
\hfsetbordercolor{gray!80}
\hfsetfillcolor{gray!40}
\tikzmarkin{gglw4}(-0.1,-2.4)(0.2,2.8)
\begin{minipage}[][][c]{0.28\textwidth}
\begin{center}
{\small
environment $e=3$:} \vspace{0.3cm}\\

$ $\hspace{-0.4cm}\begin{tikzpicture}[scale=1, line width=0.5pt, minimum size=0.58cm, inner sep=0.3mm, shorten >=1pt, shorten <=1pt]
    \small
    \draw (1,3) node(x5) [circle, draw] {$X_5$};
    \draw (0,2) node(x2) [circle, draw] {$X_2$};
    \draw (2,2) node(x4) [circle, draw] {$X_4$};
    \draw (1,1) node(y) [circle, draw] {$Y$};
    \draw (1,0) node(x3) [circle, draw] {$X_3$};
    \draw (-0.7,0) node(x) [circle] {$ $};
    \draw (0,-0.38) node(x11) [circle] {$ $};
    \draw[-arcsq] (x5) -- (x4);
    \draw[-arcsq] (x2) -- (y);
    \draw[-arcsq] (x4) -- (y);
    \draw[-arcsq] (y) -- (x3);
    \node (myfirstpic) at (3,2) {\includegraphics[angle = -60, scale = 0.02]{hammer2.pdf}};
   \end{tikzpicture}
   \end{center}
\end{minipage}
\tikzmarkend{gglw4}\\
$ $
\caption{ An example including three environments. The invariance~\eqref{eq:fullmodel} and~\eqref{invariance-nonlin} holds if we consider $S^* = \{X_2, X_4\}$. Considering indirect causes instead of direct ones (e.g. $\{X_2, X_5\}$) or an incomplete set of direct causes (e.g. $\{X_4\}$) may not be sufficient to guarantee invariant prediction.}
\label{fig:newexample}
\end{center}
\end{figure}

\else
\begin{figure}
\begin{center}
\includegraphics[width=0.99\textwidth]{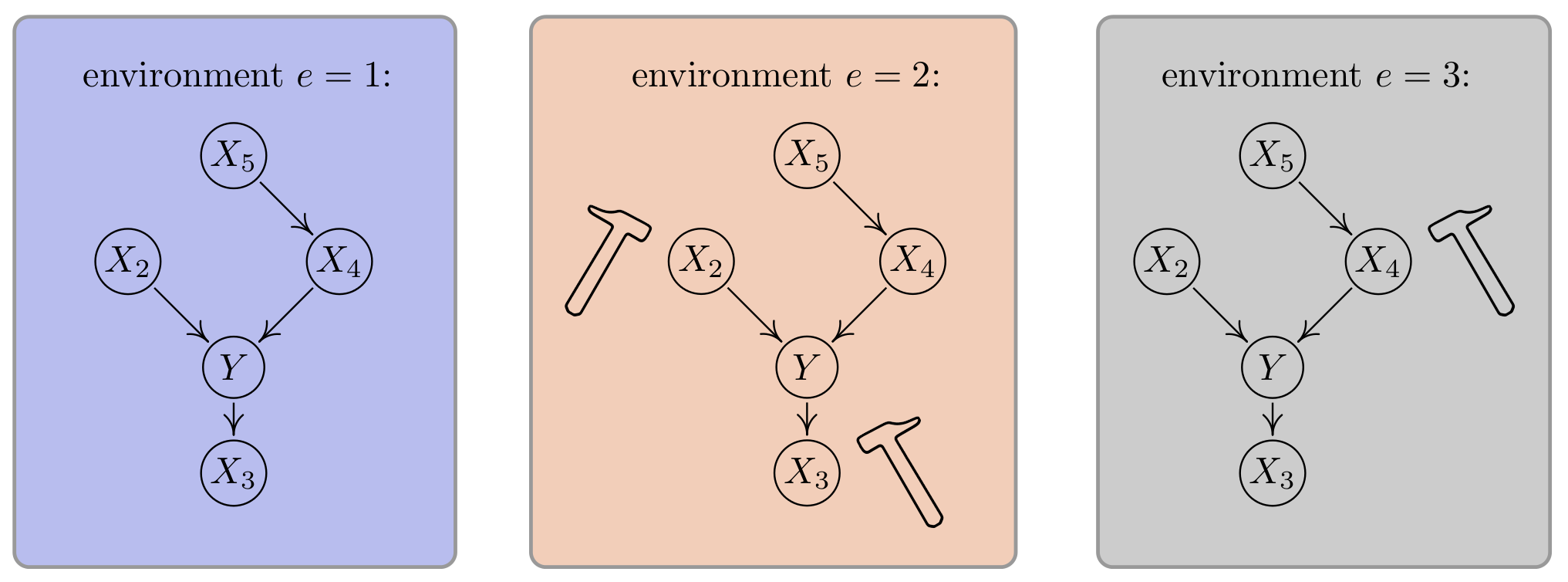}
\caption{ An example including three environments. The invariance~\eqref{eq:fullmodel} and~\eqref{invariance-nonlin} holds if we consider $S^* = \{X_2, X_4\}$. Considering indirect causes instead of direct ones (e.g. $\{X_2, X_5\}$) or an incomplete set of direct causes (e.g. $\{X_4\}$) may not be sufficient to guarantee invariant prediction.}
\label{fig:newexample}
\end{center}
\end{figure}

\fi

An example of a set of environments can be seen in
  Figure~\ref{fig:newexample}. The invariance~\eqref{eq:fullmodel} and~\eqref{invariance-nonlin} holds if the set $S^*$
  consists of all direct causes of the target variable $Y$ and if
we do not intervene on $Y$, see Proposition~\ref{propos:sem}.

Sections~\ref{subsec.instrvar},~\ref{sec:violations} and~\ref{sec:modelmismain} discuss
violations and possible relaxations of this assumption.

\subsection{New contribution}\label{subsec.ownc}

The main and novel idea 
is that  we can use the invariance of the causal relationships under
different settings $e\in \E$ for statistical estimation, which opens a
new road for causal discovery and inference.

For the sake of simplicity, we will mostly focus on a linear model with a
target or response variable and various predictor variables,
where Equation~\eqref{eq:fullmodel} is unchanged
and~\eqref{invariance-nonlin} then reads  
$Y^\e = \mu+X^\e \gamma^* +
\varepsilon^\e$, 
with $\mu$ a constant intercept term. The set $S^*$ of predictors is then given by the support
of $\gamma^*$, that is $S^* := \{k;\ \gamma^*_k \neq 0\}$. 
Assumption~\ref{assum:invariant} in Section~\ref{sec:model} summarises all requirements. 
Proposition~\ref{propos:sem} shows that structural equation models with the traditional notion of 
interventions \citep{Pearl2009} satisfy Assumption~\ref{assum:invariant}
if we choose the set $S^*$ to be the parents of $Y$. 
Proposition~\ref{prop:potentialoutcomes} in Appendix~\ref{app:potentialoutcomes} sheds some light on the relationship to potential outcomes.

Obtaining confidence statements for existing causal discovery methods 
is often difficult as one would need to determine the
distribution of causal effects estimators after having searched and estimated a
graphical structure of the model. It is unknown how one could do
this, except relying on data-splitting strategies which have been found to
perform rather poorly in such a setting \citep{buhlmann2013controlling}. 
We propose in Section \ref{sec:conf} a new method for the 
construction of (potentially) conservative confidence 
statements for causal predictors~$S^*$ and 
of (potentially) conservative intervals for
$\gamma^*_j$ for $j = 1,\ldots ,p$ without a-priori knowing or
assuming a causal ordering
of variables. 
The method provides confidence intervals without relying on
assumptions such as faithfulness or other identifiability assumptions. If
a causal effect is not identifiable from the given data, it would automatically detect
this fact and not make false causal discoveries.

Another main advantage of our methodology is that we do not need to know how the
experimental conditions arise or which type of interventions they
induce. We only assume that the intervention does not change the conditional 
distribution of the target given the causal predictors
(no intervention on the target or a hidden confounder): it
is simply a device exploiting the grouping of data  
into blocks, where every block corresponds to an experimental condition $\e
\in \E$. We will show in Section \ref{sec:pooling} that such grouping can be
misspecified and the coverage statements are still correct. 
This
is again a major bonus in 
practice as it is often difficult to specify what an intervention or change
of environment actually means. In
contrast, for a so-called
do-intervention for structural equation models \citep{Pearl2009} it needs to be specified on which variables it
acts.  
Interesting areas of
applications include studies where observational data alone are not
sufficient to infer causal effects but randomised studies are
infeasible to conduct. 

We believe that the method's underlying invariance principle is rather
general. However, for simplicity, we present our main results for linear
Gaussian models, including some settings with instrumental variables and hidden variables.

\subsection{Organization}
The invariance assumption is formulated and discussed in
Section~\ref{sec:model}. Using this invariance assumption, a general way to construct confidence statements
for causal predictors and associated coefficients is derived in
Section~\ref{sec:conf}. Two specific methods are shown, using
regression effects for various sets of predictors as the main
ingredient.
Identifiability results for
structural equation models are given in Section~\ref{sec:iden}.
The relation to instrumental variables and the behaviour in presence
of hidden variables is discussed in 
Section~\ref{subsec.instrvar}. 
We will discuss extensions to the nonlinear model~\eqref{invariance-nonlin} in Section~\ref{sec:nonlin} and extenstions to intervened targets in Section~\ref{sec:violations}.
Some robustness property against model misspecifications is discussed in Section~\ref{sec:modelmismain}.

Simulations and applications to a biological gene perturbation data set and an
educational study related to instrumental variables are presented in
Section~\ref{sec:numerical}. We discuss the results and provide an outlook
in Section~\ref{sec:discussion}. 

\subsection{Software}
The methods are available 
in the package
\texttt{InvariantCausalPrediction} for the
\texttt{R}-language \citep{R}.

\section{Assumed invariance of causal prediction}  \label{sec:model}
We formulate here the invariance assumption 
and
discuss the notion of identifiable causal predictors. 
Let $\E $ denote again the index set of $|\E|$
possible interventional or experimental settings. As stated above, we have variables $(X^\e,Y^\e)$ with a joint distribution
that will in general depend on the environment $\e\in
\E$.
In the simplest case, $|\E|=2$, and we have for example in the first 
  setting 
observational data and interventions of some (possibly unknown) nature in
the second setting. 

Our  discussion will rest on the following assumption. 
We assume the existence of a model that is invariant
under different experimental or intervention settings. 
Let
for any set $S \subseteq \{1, \ldots, p\}$, 
$X_S$ be the vector containing all variables
$X_k, k \in S$.
\begin{assumption} [Invariant prediction]
 \label{assum:invariant} 
  There exists a  vector of coefficients $\gamma^{\caus}=(\gamma^{\caus}_1,\ldots,
\gamma^{\caus}_p)^t$ with support $S^{\caus} :=\{k: \gamma^{\caus}_k\neq 0\}  \subseteq \{1,\ldots,p\}$ that satisfies
\begin{align}
\mbox{for all } \e \in \E:\quad  &X^\e \mbox{ has an arbitrary distribution } \mbox{ and }\nonumber  \\ & Y^\e = \mu + X^\e \gamma^* + \varepsilon^\e, \quad \varepsilon^\e \sim F_\varepsilon \mbox{ and } \varepsilon^\e \independent X_{S^*}^\e , \label{eq:lincausal} 
\end{align}
where $\mu\in\mathbb{R}$ is an intercept term, $\varepsilon^\e$ is random noise  with
mean zero, finite variance and the same  distribution $F_{\varepsilon}$
across all $\e \in \E$. 
\end{assumption}
The distribution $F_\varepsilon$ is not assumed to be known in general.
If not mentioned otherwise, we will always assume that an
intercept $\mu$ is added to the model~\eqref{eq:lincausal}. To simplify notation, we will from now on
refrain from writing the intercept down explicitly. We  discuss the invariance assumption with the help of some examples in Figure~\ref{fig:newexample} and~\ref{fig:examplesK}; see also Appendix~\ref{app:example} for another artificial example.
\begin{figure}
\begin{center}
\includegraphics[width=0.45\textwidth]{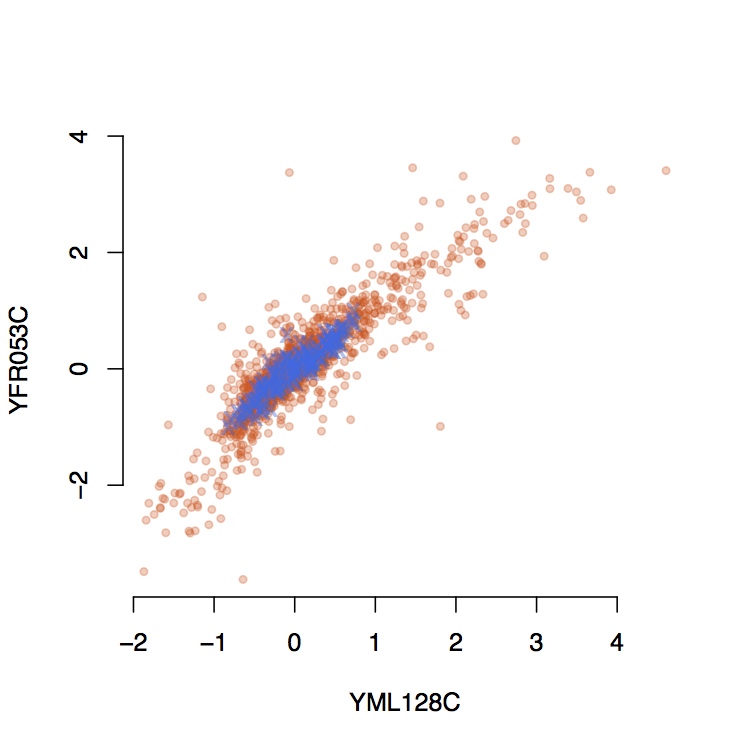}
\includegraphics[width=0.45\textwidth]{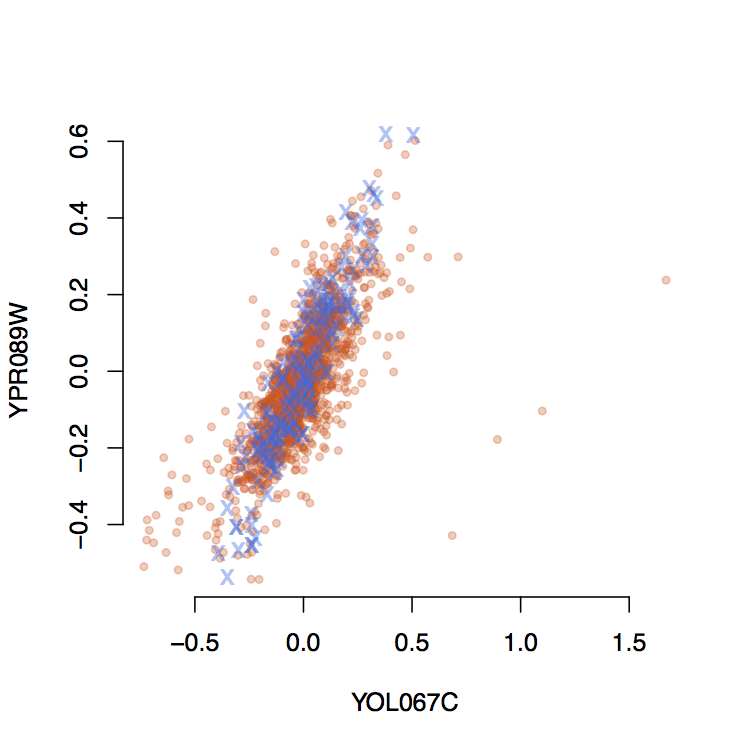}

\includegraphics[width=0.45\textwidth]{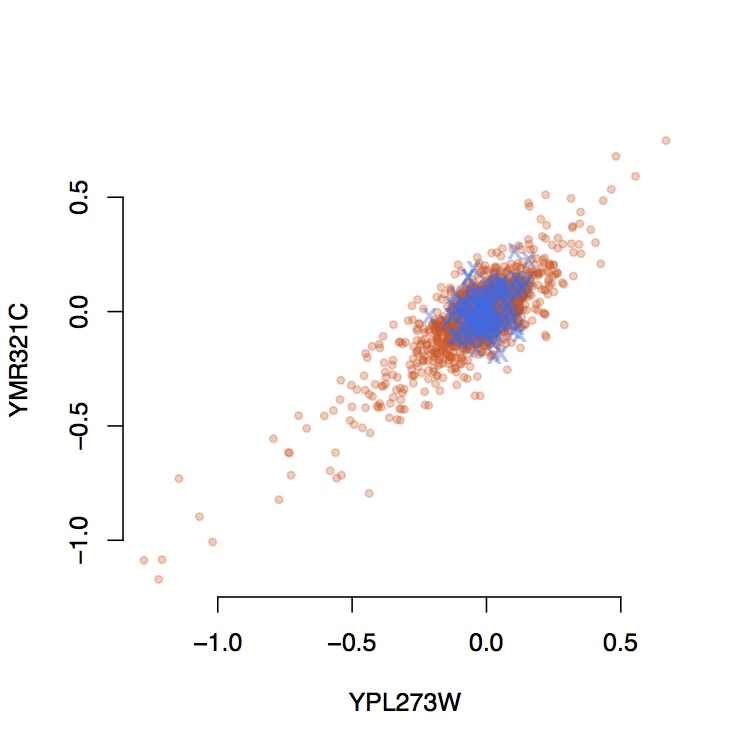}
\includegraphics[width=0.45\textwidth]{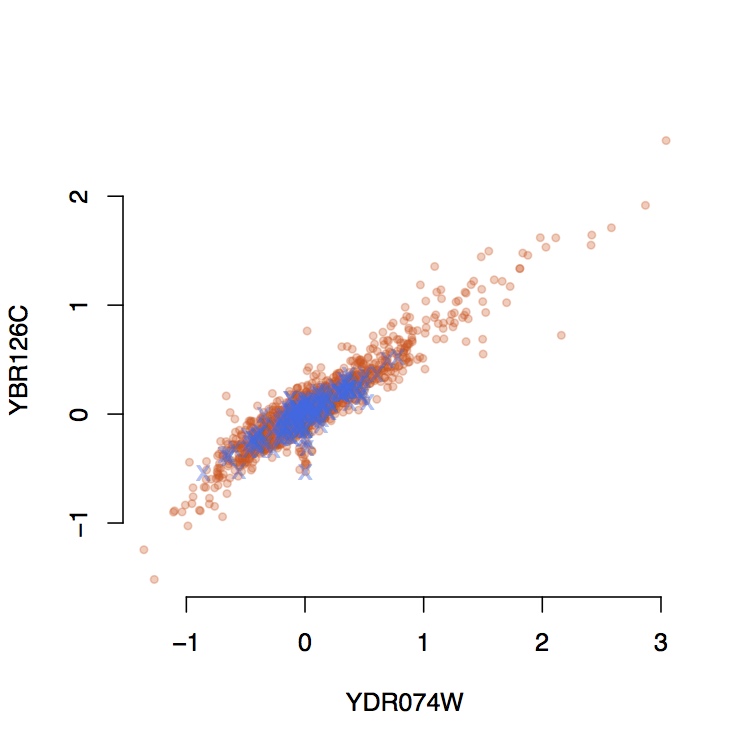}
\caption{ Some examples from the
  gene-knockout experiments in \citet{Kemmeren2014}, which
  will be discussed in more detail in
  Section~\ref{sec:geneknockout}.
   Each panel shows the distribution
of a target gene activity Y (on the respective y-axis), conditional
on a predictor gene activity X (shown on
respective x-axis). Blue crosses show observational data and
red dots show interventional data. The interventions do not occur on
any of the shown genes. The conditional distribution of $Y$, given
$X$, is not
invariant for the  examples in the first row, while invariance cannot
be rejected for the two examples in the bottom row. Take the
example of the bottom left panel.  The variance of the activity of
gene $\mathit{YMR321C}$ is clearly higher for interventional than observational
data, so we can reject that the invariance assumption holds for the
empty set $S=\emptyset$. However, if conditioning on the activity $X$ of
gene $\mathit{YPL273W}$, the conditional distribution of the activity $Y$ of gene
$\mathit{YMR321C}$ is not
significantly different between interventional and observational data,
so that the set $S=\{\mathit{YPL273W}\}$ 
fulfils the invariance
assumption~\eqref{eq:lincausal},
at least approximately.   }
\label{fig:examplesK}
\end{center}
\end{figure}

We observe each unit $i$ in only one experimental setting.
The distribution of the error $\varepsilon^\e$ is assumed to stay
identical across all environments (though see
Sections~\ref{sec:violations} and~\ref{sec:modelmismain} for approaches when this assumption is violated).
It is in general not possible to estimate the correlation between the
noise variables $\varepsilon^{e}_i$ and $\varepsilon^{f}_i$ for a
single unit~$i$ in
different hypothetical environments  $e$ and $f$, as the outcome is
observed for only one environment \citep{dawid2007counterfactuals,
  dawid2012decision}. Knowledge of the correlation would be
necessary to answer counterfactual questions about the outcome. Knowledge of the
correlation is not necessary for our method.

We deliberately avoid the term ``causality'' in Assumption~\ref{assum:invariant} in order to keep it purely mathematical.
 Proposition~\ref{propos:sem} 
 establishes a link to causality by
 showing that the parents of $Y$ in a structural equation model (SEM)
 satisfy Assumption~\ref{assum:invariant}. 
 In other words, 
 the variables that have a
 direct causal effect on $Y$ in a SEM form a set $S^*$
 for which Assumption~\ref{assum:invariant} is satisfied. 
This must not necessarily be true for the variables that have an (in)direct effect on $Y$, i.e., the ancestors of $Y$.
However, the
set $S^*$ is not necessarily unique. 
For a given set of experimental conditions $\E$, there can be multiple vectors~$\gamma^*$ that satisfy~\eqref{eq:lincausal}. For example, if only observational data are available, i.e.\ all environments are identical, 
it is apparent that for any model~\eqref{eq:lincausal} the distribution $F_{\varepsilon}$ of the residuals $\varepsilon^e$ does not depend on $e$.
If additionally $(X,Y)$ have a joint Gaussian distribution and $X$ and $Y$ are not independent, for example, then one can find 
a solution~$\gamma^*$ to~\eqref{eq:lincausal} for every subset $S^* \subseteq \{1, \ldots, p\}$.
The inference we propose works for any possible choice among the set of solutions. We can at most identify the subset of $S^*$ that is common among all possible solutions of~\eqref{eq:lincausal}, 
see Section~\ref{sec:iden} for settings with complete identifiability.

It is perhaps easiest to think about the example of a linear structural
equation model (SEM), as
defined in Section~\ref{sec:lgsem}, see also Figure~\ref{fig:examplesSEM} in Appendix~\ref{app:example}. 
We show in the following proposition
that the set of parents of $Y$ in a linear SEM is a valid set $S^*$
satisfying~\eqref{eq:lincausal}.

\begin{prop} \label{propos:sem}
Consider a linear structural equation model, as formally defined in
Section~\ref{sec:lgsem}, for the variables $(X_1=Y, X_2,\ldots, X_p, X_{p+1})$,  
with coefficients $(\beta_{jk})_{j,k=1, \ldots, p+1}$,
whose structure is given by a directed acyclic graph.
The independence assumption on the noise variables in
Section~\ref{sec:lgsem} can here be replaced by the strictly weaker
assumption that $\varepsilon^\e_{1} \independent \{ \varepsilon^\e_j; j\in
\AN{1}\}$ for all environments $e\in \E$, where $\AN{1}$ are
the ancestors of $Y$. Then Assumption~\ref{assum:invariant} holds
for the parents of $Y$, namely $S^* = \PA{1}$, and $\gamma^* =
\beta_{1,\cdot}$ as defined in Section~\ref{sec:lgsem},
under the following assumption: 
\quote{for each
  $\e \in \E$: the 
experimental setting $\e$ arises by one or several interventions on variables from $\{X_2,\ldots ,X_{p+1}\}$ but interventions on $Y$ are not
allowed;
here, we allow for do-interventions 
\citep{Pearl2009} (see also Section \ref{sec:idfirst}, and note that the assigned values can be random, too), or
soft-interventions \citep{Eberhardt2007} (see also Sections~\ref{sec:det} and~\ref{sec:idlast}).
}
\end{prop}
\begin{proof}
It follows by the definition of the interventions in~Section
\ref{sec:intervv} and because the interventions do not act on the target variable $Y$, that $
Y^e = \sum_{j \in \PA{1}} \beta_{1,j} X_j^e + \varepsilon_Y^\e $ for 
all $\e  \in \E$, where $\varepsilon_Y^\e=\varepsilon_1^\e$ is independent of $X_{\PA{1}}$
and has the same distribution for all $\e \in \E$. Thus,
Assumption~\ref{assum:invariant} holds.
\end{proof}

We remark that Proposition \ref{propos:sem} can be generalised to include some
hidden variables: the exact statement is given in Proposition
\ref{propos:semb} in Appendix~\ref{app:propsemb}.

Instead of allowing only do- or soft-interventions in
Proposition~\ref{propos:sem}, we can allow for more general interventions
which could change the structural equations for $X_2,\ldots ,X_{p+1}$
(including for example a change in the graphical structure of the model
among the variables $X_2,\ldots ,X_{p+1}$),
as long as the conditional distribution of $Y^\e$ given $X_{S^{\caus}}^\e$
remains the 
same. Such a weaker requirement is 
sometimes referred to as ``modularity'' \citep{Pearl2009} 
or
what is called ``autonomy'' \citep{Haavelmo1944, Aldrich1989}; 
structural
equations are autonomous if whenever we replace one of them due to an
intervention, all other structural equations do not change, they remain invariant. The
remaining part of the condition in Proposition~\ref{propos:sem} about
excluding interventions on the target variable $Y$ is often verifiable in many
applications; see Sections~\ref{sec:violations} and~\ref{sec:modelmismain} for violations of this assumption.

Proposition~\ref{propos:sem} refers to standard linear SEMs that do not allow for feedback cycles.
We may, however, include feedback into the SEM and consider equilibrium solutions of the new set of equations.
The independence assumption between $\varepsilon^\e$ and $X_{S^*}^\e$ allows for some feedback cycles in the linear SEM.
The independence assumption prohibits, however, cycles that include the target variable $Y$. We will leave it as an open question to what extent the approach can be generalised to more general forms of feedback models.

It is noteworthy that our inference is valid for \emph{any} set that satisfies Assumption~\ref{assum:invariant}
and not only parents in a linear SEM. 
For the
following statements we do not specify whether the set $S^*$ refers to the set of parents in
a linear SEM or any 
other set that satisfies~\eqref{eq:lincausal}, as
the confidence guarantees will be valid in either case. 
Proposition~\ref{prop:potentialoutcomes} in Appendix~\ref{app:potentialoutcomes} discusses some relationship to the potential outcome framework.

\ifgamma
\else

\subsection{Plausible causal predictors and identifiable causal predictors} \label{sec:21new}
In general, $(\gamma^*, S^*)$ is not the only pair that satisfies the assumption of invariance in~\eqref{eq:lincausal}. 
We therefore define for $\gamma\in \mathbb{R}^p$ and $S \subseteq \{1, \ldots, p\}$
the null hypothesis
$H_{0,\gamma, S} (\E)$ as
\begin{equation} \label{eq:H0gS}
H_{0,\gamma,S}(\E):
\quad  
\gamma_k = 0 \text{ if } k \notin S \;\; \text{ and }\;\;
\left\{\begin{array}{l}
\exists F_\varepsilon \mbox{
  such that for
  all } \e \in \E \\ 
  Y^\e = X^\e \gamma + \varepsilon^\e, \mbox{  where }
\varepsilon^\e \independent X^\e_{S} \mbox{ and }  \varepsilon^\e \sim
F_\varepsilon.
\end{array}
\right.
\end{equation}

As stated above, we have dropped the constant intercept notationally.
The variables that appear in \emph{any} set $S$ that satisfies 
$H_{0, S}(\E)$, we call 
plausible causal predictors. 
\begin{defin}[Plausible causal predictors and coefficients] 
$\mbox{   }$
\begin{enumerate}[(i)]
\item
We call the variables $S \subseteq \{1, \ldots, p\}$ \emph{plausible causal
predictors} under~$\E$ if the following null hypothesis holds true:
\begin{equation}\label{eq:H0S}
H_{0,S}(\E): \quad  \exists \gamma \in \mathbb{R}^p \text{ such that } H_{0,\gamma,S}(\E) \text{ is true}.
\end{equation}
\item
The \emph{identifiable causal predictors} 
 under interventions~$\E$ are defined as the following subset of plausible causal predictors
\begin{equation} \label{eq:ident}
S(\E) \; := \;  
\bigcap_{S\,:\, H_{0,S}(\E) \text{ is true }} S
\; = \; 
\bigcap_{\gamma \in \Gamma(\E)} \{k: \gamma_k \neq 0\}.
\end{equation}
\end{enumerate}
\end{defin}
Here, $\Gamma(\E)$ is defined in~\eqref{eq:Gammahatc} below (the
second equation in~\eqref{eq:ident} can be ignored for now). 
Under Assumption~\ref{assum:invariant}, $H_{0,\gamma^*, S^*}(\E)$ is true and 
therefore $S^*$ are plausible causal predictors, that is $H_{0,S^*}(\E)$ is correct, too.
The identifiable causal predictors are thus a subset of the true causal predictors, 
\begin{equation*} 
S(\E) \subseteq S^{\caus}.
\end{equation*}
This fact will
guarantee the coverage properties of the estimators we define below. Furthermore, the set of identifiable causal predictors under interventions $\E$ is growing
monotonically if we enlarge the set~$\E$,
\begin{equation*} 
S(\E_1)\subseteq S(\E_2) \quad 
\text{ for two sets of environments } \E_1, \E_2 \text{ with }
 \quad \E_1
\subseteq \E_2 .
\end{equation*}
In particular, if $|\E|=1$ (for example, there is only observational data), then 
$ S(\E) = \emptyset$ 
because $H_{0,\emptyset}(\E)$ will be true.
The set of identifiable causal predictors under a single
environment 
is thus empty
and we make no statement as to which variables are causal.

In Section~\ref{sec:iden}, we examine conditions for structural
equation models (see Proposition~\ref{propos:sem}) under which 
$S(\E)$ is identical to the parents of $Y$ we thus have complete identifiability of the causal coefficients.
In practice, the set $\E$ of experimental settings might often be such that
$S(\E)$ identifies some but not all parents of $Y$ in a SEM.

\subsection{Plausible causal coefficients}
\label{sec:22new}
We have seen that the null
hypothesis~\eqref{eq:H0gS} $H_{0,\gamma, S}(\E)$ is in general not only fulfilled for
$\gamma^{\caus}$  and
its support $S^{\caus}$
but also  potentially for other vectors $\gamma\in \mathbb{R}^p$. This
is true especially if the experimental settings $\E$ are very similar to each
other. If we consider again the extreme example of just a single environment, $|\E|=1$, and
a multivariate Gaussian distribution for $(X,Y)$, we can find for any
set $S\subseteq\{1,\ldots,p\}$ a vector $\gamma$ with support $S$ that
fulfills the null hypothesis $H_{0,\gamma, S}(\E)$, namely by using the
regression coefficient when regressing $Y$ on $X_S$. If the interventions that produce the environments $\E$
are stronger and we have more of those environments, the
set of vectors that fulfill the null becomes smaller. 
 We call vectors that fulfill the null hypothesis
plausible causal coefficients.
\begin{defin}[Plausible causal coefficients]
We define
the set
$\Gamma_S(\E)$ of
\emph{plausible causal coefficients for the set $S\subseteq\{1,\ldots,p\}$}
and the global set $\Gamma(\E)$ of \emph{plausible causal coefficients}
under~$\E$  as
\begin{align}\Gamma_S(\E) \; &:= \; \{\gamma\in \mathbb{R}^p: \; \; 
 H_{0,\gamma,S}(\E) \mbox{ is true}\}, \label{eq:GIS} \\
\label{eq:GI}
 \Gamma(\E) \; &:= \bigcup_{S\subseteq\{1,\ldots,p\}} \Gamma_S(\E).
\end{align}  
\end{defin}
Thus,
$$
\Gamma(\E_1) \supseteq \Gamma(\E_2) \quad \text{ for two sets of environments } \E_1, \E_2 \text{ with } 
\quad \E_1
\subseteq \E_2   .
$$
The global set of plausible causal coefficients~$\Gamma(\E)$ is, in other words, shrinking as we enlarge the
set~$\E$ of possible experimental settings. 

The null hypothesis $H_{0,S}(\E)$ in~\eqref{eq:H0S} can be simplified. Writing
\begin{equation} \label{eq:betapred}
\beta^{\pred,\e}(S) \;:=\; \mbox{argmin}_{\beta\in \mathbb{R}^{p}: \beta_k=0 \text{ if } k \notin S} \;  E (Y^\e-X^\e \beta)^2
\end{equation}
for the least-squares population regression coefficients when regressing the target of
interest onto the variables in $S$ in experimental setting $\e\in\E$,
we obtain the equivalent formulation of the null hypothesis for set $S\subseteq\{1,\ldots,p\}$,
\begin{align}   
H_{0,S}(\E): \quad 
\left\{
\begin{array}{l}
\exists \beta \in
  \mathbb{R}^p \text{ and } \exists F_\varepsilon \mbox{ such that for all }\e \in \E \text{ we have }
\\ 
\beta^{\pred,\e}(S)\equiv\beta\mbox{  and  }  Y^\e = X^\e \beta + \varepsilon^\e, \mbox{  where }
\varepsilon^\e \independent X^\e_{S} \mbox{ and }  \varepsilon^\e \sim
F_\varepsilon.
\end{array}
\right.
\label{eq:H0Sregr}
\end{align}
We conclude that 
\begin{equation} \label{eq:GS} 
{\Gamma}_S(\E) =\left\{ \begin{array}{cl} \emptyset & \mbox{if } H_{0,S}(\E)
    \mbox{ is false} \\ \beta^{\pred,\e}(S) &
    \mbox{otherwise}. \end{array}\right.
\end{equation}
In other words,
the set of plausible causal coefficients for a set
$S$ is either empty or contains only the population regression vector.
We will make use of this fact further below in Section~\ref{sec:conf} when computing empirical estimators.

\fi

\ifgamma
\else
\section{Estimation of identifiable causal predictors}\label{sec:conf}
We would like to estimate the set $S(\E)$ of identifiable causal predictors~\eqref{eq:ident} when
observing the distribution of $(X^\e,Y^\e)$ under different
experimental conditions $\e \in \E$. At the same time, we might be interested in obtaining confidence intervals for the linear causal
coefficients.

Recall again the definition~\eqref{eq:H0S} of the null hypothesis $H_{0,S}(\E)$.
Suppose for the moment that a statistical test for $H_{0,S}(\E)$ with size smaller than a significance level $\alpha$ is available. Then the construction of an estimator $\hat S(\E)$ and confidence sets $\hat{\Gamma}(\E)$ for the
causal coefficients can work as follows. 

\begin{mdframed}[roundcorner=10pt,frametitle=Generic method for invariant prediction]
\begin{enumerate}[1)]
\item For each set $S\subseteq\{1,\ldots,p\}$, test whether $H_{0,S}(\E)$ holds at level $\alpha$ (we will discuss later concrete examples).
\item Set $\hat{S}(\E)$ as 
\begin{equation} \label{eq:hatcausal}
  \hat{S}(\E) \; := \; \bigcap_{S: H_{0,S}(\E) \text{ not rejected}} S.
\end{equation}
\item For the confidence sets, define 
\begin{equation} \label{eq:Gammahatc} 
\hat{\Gamma}(\E) := \bigcup_{S\subseteq\{1,\ldots,p\}}
  \hat{\Gamma}_S(\E),
\end{equation}
where
\begin{equation} \label{eq:Gammahatsc} 
\hat{\Gamma}_S(\E)\; := \; \left\{ \begin{array}{cl} \emptyset & H_{0,S}(\E)
    \mbox{ can be rejected at level $\alpha$} \\ \hat{C}(S) &
    \mbox{otherwise}. \end{array}\right.  
\end{equation}
Here, $\hat{C}(S)$ is a $(1-\alpha)$-confidence set for the
regression vector $\beta^{\pred}(S)$ that is obtained by pooling the data.
\end{enumerate}
\end{mdframed}
\fi
As an example, consider again
  Figure~\ref{fig:examplesK}. Taking the example in the bottom left
  panel, we cannot reject $H_{0,S}(\E)$ for $S=\{\mathit{YPL273W}\}$. Hence we can see already from this plot that $\hat{S}(\E)$ is
  either empty or that $\hat{S}(\E)=\{\mathit{YPL273W}\}$. The
  latter case happens 
if no further set of variables is accepted that does not include the activity of
gene $\mathit{YPL273W}$ as predictor.

\ifgamma
\else
A justification for pooling the data in~\eqref{eq:Gammahatsc} is given in Section \ref{sec:pooling}.
(The construction  is also valid if the
confidence set is based only on data from a single environment, but a confidence set for the pooled data will
be smaller in general.)
This defines a whole family of estimators and confidence sets
as we have flexibility in the test we are using for the null hypothesis~\eqref{eq:H0S} and how the confidence interval $\hat{C}(S)$ is
constructed.

If the test and pooled confidence interval have the
claimed size and coverage probability, we can guarantee
coverage of 
the true causal predictors
and 
the true causal
coefficient, 
as shown below in Theorem~\ref{theo:main}.
\fi
\ifgamma
\else
\begin{theorem}\label{theo:main}
Assume that the estimator
$\hat{S}(\E)$ is
constructed according to~\eqref{eq:hatcausal} with a valid test for~$H_{0,S}(\E)$  for all sets $S
\subseteq \{1, \ldots, p\}$ at level $\alpha$ in the sense that for all $S$,
$\sup_{P\,:\,H_{0,S}(\E) \text{ true}} \; P[H_{0,S}(\E) \text{ rejected}]
\leq \nolinebreak \alpha$.
Consider now a distribution $P$ over $(Y,X)$ and consider any $\gamma^*$ and $S^*$ such that Assumption~\ref{assum:invariant} holds.  
Then, $\hat{S}(\E)$
satisfies
\begin{equation*} 
P\big[\hat{S}(\E)\subseteq S^{\caus} \big]   \ge
  \; 1- \alpha.
\end{equation*}
If, moreover,
for all $(\gamma,S)$ that satisfy Assumption~\ref{assum:invariant},
the confidence set $\hat{C}(S)$ in~\eqref{eq:Gammahatsc} satisfies
$P[\gamma \in \hat{C}(S)] \geq 1-\alpha$
then 
the set $\hat{\Gamma}(\E)$~\eqref{eq:Gammahatc} has coverage at least level
$1 - 2\alpha$:
\begin{equation*} 
P\big[\gamma^{\caus} \in \hat{\Gamma}(\E) \big]   \ge
  \; 1- 2\alpha.
\end{equation*} 
\end{theorem}

\begin{proof} 
The first property follows immediately since
\begin{align*} 
P\big[\hat{S}(\E)\subseteq S^{\caus}\big]
\; = \;
P\Big[\bigcap_{S: H_{0,S}(\E) \text{ not rejected}} S 
\subseteq S^{\caus}
\Big]
\; \ge\; 
P\big[H_{0,S^*}(\E) \text{ not rejected} \big] 
\; \ge\; 1-\alpha,
\end{align*}
where the last inequality follows by the assumption that the test for
$H_{0,S}$ is
valid at level~$\alpha$ for all sets $S\subseteq\{1,\ldots,p\}$. The
second property follows 
since 
$$
P\big[\gamma^{\caus} \notin \hat{\Gamma}(\E)\big]
\; \le \;
P\big[
H_{0,S^*}(\E) \text{ rejected  or } 
\gamma^* \notin \hat{C}(S^*)
\big]
\; \le \;
\alpha + \alpha = 2\alpha.
$$
\end{proof}
\fi
The confidence sets thus have the correct (conservative) coverage.
The estimator of the causal predictors will, with probability at
least $1-\alpha$, not erroneously
include non-causal predictors. Note that the statement is true for any
set of experimental or intervention settings. 
In the worst case, the set $\hat{S}(\E)$ might be empty but the
error control is valid nonetheless.

Since Theorem \ref{theo:main} holds for any $\gamma^*, S^*$ which fulfil
Assumption 1, and assuming the setting of Proposition \ref{propos:sem}, we
obtain the corresponding confidence statements for the causal
coefficients and causal variables in a linear structural equation model,
that is for $\gamma^* = \beta_{1,\cdot}$ and $S^* = \PA{1}$ in the notation of
Proposition \ref{propos:sem}.
\ifgamma
\else
\begin{remark}
\begin{itemize}
\item[(i)] We obtain the following empirical version of~\eqref{eq:ident}:
\begin{equation} \label{eq:hatcausal2} 
\hat{S}(\E) 
\; = \;
\bigcap_{\gamma \in \hat{\Gamma}(\E)} \{k:\gamma_k\neq 0\} 
\; = \; 
\bigcap_{S: {H}_{0,S}(\E) \text{ not rejected at }\alpha} S
\end{equation}
provided that 
if ${H}_{0,S}(\E)$ is not rejected, then
for all  $\gamma \in \hat{\Gamma}_S(\E)$ 
we have
$\mathrm{supp}(\gamma) \subseteq S$ 
and  
$H_{0,\mathrm{supp}(\gamma)}(\E)$ is not rejected either. 
\item[(ii)] In~\eqref{eq:Gammahatsc}, we have constructed confidence sets $\hat{\Gamma}_S(\E)$ based on a test for $H_{0,S}(\E)$. Alternatively, confidence sets
$\hat{\Gamma}_S(\E)$
may be available that are not based on a test procedure for
$H_{0,S}(\E)$. In this case, we may take them as a starting point and
define 
$\hat{S}(\E)$ using the first equality in~\eqref{eq:hatcausal2}, instead of~\eqref{eq:hatcausal}.
Analogously to Theorem~\ref{theo:main}, the correct coverage property of 
$\hat{\Gamma}_{S^*}(\E)$ then
implies confidence statements for $\hat{\Gamma}(\E)$ and $\hat{S}(\E)$.
\end{itemize}
\end{remark}
\fi

\subsection{Two concrete proposals}\label{sect:invariant}
\ifgamma
\else
The missing piece in the generic procedure given
by~\eqref{eq:hatcausal} and~\eqref{eq:Gammahatc} is a test for
$H_{0,S}(\E)$ that is valid at level $\alpha$ for any
given set of variables $S\subseteq\{1,\ldots,p\}$ and thus implies 
\[
P\big[  H_{0,S^{\caus}}(\E) \text{ rejected}\big]\leq \alpha.
\]
\fi
To specify a concrete procedure and derive its statistical properties, we assume
throughout the paper that the data consist of $n$ independent  
observations.
Within each experimental setting $e$,
we assume that we receive $n_{\e}$ independent and identically distributed data points 
from $(X^\e, Y^\e)$
and
thus, $\sum_{\e \in \E} n_{\e} = n$.

We now propose  a way to construct such a test, but acknowledge that different choices are possible. 
Our construction  will be based on the fact that the
causal coefficients are identical to the regression 
effects in all experimental settings $e\in \E$ if we consider only variables in the
set $S^\caus$ of causal predictors.  

For experimental setting $\e \in \E$ and a subset
$S$ of variables, define the regression coefficients
  $\beta^{\pred,\e}(S) \in \mathbb{R}^{p}$ as above in~(\ref{eq:betapred}).
Define further the population residual standard deviations when regressing $Y^\e$ on variables $X_S^\e$ as 
\begin{align*}  
\sigma^\e(S) \; &:= [E(Y^\e - X^\e \beta^{\pred,\e}(S))^2]^{1/2}.
\end{align*}
These definitions are population quantities. The corresponding
sample quantities are denoted with a hat.
As mentioned above, under Assumption~\ref{assum:invariant}, for $S=S^\caus$,
the
regression effects are identical to the causal coefficients: for all $e \in E$, 
\[ \beta^{\pred,\e}(S^\caus)\equiv \gamma^\caus \qquad \mbox{and} \qquad 
\sigma^\e(S^\caus) \equiv \var(F_{\varepsilon})^{1/2}.\]
To get a test valid at level $\alpha$ for all subsets $S$ of predictor
variables,
we first
weaken $H_{0,S}(\E)$ in~\eqref{eq:H0Sregr} to
\begin{equation}   \label{eq:H0tSregr}
\tilde H_{0,S}(\E):\quad
\exists (\beta,\sigma)\in
  \mathbb{R}^p\times \mathbb{R}_+ \mbox{
    such that }
  \beta^{\pred,\e}(S)\equiv\beta\mbox{  and  } \sigma^\e(S)
  \equiv\sigma\mbox{ for all } \e\in \E.
\end{equation}
The null hypothesis $\tilde H_{0,S}(\E)$ is true whenever the original null hypothesis~\eqref{eq:H0Sregr} is
true. 
\ifgamma
\else
As in~\eqref{eq:Gammahatsc}, we set
\fi
\begin{equation*} 
\hat{\Gamma}_S(\E) :=\left\{ \begin{array}{cl} \emptyset & \tilde  H_{0,S}(\E)
    \mbox{ can be rejected at level $\alpha$} \\ \hat{C}(S) &
    \mbox{otherwise}. \end{array}\right.
\end{equation*}

We now give a concrete example which we will use in the numerical
examples under the assumption of Gaussian errors and that the design matrix $\B{X}_\e$  of
all $n_\e$ samples in experimental setting $\e\in\E$ has full rank.
(We write the design matrix in bold letters, as opposed to the random variables $X^e$.)
The whole procedure is then a specific version of the general procedure
given further above, where we use a specific test in the first step (the second step is unchanged).
{\small 
\begin{mdframed}[roundcorner=10pt,frametitle={Method I: 
Invariant prediction using test on regression coefficients
} ] 
\begin{enumerate}[1)]
\item For each $S\subseteq\{1,\ldots,p\}$ and $\e\in \E$:
\begin{enumerate}
\item Let $I_\e$ with $n_e = |I_\e|$ be the set of observations where experimental setting
$\e\in \E$ was active. Likewise, let $I_{-\e}=\{1,\ldots,n\}\setminus
I_\e$ with $n_{-\e} := |I_{-\e}| $ be the set of observations when using only observations where
experimental setting $\e\in \E$ was \emph{not} active.   Let
$\B{X}_{\e,S}$ be the $n_\e \times (1+|S|)$-dimensional matrix when
using all samples in $I_\e$ and  all
predictor variables in $S$, adding an intercept term to the
design matrix as mentioned previously. If $S=\emptyset$, the matrix consists only of a single
intercept column. Analogously, $\B{X}_{-e,S}$ is defined
with the samples in $I_{-\e}$. 
Let  $\hat{Y}_\e$ be the predictions for observations
in set $I_\e$ when using the OLS estimator computed on samples in
$I_{-\e}$ and let $D:=Y_e - \hat Y_e$ be the difference between the actual
observations $Y_e$ on $I_\e$ and the predictions.

\item 
Under Gaussian errors, if~\eqref{eq:H0tSregr} is true for a set $S$, 
then \citep{chow1960tests}  
\begin{equation}\label{eq:chow} 
\frac{D^t  \Sigma_D^{-1} D
}{\hat{\sigma}^2 \, n_\e} \sim \; F(n_\e,n_{-\e} - |S|-1),
\end{equation}
where $\hat{\sigma}^2$ is the estimated variance on the set $I_{-\e}$
on which the OLS estimator is computed. The covariance matrix
$\Sigma_D$ is given by
\[ \Sigma_D = {1}_{n_\e} +  \B{X}_{\e,S} (\B{X}_{-\e,S}^t \B{X}_{-\e,S})^{-1} \B{X}_{\e,S}^t,\]
letting  $1_n$ be the identity
matrix in $n$-dimensions.
For any set $S$, we reject the null hypothesis $\tilde H_{0,S}(\E)$
if the $p$-value of~\eqref{eq:chow} is below 
$\alpha/|\E|$ for any $\e\in \E$. 
\end{enumerate}
\item[2)] As in the generic algorithm, using~\eqref{eq:hatcausal}.
\item[3)] If we do reject a set $S$ we set
  $\hat{\Gamma}_S(\E)=\emptyset$. Otherwise, we set
  $\hat{\Gamma}_S(\E)$ to be a 
$(1-\alpha)$-confidence
  interval for $\beta^\pred(S)$ when using all data
  simultaneously. For
  simplicity, we will use a rectangular confidence 
  region where
the constraint for $\beta^\pred(S)_k$ is 
  identically 
  0 if $k\notin S$ and for coefficients in $S$ given by 
  $(\hat{\beta}^{\pred} (S))_S \pm  t_{1-\alpha/(2|S|),n-|S|-1} \cdot \hat{\sigma} \, \mbox{diag}((\B{X}_S^t \B{X}_S)^{-1})$,
  where
$\B{X}_S$ is the design matrix of the pooled data when using variables in
$S$, $t_{1-\alpha;q}$ is the $(1-\alpha)$-quantile of a
t-distribution with $q$ degrees of freedom, and $\hat{\sigma}^2$ the
estimated residual variance. 
\end{enumerate}
\end{mdframed} }

A justification of the pooling in step 3 is given in Section
  \ref{sec:pooling}.
The
procedure above has some shortcomings. For example, the inversion of the covariance matrix
in~\eqref{eq:chow} might be too slow if we have to search many sets
and the sample size is large. One can then just work with a random
subsample of the set $I_\e$ of size, say, a few hundred, to speed up the
computation. It also depends on the assumption of Gaussian errors,
although this could be addressed by using rank tests or other nonparametric
procedures. Lastly, it is not straightforward to
extend this approach to classification and nonlinear models. 

We thus provide a second possibility.
The fast approximate version below is not fitting a model on each
experimental setting separately as in Method I, but is just fitting
one global model to all data and comparing the distribution of the
residuals in each experimental setting. This is ignoring the sampling
variability of the coefficient estimates but leads to a faster
procedure.
{\small
\begin{mdframed}[roundcorner=10pt,frametitle={Method II: 
Invariant prediction using fast(er) approximate test on residuals
}]
\begin{enumerate}[1)]
\item For each $S\subseteq\{1,\ldots,p\}$ and $\e\in \E$:
\begin{enumerate}
\item Fit a linear regression model on all
  data to get an estimate $\hat \beta^\pred(S)$ of the optimal coefficients
  using set $S$ of variables for linear prediction in regression. Let $R=Y- X \hat \beta^\pred(S)$.
\item Test the null hypothesis that the mean of $R$ is identical for each
set $I_\e$ and $\e\in \E$, using a two-sample t-test for residuals in
$I_\e$ against residuals in $I_{-\e}$ and combing via Bonferroni
correction across all $\e\in \E$.  Furthermore, test whether
the variances of $R$ are identical 
in $I_\e$ and $I_{-\e}$, using an F-test,  and combine again via Bonferroni correction for all $\e\in
\E$. Combine the two $p$-values of equal variance and equal mean by
taking twice the smaller of the two values. 
If the $p$-value for the set $S$ is smaller than 
$\alpha$,
we
  reject the set $S$.
\end{enumerate}
\item As in the generic algorithm, using~\eqref{eq:hatcausal}.
\item If we do reject a set $S$ we set
  $\hat{\Gamma}_S(\E)=\emptyset$. Otherwise, we set
  $\hat{\Gamma}_S(\E)$ to be the conventional 
  $(1-\alpha)$-confidence
  region for $\beta^\pred(S)$ when using all data
  simultaneously. For simplicity, we will use rectangular confidence
  regions, exactly as in step 3 of Method I. 
\end{enumerate}
\end{mdframed}}
Besides a computational advantage, the method can also easily be
extended to nonlinear and logistic regression models. For logistic
regression, one can test the residuals $R=Y-\hat{f}(X)$
for equal mean across the experimental settings, for example.

\subsection{Data pooling} \label{sec:pooling}
So far, we have assumed that the set $\E$ of experimental settings is
given and fixed. An experimental setting $\e\in \E$ can for example correspond to
\begin{enumerate}[(i)]
\item observational data;
\item a known intervention of a certain
type at a known variable;
\item a random intervention at an
unknown and random location;
\item  observational data in a changed environment.
\end{enumerate} 
We have used data pooling in Methods~I and~II to get confidence intervals for the regression coefficients (which is not necessary but increases power in general). 
A justification of this pooling is in order. The
joint distribution of $(X_{S^\caus}^\e,Y^\e)$ will vary in general
with $e\in \E$. 
Under
Assumption~\ref{assum:invariant}, however,
the conditional distribution  $Y^\e \given X_{S^\caus}^e$
is constant as a function of $\e\in \E$, see Section~\ref{sec:nonlin}.
As long as our tests and confidence intervals require only an
invariant conditional distribution for $S^\caus$ (which is the case
for the procedures given above), we can pool data from various $\e\in
\E$.

To make it more precise, assume there is a set of countably many
experimental settings or interventions $\I$ and $(X^\i,Y^\i)$ follow
a certain distribution $F_\i$ for each $\i \in \I$. Then each
encountered experimental setting $\e$ can be considered to be equivalent to
a probability mixture distribution over the experimental settings in $\I$, that is 
\[ F_\e = \sum_{\i\in \I} w_\i^\e F_\i,\]
where $w^\e_\i$ corresponds to the probability that an observation
under setting $\e$ follows the distribution~$F_\i$. 
We can then pool two experimental settings $\e_1$ and $\e_2$, for example, thereby
creating a new experimental setting with the averaged weights
$(w^{\e_1}+w^{\e_2})/2$. 

Pooling is a trade-off between identifiability and statistical power,
assuming that Assumption~\ref{assum:invariant} holds for the settings from
${\cal J}$.
The richer the set $\E$
of experimental settings, the smaller the set $\Gamma(\E)$ of
plausible causal coefficients will be and the larger the set of
identifiable causal predictors $S(\E)$. By pooling data, we make the set of identifiable causal
variables smaller, that is 
$ S(\E)$ is shrinking as we reduce the number $|\E|$ of different settings. 
The trade-off can either be settled a-priori (for example if we know
that we have ``sufficiently'' many observations in each known
experimental setting, we would typically not pool data) or one can try
various pooling procedures and combine all results, after adjusting
the level $\alpha$ to account for the increased multiplicity of the
associated testing problem. 
Section~\ref{sec:iden} discusses conditions on the interventions
under which all true causal effects are identifiable.

\subsection{Splitting purely observational data}
In the case of purely observational data, the null
hypothesis~\eqref{eq:H0gS} is correct for $\gamma = 0$ and $S = \emptyset$.
Therefore, $S(\E) = \emptyset$ and $\hat S(\E) = \emptyset$ with high probability, i.e.,  
our method stays conservative and does not make any causal claims.

In a reverse operation to data pooling across experiments, 
the question arises whether we can identify the causal predictors by
artificially separating data into several blocks although the data
have been generated under only one experimental setting (e.g. the data
are purely observational). 
If the distribution is generated by a SEM  (see Section~\ref{sec:lgsem}), we may
consider a 
variable $\II$ 
that is not $Y$ and known to be a 
non-descendant of the target variable $Y$, that is, there is no
directed path from $Y$ to $\II$, for example as it precedes $Y$
chronologically. 
(This is similar as in an instrumental variable setting, see Section~\ref{subsec.instrvar}.)
We may now split the data by conditioning on this
variable $\II$ or any function $h(\II)$. 
Our method
then still has the correct coverage for any function $h(\II)$ as long as
$\II$ is a non-descendant of $Y$, because the conditional distribution
of $Y$ given its true causal predictors $X_{S^{\caus}}$ does not
change and for all $z$ in the image of $h$,
\begin{equation} \label{eq:datasplittingworks}
Y\given X_{S^{\caus}} 
\quad\overset{d}{=}\quad
Y\given X_{S^{\caus}}, h(\II)=z
\end{equation}
Note that $\II$ might or might not be part of the set
$X_{S^*}$ but we expect the method to have more power if it is not. 
Equation~\eqref{eq:datasplittingworks} is a direct implication of the local Markov property that is satisfied for a SEM \citep[][Theorem~1.4.1]{Pearl2009}.
The confidence intervals remain valid but 
the implication on (partial) identifiability  of the causal predictors 
remains as an open question.

Even without data splitting, there might still be some directional information in the data set that is not exploited by our method; this may either be information in the conditional independence structure \citep{Spirtes2000, Chickering2002}, information from non-Gaussianity \citep{Shimizu2006}, nonlinearities \citep{Hoyer2008, Peters2014JMLR, Buhlmann2014annals}, 
equal error variances \citep{Peters2012} or shared information between regression function and target variable \citep{Janzing2012}. Our method does not exploit these sources of identifiability.
We believe, however, that it might be possible to 
incorporate the identifiability
based on non-Gaussianity or
nonlinearity. 

\subsection{Computational requirements} \label{sec:highdim}
 The construction of the confidence regions for the set
of plausible causal coefficients and the identifiable causal predictors
requires to go through all possible sets of variables in step 1) of
the procedures given above. The
computational complexity of the brute force scheme seems to grow super-exponentially
with the number of variables. 

There are several aspects to this issue. Firstly, we often do not have
to go through all sets of variables. If we are looking
for a non-empty set $\hat{S}(\E)$, it is worthwhile in general to
start generating the confidence regions $\hat{\Gamma}_S(\E)$ for the
empty set $S=\emptyset$, then for all singletons and so forth. If the
empty set is not rejected, we can stop the search immediately, as
then $\hat{S}(\E)=\emptyset$. If the empty set is rejected, we can
stop early as soon as we have accepted more than one set $S$ and the
sets have an empty overlap (as $\hat{S}=\emptyset$ in this case no
matter what other sets are accepted). The method can thus finish
quickly if $\hat{S}=\emptyset$.
However, in a positive case (where we do hope to get a non-empty
confidence set) we will still have to go through all sets of variables
eventually. There are two options to address the computational
complexity. 

The first option is to limit a-priori the size of the
set of causal predictors. Say we are willing to make the
assumption that the set of causal variables is at most 
$s<p$. Then we just have to search over all subsets of size at most
$s$ and incur a computational complexity that grows like $O(p^s)$ as a
function of 
the number of variables. 

A second option (which can be combined with the first one) is an adaptation of the confidence
interval defined above, in which the number of variables is first
reduced to a subset of small size that contains the causal predictors with high probability. 
Let $\hat{B}\subseteq\{1,\ldots,p\}$ be, for the pooled data, an
estimator of the variables with non-zero regression coefficient when
using all variables as predictors. For example, $\hat{B}$ could be
the set of variables with non-zero regression coefficient with
square-root Lasso estimation \citep{belloni2011square}, Lasso
\citep{tibshirani1996regression} or boosting
\citep{schapire1998boosting,fried01,buhlmann2003boosting}
with cross-validated penalty parameter. If the initial 
screening is chosen such that the causal predictors are contained with
high probability,
$ P\big[ S^{\caus} \subseteq
  \hat{B}\big]\ge 1-\alpha,$
and we construct the confidence set $\hat{S}(\E)$ as above,
but just letting $S$ be a subset of $\hat{B}$ instead of
$\{1,\ldots,p\}$, it will have coverage at least $1-2\alpha$.
Sufficient assumptions of such a coverage (or screening) condition are
discussed in the literature \citep[e.g.][]{Buhlmann2011book}.
If the second option is combined with the first option, the computational complexity
would then scale like $O(q^s)$ instead of $O(p^s)$, where $q$ is the
maximal size of the set $\hat{B}$ of selected variables. 
For the sake of simplicity, we will not develop this argument further here
but rather focus on the  identifiability  results
for the low(er)-dimensional case. 

\section{Identifiability results for structural equation models} \label{sec:iden}
The question arises whether the proposed confidence sets for
the causal predictors can recover an assumed true set of causal
predictors. 
  Such identifiability issues are
  discussed next. 
Sections~\ref{sec:lgsem} and~\ref{sec:intervv} describe possible data generating mechanisms and Section~\ref{sec:idres} provides corresponding identifiability results.

\subsection{Linear Gaussian SEMs} \label{sec:lgsem}
We consider linear Gaussian structural equation models (SEMs)
\citep[e.g.][]{Wright1921, Duncan1975}.
We assume that each element $\e \in
\E$ represents a different interventional setup.  
Let the first block of data ($e=1$) always correspond to an ``observational'' (linear) Gaussian SEM. Here, a distribution over $(X^1_1, \ldots, X^1_{p+1})$ is said to be generated from a Gaussian SEM if
\begin{equation} \label{eq:semmmmm}
X_j^1 = \sum_{k \neq j} \beta_{j,k}^1 X^1_k + \varepsilon^1_j, \qquad j = 1, \ldots, p+1,
\end{equation}
with $\varepsilon_j^1 \iids \mathcal{N}(0,\sigma_j^2)$, $j=1, \ldots, p+1$. 
The corresponding directed graph is obtained by drawing arrows from
variables $X^1_k$ on the right-hand side of~\eqref{eq:semmmmm} with
$\beta_{jk}^1 \neq 0$ to the variables $X^1_j$ of the left-hand
side. This graph is assumed to be acyclic.
Without loss of generality
let us assume that $Y^1 := X^1_{1}$ is the target variable and we write
$X := (X_2, \ldots, X_{p+1})$.
 We further assume that all variables are
observed; this assumption can be weakened, see
Proposition~\ref{propos:semb} in Appendix~\ref{app:propsemb} and Section~\ref{subsec.instrvar}.

The parents of $Y$ are given by 
$$
\PA{Y} = \PA{1} = \{k \in \{2, \ldots, {p+1}\}\,:\, \beta_{1, k}^1 \neq 0\}.
$$
Here, we adapt the usual notation of graphical models
\citep[e.g.][]{Lauritzen1996}. For example, we write $\PA{j}$,
$\DE{j}$, $\AN{j}$
and $\ND{j}$ for the parents, descendants, ancestors and non-descendants of $X_j$,
respectively.  

Let us assume that the other data blocks are generated by a linear SEM, too: 
\begin{equation} \label{eq:sem}
X_j^\e = \sum_{k \neq j} \beta_{j,k}^\e X^\e_k + \varepsilon^\e_j, \qquad j = 1, \ldots, p+1,\quad e \in \E.
\end{equation}
Assumption~\ref{assum:invariant} states that the influence of
the causal predictors remains the same under interventions, that is 
$Y^\e = X^\e \gamma^{\caus} + \varepsilon^1_{1}$ for $\gamma^{\caus} =
(\beta_{1,2}^1, \ldots, \beta_{1, p+1}^1)^t $ 
  and $\varepsilon^\e_{1}\overset{d}{=} \varepsilon^1_{1}$ for $\e \in
\E$. The other coefficients $\beta_{j,k}^\e$ and noise variables
$\varepsilon^\e_j$, $j \neq 1$, however, 
may be different from the ones in the observational setting~\eqref{eq:semmmmm}.
Within this setting, we now define various sorts of interventions. 

\subsection{Interventions} \label{sec:intervv}
We next  discuss three different types of interventions that
all lead to identifiability of the causal predictors for the target variable. 
\subsubsection{Do-interventions} \label{sec:idfirst}
These types of interventions correspond to the classical do-operation from
\citet[e.g.]{Pearl2009}. 
In the $e$-th experiment, we intervene on variables $\C{A}^\e \subseteq \{2, \ldots, p+1\}$ and set them to values $a^\e_j \in \mathbb{R}$, $j \in \C{A}^\e$.
For the observational setting $e=1$, we have $\C{A}^1 = \emptyset$. We
specify the model~\eqref{eq:sem}, for $\e \neq 1$, as follows:
$$
\beta_{j,k}^\e = 
\left\{
\begin{array}{cl}
\beta_{j,k}^1 & \text{ if } j \notin \C{A}^\e\\
0 & \text{ if } j \in \C{A}^\e,
\end{array}\right.
$$
and
$$
\varepsilon_{j}^\e \equald 
\left\{
\begin{array}{cl}
\varepsilon_{j}^1 & \text{ if } j \notin \C{A}^\e\\
a^\e_j & \text{ if } j \in \C{A}^\e.
\end{array}\right.
$$
The do-interventions correspond to fixing the intervened variable at a
specific value.  The following two types of interventions consider
``softer'' forms of interventions which might be more realistic for
certain applications.

\subsubsection{Noise interventions} \label{sec:det}
Instead of fixing the intervened variable at a specific value,
noise interventions correspond to ``disturbing'' the variable by
changing the distribution of the noise 
variable. 
This is an instance of what is sometimes called a
``soft intervention'' \citep[e.g.][]{Eberhardt2007}. 
We now consider a kind of soft intervention, in which we scale the noise distributions of variables $\C{A}^\e \subseteq \{2, \ldots, p+1\}$ by a factor $A^\e_j$, $j \in \C{A}^\e$. 
Alternatively, we may also shift the error distribution by a variable $C^e_j$.
More precisely, we specify the model in~\eqref{eq:sem}, for $\e \neq 1$, as follows:
$$
\beta_{j,k}^\e = \beta_{j,k}^1 \quad \text{for all }j,
$$
and
$$
\varepsilon_{j}^\e \equald 
\left\{
\begin{array}{cl}
\varepsilon_{j}^1 & \text{ if } j \notin \C{A}^\e\\
A^\e_j \varepsilon_{j}^1 & \text{ if } j \in \C{A}^\e,
\end{array}\right.  
\quad \text{ or } \quad 
\varepsilon_{j}^\e \equald 
\left\{
\begin{array}{cl}
\varepsilon_{j}^1 & \text{ if } j \notin \C{A}^\e\\
\varepsilon_{j}^1 + C^\e_j& \text{ if } j \in \C{A}^\e.
\end{array}\right.
$$
The factors $A^\e_j$ and the shifts $C^e_j$ are considered as random but
may be constant with probability one. They are assumed to be independent of
each other and independent of all other random variables considered in the
model except for $X_k^e$ for $k \in \DE{j}$.
\subsubsection{Simultaneous noise interventions} \label{sec:idlast}

The noise interventions above operate on clearly defined variables
$\C{A}^\e$ which can vary between different experimental settings
$\e\in \E$. In
some applications, it might be difficult to change or influence the
noise distribution at a single variable but instead one could imagine
interventions that change the noise distributions at many variables
simultaneously. 
As a third example, we thus consider a special case of the preceding Section~\ref{sec:det}, in which
we pool all interventional experiments into a single data set. That is, $|\E| = 2$ and, for all $j \in \{2, \ldots, p+1\}$, 
\begin{equation} \label{eq:simnoise}
\beta_{j,k}^{e=2} = 
\beta_{j,k}^{e=1}
\end{equation}
and
$$
\varepsilon_{j}^{e=2} \equald A_j \varepsilon_{j}^{e=1}
\quad \text{ or } \quad
\varepsilon_{j}^{e=2} \equald \varepsilon_{j}^{e=1} + C_j
.
$$
The random variables $A_j \geq 0$ are assumed to have a distribution that
  is absolutely continuous w.r.t.\ Lebesgue measure with $\mean A_j^2 < \infty$
  and to be independent of all other variables and among themselves.
The pooling can either happen explicitly or, as stated above, as we
cannot control the target of the interventions precisely and a given
change in environment might lead to changes in the error distributions
in many variables simultaneously. As an example we mention gene knock-out
experiments with off-target effects in biology 
\citep[e.g.][]{Jackson2003, kulkarni2006evidence}. 

\subsection{Identifiability results} \label{sec:idres}
The following Theorem~\ref{prop:1} gives sufficient conditions for identifiability of the causal
predictors. We then discuss some conditions under which the assumptions can
or cannot be relaxed further below. Proofs can be found in  
Appendix~\ref{app:proofs}. 
\begin{theorem} \label{prop:1}
Consider a (linear) Gaussian SEM as in~\eqref{eq:semmmmm} and~\eqref{eq:sem} with interventions. 
Then, with $S(\E)$ as in~\eqref{eq:ident}, all causal predictors are identifiable, that is
\begin{equation} \label{eq:mapro}
S(\E) = \PA{Y} = \PA{1}
\end{equation}
if one of the following three assumptions is satisfied:
\begin{enumerate}[i)]
\item The interventions are {\bf do-interventions}
  (Section~\ref{sec:idfirst}) with $a^\e_j \neq \mean ( X^1_j)$ and there is at least one single intervention on each variable other than $Y$, that is for each $j \in \{2, \ldots, p+1\}$ there is an experiment $e$ with $\C{A}^\e = \{j\}$.
\item The interventions are {\bf noise interventions}
  (Section~\ref{sec:det}) with $1 \neq \mean (A^\e_j)^2 < \infty$, and again, there is at least one single intervention on
  each variable other than $Y$. 
  If the interventions act additively rather than multiplicatively, we require 
$\mean C^\e_j \neq 0$ or $0 < \var \, C^\e_j < \infty$.
\item The interventions are {\bf simultaneous noise interventions}
  (Section~\ref{sec:idlast}). 
This result still holds if we allow changing linear coefficients
$\beta_{j,k}^{e=2} \neq \beta_{j,k}^{e=1}$
in~\eqref{eq:simnoise} with (possibly random) coefficients $\beta_{j,k}^{e=2}$. 
\end{enumerate}
The statements remain correct if we replace the null hypothesis~\eqref{eq:H0Sregr} with its weaker version~\eqref{eq:H0tSregr}.
\end{theorem}

These are examples for sufficient conditions for identifiability but
there may be many more. For example, one may also consider random
coefficients or changing graph structures (only the parents of $Y$ must
remain the same). 

\paragraph{Remark.}
In general, the conditions given above are not necessary. The following
remarks, however, provide two specific counter examples that show the necessity of some conditions.
\begin{enumerate}[i)]
\item
We cannot remove the condition $a^\e_j \neq \mean ( X^1_j)$ from
Theorem~\ref{prop:1} i): the following SEMs correspond to
observational data in experiment $e=1$, interventional data with $do(X_2 = 0)$ in experiment $e=2$,  and interventional data with  $do(X_3 = 0)$ in experiment $e=3$:
\begin{alignat*}{4}
e&=1:& \qquad Y^1&= X_2^1 + X_3^1 + \varepsilon_Y,\qquad&X_2^1 &= \varepsilon_2, \qquad&X_3^1 &= -X_2^1 + \varepsilon_3,\\
e&=2:& \qquad Y^2&= X_2^2 + X_3^2 + \varepsilon_Y,\qquad&X_2^2 &= 0,    \qquad& X_3^2 &= -X_2^2 + \varepsilon_3, \\
e&=3:& \qquad Y^3&= X_2^3 + X_3^3 + \varepsilon_Y,\qquad&X_2^3 &= \varepsilon_2, \qquad&X_3^3 &= 0,
\end{alignat*}
with $\varepsilon_2$ and $\varepsilon_3$ having the same distribution.
Then, we cannot identify the correct set of parents $S^{\caus} = \{1, 2\}$.
The reason is that even $S = \emptyset$ leads to a correct null hypothesis~\eqref{eq:H0Sregr}.
\item 
If we only check the null hypothesis~\eqref{eq:H0tSregr}
instead of the stronger version~\eqref{eq:H0Sregr} 
(namely whether the residuals have the same variance rather than the same distribution), 
the condition $\mean (A^\e_j)^2 \neq 1$ is essential.
Consider a two-dimensional
observational distribution from experiment $e=1$ and
an intervention distribution from experiment $e=2$:
\begin{alignat*}{3}
e=1:& \qquad &X^1 &= \varepsilon_X, \qquad &
Y^1 &= X^1 + \varepsilon_Y ,\\
e=2:& \qquad &X^2&= A \cdot \varepsilon_X, \qquad &
Y^2 &= X^2 + \varepsilon_Y, 
\end{alignat*}
with $\mean (A)^2 = 1$ and $\varepsilon_X, \varepsilon_Y \iids \mathcal{N}(0,1)$.
Then we cannot identify the correct set of parents $\PA{Y} = \{X\}$ because again $S=\emptyset$ leads to the same residual variance and therefore a correct null hypothesis~\eqref{eq:H0tSregr}.
If we use hypothesis~\eqref{eq:H0Sregr}, however, condition $\mean (A^\e_j)^2 \neq 1$ can be weakened (if densities exist), see the proof of Theorem~\ref{prop:1} (iii). 
\end{enumerate}
In practice, we expect stronger identifiability results than
Theorem~\ref{prop:1}. Intuitively, intervening on (some of) the ancestors of $Y$ should 
be sufficient for identifiability in many cases.
Note that the two counter-examples above are non-generic in the way
that they violate faithfulness \citep[e.g.][]{Spirtes2000}.
The following theorem shows for some graph structures (which need not
to be known) that even one interventional setting with an intervention
on a single node 
may be sufficient, as long as the data generating model is chosen ``generically'' (see Appendix~\ref{app:example} for an example).
\begin{theorem}\label{prop:onlyone}
Assume 
a linear Gaussian SEM as in~\eqref{eq:semmmmm} and~\eqref{eq:sem} with all non-zero parameters 
drawn from a joint
density w.r.t.\ Lebesgue measure. Let $X_{k_0}$ be a youngest
parent of target variable $Y= 
X_{1}$, that is there is no directed path from $X_{k_0}$ to any other
parent of $Y$. Assume further that there is an edge from any
other parent of $Y$ to $X_{k_0}$.  
Assume that there is only one intervention setting, where the intervention
took place on $X_{k_0}$, that is $|\mathcal{E}| = 2$ and $\mathcal{A}^{e=2}
= \{k_0\}$ ($k_0$ does not need to be known). 

Then, with probability one, all causal predictors are identifiable, that is
\begin{equation*} 
S(\E) = \PA{Y} = \pa{1}  
\end{equation*}
if one of the following two assumptions is satisfied:
\begin{enumerate}[i)]
\item The intervention is a {\bf do-intervention} (Section~\ref{sec:idfirst}) with $a^{e=2}_{k_0} \neq \mean X^1_{k_0}$.
\item The intervention is a {\bf noise intervention} (Section~\ref{sec:det}) with $1 \neq \mean (A^{e=2}_{k_0})^2 < \infty$ or $\mean C^{e=2}_{k_0} \neq 0$, respectively.
\end{enumerate}
\end{theorem}

It is, of course, also sufficient for identifiability if the
interventional setting $\mathcal{A}^{e=2}
= \{k_0\}$ is just a member of a larger number of interventional settings.
We anticipate that more identifiability results of similar type can be
derived in specific settings. Theorem~\ref{prop:onlyone} shows that interving on the youngest
parent can reveal the whole set of parents of the target variable so
this intervention is in a sense the most informative intervention
under the made assumptions. 
Intervening on
descendants of $Y$ will, in contrast, only rule out these variables as parents of
$Y$. Some interventions are also completely non-informative; intervening on a
variable that is independent of all other variables (including the target) will, for example,
not help with identification of the set of parents of the target
variable.

\section{Instrumental and hidden variables with confounding}\label{subsec.instrvar}

We now discuss an extension of the invariance idea that is suitable 
in the presence of hidden variables. 
Instrumental variables can sometimes be used when the causal
relationship of interest is confounded and there are no randomised
experiments available
\citep{Wright1928,bowden1990instrumental,angrist1996identification,didelez2010assumptions}.
For simplicity,
let us assume that $I$ is binary. 
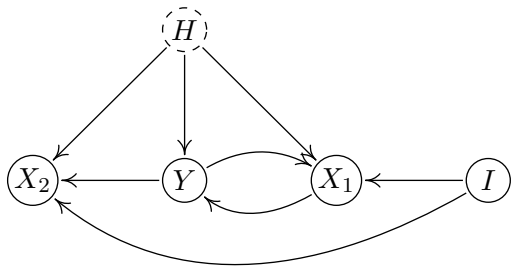
\begin{figure}[ht]
\begin{center}
\begin{tikzpicture}[scale=1, line width=0.5pt, minimum size=0.58cm, inner sep=0.3mm, shorten >=1pt, shorten <=1pt]
    \normalsize
    \draw (6,0) node(i) [circle, draw] {$I$};
    \draw (0,0) node(x2) [circle, draw] {$X_2$};
    \draw (2,0) node(y) [circle, draw] {$Y$};
    \draw (4,0) node(x1) [circle, draw] {$X_1$};
    \draw (2,2) node(w) [circle, dashed, draw] {$H$};
    \draw [-arcsq] (i) to (x1);
    \draw [-arcsq] (i) to [ out=210,in=320] (x2);
    \draw [-arcsq] (w) to  (x2);
    \draw [-arcsq] (w) to (y);
    \draw [-arcsq] (w) to (x1);
    \draw [-arcsq] (y) to (x2);
    \draw [-arcsq] (y) to [out=30,in=150] (x1);
    \draw [-arcsq] (x1) to [out=210,in=320] (y);
   \end{tikzpicture}
\end{center}
\caption{In this example of a graph of model that satisfies~\eqref{eq:ivfeedback}, variable $Y$ has a
  direct causal 
  effect only on $X_2$,  while there is a feedback between $Y$ and
  $X_1$. 
 }
\label{fig:iv}
\end{figure}
We assume that the SEM for a $p$-dimensional predictor $X$, a univariate target variable $Y$ of
interest and a $q$-dimensional hidden variable $H$ can be written as 
\begin{align}
 X &= f(I, H, Y, \eta),  \nonumber \\
Y & = X \gamma^* + g(H,\varepsilon), \label{eq:ivfeedback}
\end{align}
where $\gamma^*$ is the unknown vector of causal coefficients, $f,g$ are unknown real-valued functions and $\eta$ and $\varepsilon$ are
random noise variables in $p$ dimensions and one dimension
respectively. As it is commonly done for SEMs, we require the noise
variables $H, \eta, \varepsilon, I$ to be jointly independent. 
Figure~\ref{fig:iv} shows an example of a SEM that satisfies~\eqref{eq:ivfeedback}.

Again, we are interested in the causal coefficient $\gamma^*$. Because of the
hidden variable~$H$, 
however, regressing $Y$ on $X$ does not yield a consistent estimator
for $\gamma^*$. 

Two remarks on the model~\eqref{eq:ivfeedback} are in place. First,
the model requires that $I$ has no direct 
effect on $Y$, which is standard assumption for instrumental variable models. For a discussion on why a violation of this
assumption usually leads to no false conclusions (only a reduction in
power), see Section~\ref{sec:modelmismain}.
 Second, the model~\eqref{eq:ivfeedback} allows for feedback between $X$
  and $Y$, that is the corresponding graph in a SEM is not required
  to be acyclic. If feedback exists, the solutions are typically understood
  to be stable equilibrium solutions of~\eqref{eq:ivfeedback}  but we
  will here only require that the solutions satisfy
  equations~\eqref{eq:ivfeedback}. 

We can use $I$ as an instrument in a classical sense and  estimate~$\gamma^*$ by
the following well-known two-stage least squares procedure
\citep{angrist1996identification}: first we estimate
the influence of $I$ on $X$ and then we regress $Y$ on the predicted
values of~$X$ given~$I$. For non-linear models one can use two-stage
predictor substitution or two-stage residual inclusion; see
\citep{terza2008two} for an overview. If we strive for identification
of
$\gamma^*$, three limitations with this approach are:
\begin{enumerate}[(i)]
\item The target $Y$ is
not allowed to be a parent of any component of $X$, i.e.,
  $f(I,H,Y,\eta)=f(I,H,\eta)$. This also excludes the possibility of
  feedback between $X$ and $Y$.
\item The conditional expectation $E(X \given I)$ is not allowed to be constant for $I\in
  \{0,1\}$. 
\item The predictor $X$ has to be univariate for a
  univariate instrument $I$, that is $p=1$ is required.
\end{enumerate}

What happens if we interpret the two different values of $I$ as two
experimental settings? In other words: what happens if $I$ plays the
role of the indicator of environment (that we call $E$ at the end of
Section~\ref{sec:nonlin}) and we apply the method described above? 
We can define  $\E$ as two distinct environments by collecting all samples with $I=0$
in the first
environment and  all samples with $I=1$ in the second environment.  Of
course, another split into distinct environments is also possible and
allowed 
as long as the split into distinct environments is not a
function of $Y$, a descendant of $Y$ or the hidden variables
$H$.

We  stated in Proposition~\ref{propos:sem}  that
SEMs (with interventions) satisfy the assumptions of invariant
predictions if there are no hidden variables between the target
variable and the causal predictors. Because here there is the hidden
variable $H$ we cannot justify our method using
Proposition~\ref{propos:sem} (nor with Proposition~\ref{propos:semb} in
general). However, the invariant prediction procedure~\eqref{eq:lincausal} 
can be extended to cover models of the form~\eqref{eq:ivfeedback} as 
these models fulfil
\begin{align}\label{invariance-nonlin-hidden}
\mbox{for all } \e \in \E:\quad X^e &\mbox{ has an  arbitrary
  distribution } \nonumber   \\ 
Y^\e &=  X^\e \gamma^*+ g( H^\e,\varepsilon^\e), 
\end{align}
with unknown causal coefficients $\gamma^*\in
\mathbb{R}^p$ and unknown function $g:\mathbb{R}^q\times \mathbb{R} \rightarrow \mathbb{R}$.  


In the absence of hidden variables, 
the residuals
$Y^\e - X^\e\gamma^*$ 
are independent of the causal predictors 
$X^\e_{S^*} = X^\e_{\mathrm{supp}(\gamma^*)}$
and have the same distribution across all environments. 
In the presence of hidden variables, we cannot require
independence of the residuals and the causal predictors $X_{S*}$ but can
 adapt the null hypothesis $H_{0,S}$ in~\eqref{eq:H0S}
to the weaker form
\begin{align}
H_{0,S,\mathit{hidden}}(\E): \quad & \exists \gamma \in \mathbb{R}^p \text{ such
  that }\gamma_k = 0 \text{ if } k \notin S \mbox{  and } \nonumber \\ & \mbox{ the
  distribution of } Y^\e - X^e \gamma \mbox{  is identical for all  }\e\in \E .\label{eq:H0_hidden}
\end{align} 
Testing the null hypothesis~\eqref{eq:H0_hidden}
  is computationally more challenging than for the corresponding null hypothesis in the absence of
  hidden confounders~\eqref{eq:H0S}. In contrast to~\eqref{eq:H0S}, we
  cannot attempt to find for a given set $S$ the vector $\gamma$ by regressing $Y^e$ on $X^e$. The
  reason is that even if~\eqref{eq:H0_hidden} holds, it does not require the residuals $Y^\e - X^e \gamma$ to be independent of $X_{\mathrm{supp}(\gamma)}^e$. 
  
Suppose nevertheless that we have a test for the null hypothesis
${H}_{0,S,\mathit{hidden}}(\E) $  and define in analogy
to~\eqref{eq:hatcausal} the estimated set of causal predictors as 
\begin{equation}\label{eq:hatScausalIV} \hat{S}(\E) = \bigcap_{S: H_{0,S,\mathit{hidden}}(\E) \text{ not
    rejected}} S .\end{equation}
Then the coverage property follows immediately in the following sense.
\begin{prop} \label{prop:IV}
Consider model~\eqref{eq:ivfeedback} and let $S^*=\{k:\gamma^*_k\neq 0\}$. Suppose the test for ${H}_{0,S,\mathit{hidden}}(\E) $ is conducted at
level $\alpha$ and $\hat{S}$ is defined as
in~\eqref{eq:hatScausalIV}. Then
\[ P[\hat{S}(\E)\subseteq S^*]\;\ge\; 1-\alpha .\]
\end{prop}
\begin{proof}
The hypothesis
  $H_{0,S,\mathit{hidden}}(\E)$ is obviously true for $S^*$ as  $Y^\e-X^\e\gamma^*=
  g(H^\e,\varepsilon^\e)$ and the distribution of
  $g(H^\e,\varepsilon^\e)$ is invariant across the environments $\e\in
  \E$ (defined by $I$) as $I$
  is independent of $H$ and $\varepsilon$. 
\end{proof}
 
The method has thus guaranteed coverage  for model~\eqref{eq:ivfeedback}
even if the necessary assumptions (i)-(iii) for identification
under a two-stage
instrumental-variable approach are violated.
The power of the procedure  depends again on the type of interventions, the function class 
 and the chosen test for the null hypothesis.
We can  ask for specific examples whether $\hat{S}(\E)=S^*$ in the
population limit. 
\begin{prop}\label{prop:popIV}
Assume as a special case of~\eqref{eq:ivfeedback}  a
shift in the variance of $X$ under $I=1$ compared to $I=0$
observations:
\begin{align}
 X &= f(H, \eta) +  Z\cdot 1_{I=1}  \nonumber \\
Y & = X \gamma^* + g(H,\varepsilon), \label{eq:ivfeedbackspecial}
\end{align}
where the $p$-dimensional mean-zero random variable $Z$  is independent of
$H,\varepsilon,\eta$ and $I$
  and has a full-rank covariance
matrix. Then $\gamma^*$ and $S^*$ are identifiable in a population
sense. Specifically, if the test of ${H}_{0,S,\mathit{hidden}}(\E) $
has power 1 against any alternative, then 
\[ P[\hat{S}(\E) = S^*] \ge 1-\alpha .\]
\end{prop}
A proof is given in Appendix \ref{app:popIV}.
Note that the causal variables and coefficients can be identified for
\eqref{eq:ivfeedbackspecial}, even though the model violates the 
above-mentioned assumptions (ii) and (iii) for identifiability with a classical two-stage
instrumental variable analysis: $X$ can be of arbitrary
dimension even though the instrumental variable $I$ is univariate  and
there is no shift in $E(X \given I)$ between $I=1$ and $I=0$.

A further advantage of the invariance approach might be that no test
for a weak influence of $I$ on $X$ is necessary. A weak instrument can lead to amplification of
biases in conventional instrumental variable regression
\citep{hernan2006instruments}. With the invariance approach, the
confidence intervals for $\gamma^*$ are naturally wide in case of a weak influence of
$I$ on $X$, leading to small sets $\hat{S}$ of selected causal variables. 

 Ignoring the computational
difficulties, this shows that
the approach can be generalised to include hidden variables that violate assumption (ii) c) in Proposition~\ref{propos:semb}, for example by
replacing~\eqref{eq:H0S} with the null hypothesis~\eqref{eq:H0_hidden}.
As a possible implementation of the general approach we 
must therefore test~\eqref{eq:H0_hidden} for every set
$S\subseteq\{1,\ldots,p\}$. We are faced with a formidable
computational challenge because the coefficients $\gamma^*$ cannot be found by simple linear regression anymore. 
One possibility is to place a stricter constraint on the form of
allowed interventions. For shifted soft interventions from
Section~\ref{sec:idlast}, for example, such an approach is described in
\citet{rothenhausler2015backshift}. 
For general interventions, we can test~\eqref{eq:H0_hidden}  in a
brute-force way by testing the invariance of the
distribution over a grid of $\gamma$-values. However, the computational
complexity of this approach is exponential in the predictor dimension and it would be valuable to
identify computationally more efficient ways of testing the null
hypothesis~\eqref{eq:H0_hidden}.

\section{Further extensions and model misspecification} \label{sec:furtherext}

\subsection{Nonlinear models} \label{sec:nonlin}

We have shown an approach to obtain confidence intervals for the
causal coefficients in linear models. 
We might be interested in identifying the set of causal
predictors $S^*$ in the more general nonlinear setting~\eqref{invariance-nonlin}. The equivalent
null-hypothesis to~\eqref{eq:H0S} is then 
\begin{align}
\label{eq:H0S-nonlin}
  H_{0,S,\mathit{nonlin}}(\E):\quad   
\begin{array}{rl}  
\quad  & \mbox{There exists } g:\mathbb{R}^{|S|}\times \mathbb{R} \rightarrow \mathbb{R} \text{ and } \varepsilon^\e \text{ such that}  \\  
&  Y^\e = g(X^\e_{S},  \varepsilon^\e), \quad \varepsilon^\e
\sim F_{\varepsilon} \mbox{  and } \varepsilon^\e \independent X^\e_{S} \mbox{ for all } \e \in \E.
  \end{array}
  \end{align} 
It is interesting to note that $S$ satisfies~\eqref{eq:H0S-nonlin} if and only if it satisfies
\begin{align}\label{eq:H0S-nonlin-cond}
  H_{0,S,\mathit{nonlin}}(\E):\quad   
\begin{array}{l}  
  \forall \e, f \in\E \mbox{ the conditional distributions } Y^\e \given X_S^\e = x \text{ and } Y^f \given X_S^f = x\\
  \text{are identical for all } x \text{ such that both cond.\ distr.\  are well-defined}.
  \end{array}
  \end{align} 
The ``only if'' part is immediate and 
for the ``if'' part 
we can use a similar idea as in 
\citep[][Prop.~9]{Peters2014JMLR},
for example, and choose a Uniform$([0,1])$-distributed 
$\varepsilon$ and $g(a,b)= 
g^e(a,b):= F^{-1}_{Y^e \given X_S^e = a}(b)$, where $F_{Y^e\given X_S^e = a}$ is the cdf of $Y^e\given X_S^e = a$.

As in the linear case, 
we can consider a SEM with environments corresponding to different interventions and, 
again, the parents of $Y$ satisfy the null hypothesis.
 More precisely, we have the following remark.
 \begin{remark} \label{rem:extpropos:sem}
Proposition~\ref{propos:sem} and Proposition~\ref{propos:semb} still hold
if we replace   
linear SEMs~\eqref{eq:semmmmm} with
nonlinear SEMs 
$$
Y_j = f_j(X_{\PA{j}}, \varepsilon_j), \quad j = 1, \ldots, p + 1
$$
and replace Assumption~\ref{assum:invariant} with 
the assumption that there exists $S^*$ satisfying~\eqref{eq:H0S-nonlin}.
\end{remark}
\begin{proof}
Again, the proof is immediate. Only the case with hidden variables requires an argument. From the SEM, we are given $Y^e = f(X_{S^0}^e, X^e_{S_H^0}, \tilde \varepsilon^e)$ with $S_H^0$ being the hidden parents of $Y$ and $(X^e_{S_H^0}, \tilde \varepsilon^e) \independent X_{S^0}^e$. We can then write $Y^e = g(X_{S^0}, \varepsilon^e)$ for a uniformly distributed $\varepsilon^e$ that is independent of $X_{S^0}$ and $g(x,n) := F^{-1}_{f(x,X^e_{S_H^0}, \tilde \varepsilon^e)}(n)$. The function $g$ does not depend on $e$ because $X^e_{S_H^0}$ and $\tilde \varepsilon^e$ have the same distribution for all $e \in \E$.
\end{proof}

Assume we have a test for the null hypothesis $
H_{0,S,\mathit{nonlin}}(\E)$. Then, testing all possible sets
$S\subseteq\{1,\ldots,p\}$, we can get a confidence set for $S^*$ in a
similar way as in the linear setting~\eqref{eq:hatcausal2} by 
\begin{equation} \label{eq:hatcausal-nonlin}
\hat{S}(\E) \; := \; 
\bigcap_{S: H_{0,S,\mathit{nonlin}}(\E) \text{ not rejected}} S.
\end{equation}
If all tests are conducted individually at level $\alpha$, we have
again the property that for any $S^*$ which
  fulfills~\eqref{eq:H0S-nonlin} or \eqref{eq:H0S-nonlin-cond}, $P( \hat{S}(\E) \subseteq S^*) \ge
1-\alpha$ since the null hypothesis for $S^*$ will be accepted with
probability at least $1-\alpha$.

Constructing suitable tests for~\eqref{eq:H0S-nonlin-cond} is easier if we are  willing to assume that the function $g$ in~\eqref{eq:H0S-nonlin} is additive in the noise component, that is 
\begin{equation} \label{eq:nonlincausal} 
  H_{0,S,\mathit{additive}}(\E):\quad   
\begin{array}{rl}  
\quad  & \mbox{there exists } g:\mathbb{R}^{|S|} \rightarrow \mathbb{R} \text{ and } \varepsilon^\e \text{ such that}  \\  
&  Y^\e = g(X^\e_{S}) +  \varepsilon^\e, \quad \varepsilon^\e
\sim F_{\varepsilon} \mbox{  and } \varepsilon^\e \independent X^\e_{S} \mbox{ for all } \e \in \E.
  \end{array}
\end{equation}
Then, we can construct tests for the null hypothesis~\eqref{eq:H0S-nonlin}
that are similar as in the linear case. Analogously to Method~I in
  Section~\ref{sect:invariant}, we
can perform nonlinear regression in each environment and test whether
the regression functions are identical \citep[e.g.][for 
isotonic regression functions]{durot2013testing}. As an alternative, we can also fit
a regression model on the pooled data set and test whether the residuals have the
same distribution in each environment, see Method~II in
  Section~\ref{sect:invariant}.

We may also test~\eqref{eq:H0S-nonlin-cond}
without assuming additivity of the noise component. This could be
addressed by introducing an environment variable~$E$ and then
performing a conditional independence test for $Y \independent E
\given X_S$, see also Appendix~\ref{app:proofmodelmis}.   The details of these approaches lie beyond the scope of this paper.

\subsection{Interventions on the target variable and its causal mechanism} \label{sec:violations}
So far, we have assumed that the error distribution of the target
variable is unchanged across all environments $\e\in\E$, see
Assumption~\ref{assum:invariant} for linear models. This
precludes interventions on $Y$ and precludes a change of the causal
mechanism for the target variable. For the gene-knockout experiments
mentioned in Section~\ref{sec:model} and treated in detail in Section~\ref{sec:geneknockout}, we
would for example know whether we have intervened on the target gene
or not. In other situations, we might not be sure whether an
intervention on the target variables occurred or not. 

If interventions are sparse, other approaches are possible, too.
For any given target variable~$Y$, we
might not be sure whether an intervention on $Y$ occurred or not,
but we can assume that an intervention on $Y$ happened in at
  most $V\ll |\E|$ different 
environments, even if we do not know in which of the
  environments it occurred, see 
    \citet{kang2015instrumental} for a related setting in instrumental variable regression.
The null hypothesis~\eqref{eq:H0S-nonlin-cond} in the general nonlinear
case can then be weakened to 
\begin{equation}  
H'_{0,S,\mathit{nonlin}}(\E):
\; 
\begin{array}{l}  
\exists \E'\subseteq \E \mbox{
    with } |\E'| \ge  |\E|-V
\text{ s.t. }
  \forall \e, f \in\E' \mbox{ the cond.\ distr.\ } 
  Y^\e \given X_S^\e = x \\
  \text{and } Y^f \given X_S^f = x
  \text{ are identical } \forall x \text{ s.t.\ both cond.\ distr.\ are well-defined}.
  \end{array}
\label{eq:H0S-nonlin-weak} 
\end{equation}  
  
  The null hypothesis $H'_{0,S^*,\mathit{nonlin}}$ is then still 
true even when interventions happen on $Y$ in some environments, where $S^*$ is the causal set of variables
that satisfies the
invariance assumption in the absence of interventions on $Y$. 
 Any test
for~\eqref{eq:H0S-nonlin-cond} can be extended as a test for the weaker
null hypothesis~\eqref{eq:H0S-nonlin-weak} by testing all subsets
$\E'$ with $|\E'| \ge |\E|- V$ at level $\alpha$, e.g. using a test for~\eqref{eq:H0S-nonlin},  and
rejecting~\eqref{eq:H0S-nonlin-weak} only if we can reject all such
subsets. We can then treat $H_{0,S, \mathit{nonlin}}(\E)$ as
  being ``accepted'' if we find one subset $\E'$ whose corresponding
  null hypothesis
  cannot be rejected.

\subsection{Model misspecification} \label{sec:modelmismain}

We have shown how the approach can be extended to cover hidden variables, nonlinear models and interventions on the target variable. The question arises how the original approach behaves if these model assumptions are violated but we use the original approach instead of the proposed extensions. 
We again write $\hat{S}(\E)$ as in~\eqref{eq:hatcausal2} as
\begin{equation*} 
\hat{S}(\E) \; := \; 
\bigcap_{S: H_{0,S} \text{ not rejected}} S.
\end{equation*}
Our approach still satisfies the coverage property $P(\hat{S}(\E) \subseteq S^*) \ge 1-\alpha$ for any set $S^*$ that satisfies Assumption~\ref{assum:invariant}. 
Let $S_c^{*}$ be a set that is considered to be causal, for example, because it is the set of observed parents of $Y$ in a SEM. 
Under no model misspecification, Proposition~\ref{propos:sem} shows that this set will satisfy Assumption~\ref{assum:invariant} or, in the general case Equation~\eqref{eq:H0S-nonlin-cond}.
If the model assumptions are violated, however, then either $H_{0,S_c^*}$ is still true (in which case the desired confidence statements $P(\hat{S}(\E) \subseteq S_c^*) \ge 1-\alpha$ is still valid) or $H_{0,S_c^*}$ is not longer true. 
The latter case thus warrants our attention. There are two  possibilities. If $H_{0,S}$ is also false for all other sets $S\subseteq\{1,\ldots,p\}$, then $\hat{S}(\E)=\emptyset$ 
for a test that has maximal power to reject false hypotheses.
Thus, the desired coverage property $P(\hat{S}(\E) \subseteq S_c^*) \ge 1-\alpha$ is still valid, even though the method will now have no power to detect the causal variables. It could happen, on the other hand, 
that 
there exists some set $S'\subseteq\{1,\ldots,p\}$ with $S'\setminus S^*_c \neq \emptyset$ for which $H_{0,S'}$ is true.
Proposition~\ref{prop:modelmis} in Appendix~\ref{app:proofmodelmis} 
shows that under some assumptions even in this case, the mistake is
not too severe: 
then there exists a different set $\tilde{S}$,  for which $H_{0,S'}$ is true, and that contains only ancestors of
the target $Y$ and no descendants. 
Then, by construction, the same also holds for $\hat{S}(\C{E})$, with probability greater than $1-\alpha$.

\section{Numerical results}\label{sec:numerical}

We apply the method to simulated data, gene perturbation experiments from
biology with interventional data and and an instrumental variable type setting
from educational research.  

\subsection{Simulation experiments}\label{sec:simul}
For the simulations, we generate data from randomly chosen linear Gaussian structural equation
models (SEMs) and compare various approaches to recover the causal predictors of
a target variable. 

The generation of linear Gaussian SEMs is
described in Appendix~\ref{app:pars}.
We sample 100 different settings and for each of those 100 settings, we
generate 1000 data sets. 
We tried to cover a wide range of scenarios, some (but not all of which)
correspond to the 
theoretical results developed in Section~\ref{sec:idres}.  
After randomly choosing a node as target variable, we can then test how well
various methods recover the parents (the causal predictors) of this target. We check whether false variables were selected as parents (false
positives) or whether the correct parents were recovered (true positives).   

For the proposed invariant prediction method, we divide the data into a block of
observational data and a block of data with interventions. Some other
existing methods make use of the exact nature of the interventions but
for our proposed method
this information is discarded or presumed unknown.
The estimated causal predictors $\hat{S}(\E)$ at confidence 
$95\%$, computed as in Method~I in Section~\ref{sect:invariant},
are then  compared to 
the true causal predictors $S^{\caus}$ of a target variable in 
the causal graph (which can sometimes be the empty set).  
The results of
Method~II are very similar in the simulations and are not shown
separately. We record whether any errors were made
($\hat{S}(\E) \nsubseteq S^{\caus}$) and whether the correct set was
recovered ($\hat{S}(\E)= S^{\caus}$). 
We compare the proposed confidence intervals with point estimates
given by several procedures for linear SEMs:
\begin{enumerate}
\item \emph{Greedy equivalence search (GES)} \citep{Chickering2002}. 
In the case of purely observational data, we can identify the so-called Markov equivalence 
class of the correct graph from the joint distribution, i.e. we can find its skeleton and orient the v-structures, i.e. some of the edges \citep{Verma1991}. 
Although, many directions remain ambiguous in the general case, it might be that we can orient some 
connections of the target variable $X_j - Y$. If the edge is pointing towards $Y$, we identify $X_j$ as a direct cause of $Y$.
The GES searches greedily over equivalence classes of graph structures in
order to maximise a penalised likelihood score. 
Here, we apply GES on the pooled data set, pretending that all data are observational.

\item \emph{Greedy interventional equivalence search (GIES) with known
    intervention targets} \linebreak \citep{Hauser2012}. 
The greedy interventional equivalence search (GIES) considers soft
interventions (at node $j$) where the conditional $p(x_j\given x_{\PA{j}})$ is
replaced by a Gaussian density in $x_j$. One can identify interventional
Markov equivalence classes from the available distributions that are
usually smaller than the Markov equivalence classes obtained from
observational data. GIES is a search procedure over interventional Markov
equivalence classes maximising a penalised likelihood score.
In comparison, a benefit of our new approach is that we do not need to
specify the different experimental conditions. More precisely, we do not
need to know which nodes have been intervened on.  

\item \emph{Greedy interventional equivalence search (GIES) with unknown
    intervention targets}.
To obtain a more fair comparison to the other methods, we hide the
intervention targets from the GIES algorithm and pretend that every
variable has been intervened on. 

\item \emph{Linear non-Gaussian acyclic models (LiNGAM)} \citep{Shimizu2006}.
The assumption of non-Gaussian distributions for the structural equations
leads to identifiability. We use an \texttt{R}-implementation \citep{R} of LiNGAM which is based
on independent component analysis, as originally proposed by \citet{Shimizu2006}.
In the observational setting, the structural equation of a specific
variable $X_j$ reads 
$$
X_j^1 = \sum_{k \in \PA{j}} \beta_{j,k} X_k^1 + \varepsilon_j^1,
$$
whereas in the interventional setting (if the coefficients $\beta_{j,k}$ remain the same), we have 
$$
X_j^2 = \sum_{k \in \PA{j}} \beta_{j,k} X_k^2 + \varepsilon_j^2.
$$
One may want to model the pooled data set as coming from a structural
equation model of the form
$$
\tilde X_j = \sum_{k \in \PA{j}} \beta_{j,k} \tilde X_k + \tilde \varepsilon_j,
$$
where $\tilde \varepsilon_j$ follows a distribution of the mixture of
$\varepsilon_j^1$ and $\varepsilon_j^2$ and thus has a
  non-Gaussian distribution (Kun Zhang mentioned this idea to JP in a private discussion). 
The new noise variables $\tilde \varepsilon_{1}, \ldots, \tilde \varepsilon_{p}$ are not independent of each other: if, for any $j\neq k$, $\tilde \varepsilon_j$ comes from the first mixture, then $\tilde \varepsilon_k$ does so, too.
We can neglect this violation of LiNGAM and apply
the method nevertheless. There is no theoretical result which would justify
LiNGAM for interventional data. 
\item \emph{Regression}. We pool all data and use a linear least-squares regression and
  retain all variables which are significant at level $\alpha/p$, in
  an attempt to control the family-wise error rate (FWER) of falsely
  selecting at least a single variable at level $\alpha$ in a
  regression (not causal) sense. As a regression technique, this method
  cannot correctly identify causal predictors. 

\item \emph{Marginal regression}. We pool all data and retain all variables that have a
  correlation with the outcome at significance level $\alpha/p$. As above,
  this regression method cannot correctly identify causal predictors. 

\end{enumerate}

\begin{figure}
\begin{center}
\includegraphics[width=0.85\textwidth]{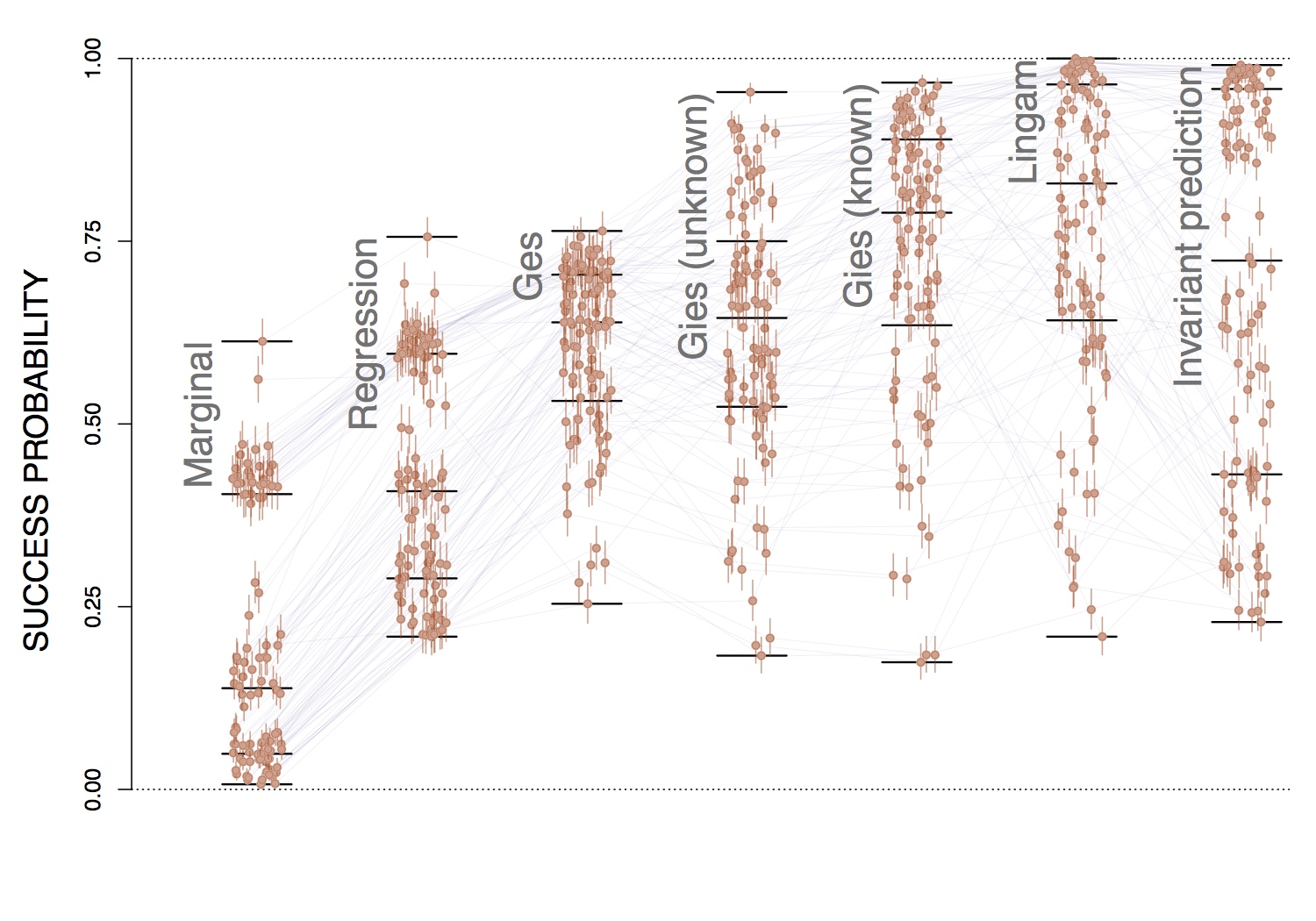}
\caption{ \label{fig:comp} The probability of success, defined as
  $P(\hat{S}(\E)=S^{\caus})$ for various methods, including our new proposed
  invariant prediction in the rightmost column. 
Each dot within a column (the x-offset within a column is uniform) corresponds to one of the 100 simulation
 scenarios.
 The dot's height shows 
 the empirical probability of success
  over 1000 simulations and the small bars indicate a 95\% confidence for
the true success probability. 
Identical scenarios are connected by grey solid lines. 
For each method,
the 
maximal and minimal values along with the quartiles of each
distribution are indicated by horizontal solid bars.}
\end{center}
\end{figure}

\begin{figure}
\begin{center}
\includegraphics[width=0.85\textwidth]{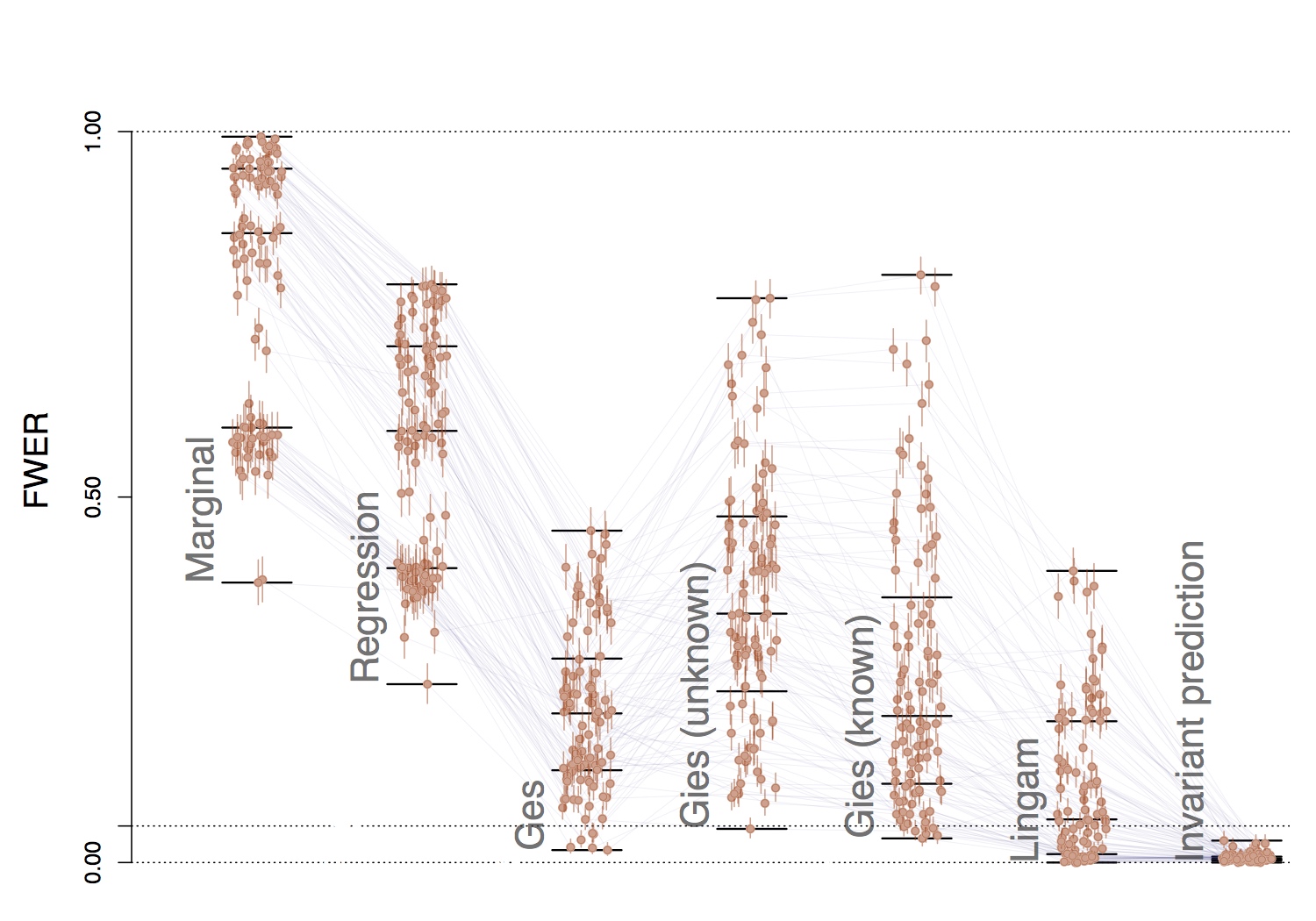}
\caption{ \label{fig:compF}  The probability of erroneous selections
  $P( \hat{S}(\E) \nsubseteq S^{\caus})$ (FWER) for the considered
  methods, including the proposed invariant prediction to the right. The
  figure is otherwise analogously generated as Figure~\ref{fig:comp}. 
  The dotted line indicates
the $0.05$ level at which the invariant prediction method was (successfully)
controlled. All other methods do not offer FWER control. }
\end{center}
\end{figure}

We show the (empirical) probability of false selections, $P(\hat{S}(\E) \nsubseteq S^{\caus})$, in
Figure~\ref{fig:compF} for all methods.  The probability of success, $P(\hat{S}(\E)=S^{\caus})$,
is shown in Figure~\ref{fig:comp}.

The success probabilities show some interesting patterns. First, there
is (as expected) not a method that performs uniformly best overall
scenarios. However, \emph{regression} and \emph{marginal regression} are dominated across all
100 scenarios by \emph{GIES} (both with known and unknown
interventions), \emph{LiNGAM} and the proposed {\emph{invariant
    prediction}). Among the 100 settings, there were 3 where \emph{GES}
  performed best on the given criterion, 14 where \emph{GIES} (with known interventions)
  performed best, 54 for \emph{LiNGAM} and 23 where the proposed
  \emph{invariant prediction} were optimal for exact recovery. There is
  no clear pattern as to which parameter is driving the differences in
  the performances: Spearman's correlation between the parameter settings
and the differences in performances between all pairs of methods
  was less than 0.3 for all parameters. The interactions between the parameter
  settings seem responsible for the relative merits of one method over
  another. 

The pattern for false selections in Figure~\ref{fig:compF} is very
clear on the other hand. The proposed invariant prediction method controls
the rate at which mistakes are made at the desired $0.05$ (and often
lower due to a conservativeness of the procedure). All other methods
have FWE rates that reach 0.4 and higher. No other method offers a
control of FWER and the results show that the probability of erroneous
selections can indeed be very high. The control of the FWER (and the
associated confidence intervals) is the key advantage of the
proposed \emph{invariant prediction}.

\subsection{Gene perturbation experiments}\label{sec:geneknockout}
\paragraph{Data set.}
We applied our method to a yeast (Saccharomyces cerevisiae) data set
\citep{Kemmeren2014}. Genome-wide mRNA expression levels in yeast were
measured and we therefore have data for $p=6170$ genes. There are $n_{obs} = 160$
``observational'' samples of wild-types and $n_{int} = 1479$ data points for
the ``interventional'' setting where each of them corresponds to a
  strain for
which a single gene $k \in K := \{k_1, \ldots, k_{1479}\} \subset \{1,
\ldots, 6170\}$ has been deleted 
(meanwhile, there is an updated data set with five more mutants). 
If the method suggests, for example, gene $5954$ as a cause of gene $4710$, and there is a deletion strain corresponding to gene $5954$, we can use this data point to determine whether gene $5954$ indeed has a (possibly indirect) causal influence on $4710$. We
say that the pair is a true positive if the expression level of gene $4710$ after intervening on $5954$ lies in the $1\%$ lower or upper tail of the observational distribution of gene $4710$, see also Figure~\ref{fig:genes} below. (We  
additionally require that the intervention on gene $5954$ appears to be
``successful'' in the sense that the expression level of gene $5954$ after intervening on this gene $5954$ lies in the $1\%$
lower or upper tail of the observational distribution of gene $5954$. This was not the case
for $38$ out of the $1479$ interventions.) With this criterion, there are
about $9.2\%$ relevant effects, which corresponds
to the proportion of true positives for a random guessing method.

\paragraph{Separation into observational and interventional data.}
For predicting a causal influence of, say, gene $5954$ on another gene we do not want to
use interventions on the same gene $5954$ (this would use information about the ground truth).  We therefore apply the following procedure:
for each $k \in K$ we consider the observational data as $e=1$ and the
remaining $1478 = 1479-1$ data points corresponding to the deletions of
genes in $K \setminus \{k\}$ as the interventional setting $e=2$. Since this
would require $n_{int} \times p$ applications of our method, we instead
separate $K$ into $B=3$ subsets of equal size, consider the two subsets
not containing $k$ as the interventional data, and do not make any use of the
  subset containing $k$. This leaves some information in the data unused but yields a huge
computational speed-up, since we need to apply our method in total only $3
\times p$ times. Additionally, when looking for potential causes of gene
$4710$, we do not consider data points corresponding to interventions on
this gene (if it exists), see Proposition~\ref{propos:sem}.

\paragraph{Goodness of fit and $p$-values.}
If we would like to avoid making a single mistake on the data set with
high probability $1-\alpha$, we can set the significance level to
for each gene to $\alpha/n_{int}$, using a Bonferroni correction in order to
take into account the $n_{int} = 1479$ genes that have been intervened
on. We work with $\alpha=0.01$ if not mentioned otherwise. The guarantee requires,
however, that the model is correct (for example the linearity
assumption is correct and there are no hidden variables with strong
effects on both genes of interest). These assumptions are likely
violated, and the implications have been partially discussed in the
previous Section~\ref{sec:furtherext}.
To further guard against false positives that are due to model
misspecification we require that there is at least one model (one
subset $S\subseteq\{1,\ldots,p\}$) for which the model fits reasonably
well: we define this by requiring a $p$-value above $0.1$  for
  testing $H_{0,S}(\E)$
  for the best-fitting set $S$ of variables (the set with the highest $p$-value), if not mentioned otherwise (but we also vary the
threshold to test how sensitive our method is with regard to parameter
settings). If no set of variables attains this threshold, we discard the
models and make no prediction.

\paragraph{Method.}
We use $L_2$-boosting \citep{fried01,buhlmann2003boosting} from the
R-package \texttt{mboost} \citep{Hothorn2010} with shrinkage $0.1$ as a way
to preselect for each response variable ten potentially causal variables,
to which we then apply the causal inference methods. We primarily use
 Method II as Method I requires subsampling  for computational
 reasons. Subsampling can lead to a loss of power as there is a
 not-negligible probability of loosing the few informative data points
 in the subsampling process.
For a computational speed-up we only consider subsets of size $\leq 3$ as
candidate sets $S$. 
Furthermore, we only retain results where just a single variable has been
shown to have a causal influence to avoid testing more difficult
scenarios where one would have to intervene on multiple genes
simultaneously.

\paragraph{Comparisons.} As alternative
methods we consider IDA
\citep{Maathuis2009} based 
on the PC algorithm \citep{Spirtes2000}
and a method that ranks the absolute value of marginal
correlation ($j_1 \rightarrow j_2$ and $j_2 \rightarrow j_1$ obtain the same score and are ranked randomly), both of which make use only of the observational data.    
We also compare with IDA based on greedy interventional equivalence search
(GIES) \citep[][]{hapb14} and a correlation-based method that ranks pairs according to correlation on the pooled observational and interventional
data. 
It was not feasible to run LiNGAM \citep{Shimizu2011} on this data
set.

\paragraph{Results.}
 \begin{figure}
\begin{center}
\includegraphics[width = 0.9 \textwidth]{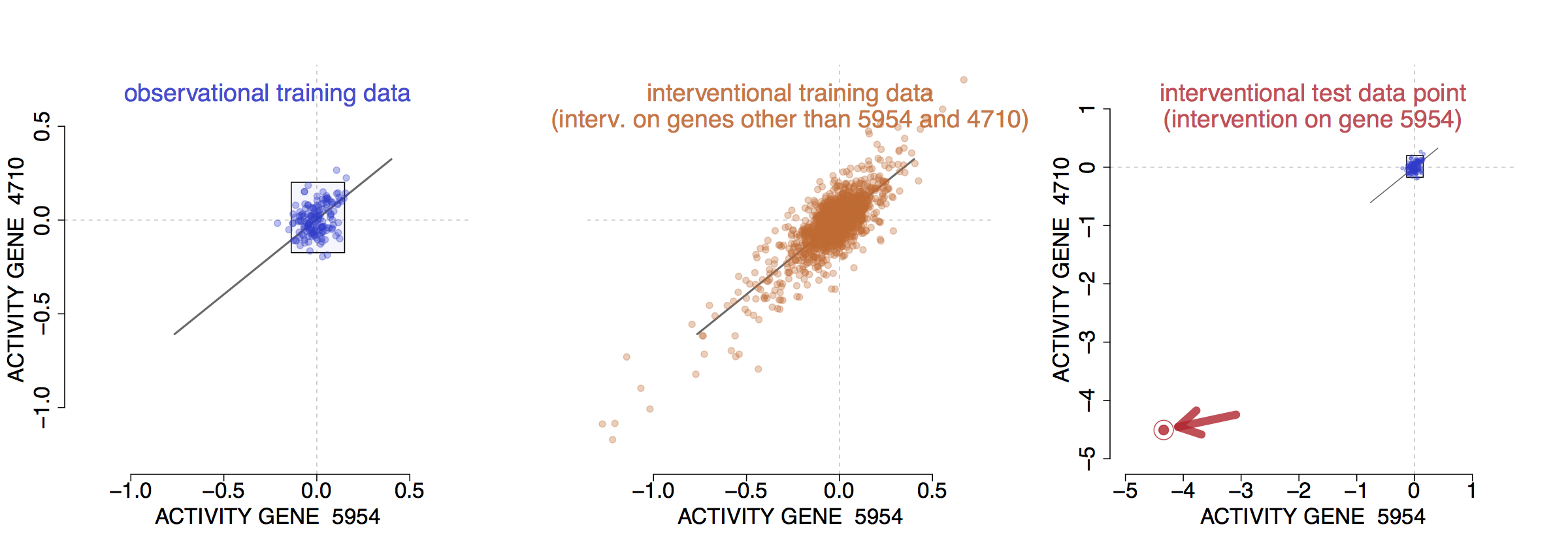}
\includegraphics[width = 0.9 \textwidth]{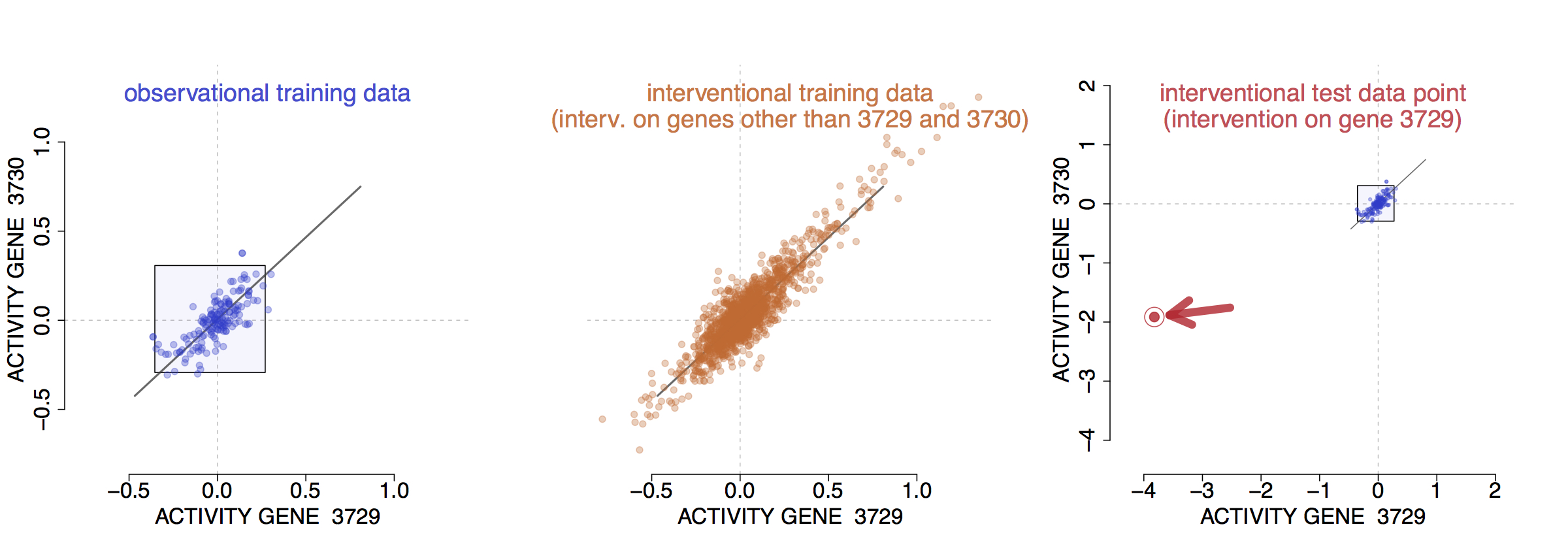}
\includegraphics[width = 0.9 \textwidth]{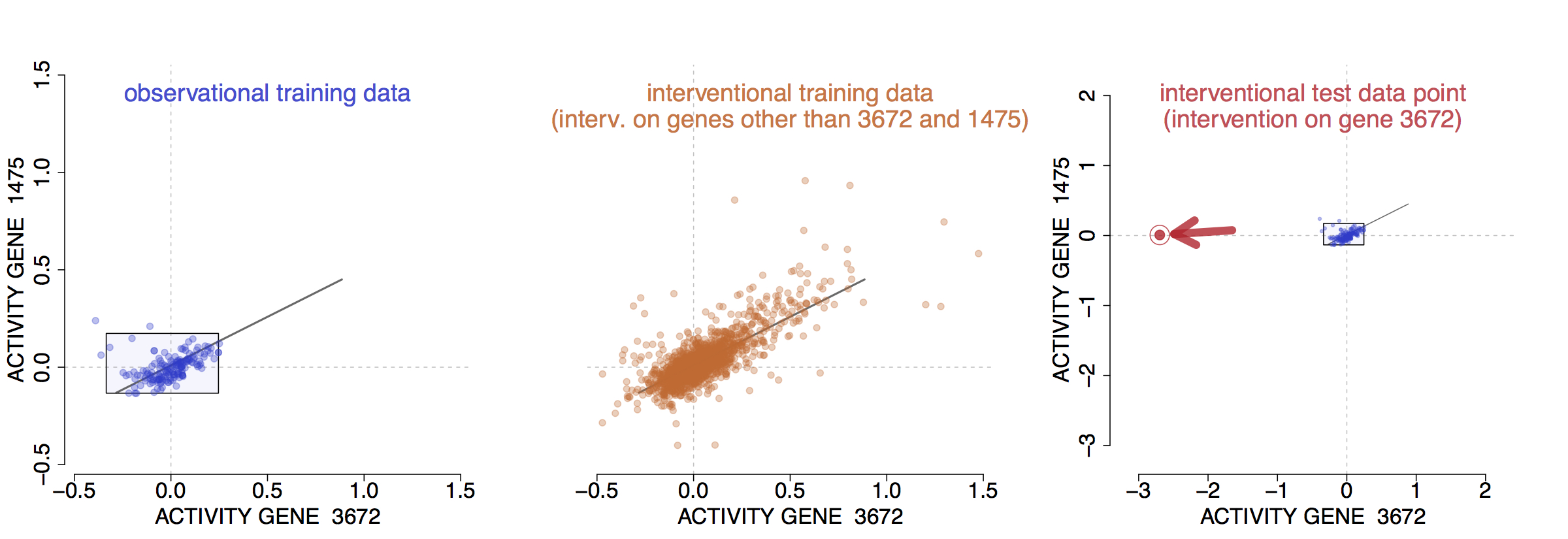}
\caption{ The three rows correspond to the three most significant
  effects found by the proposed method (with the most significant
  effect on top, suggesting a causal effect of gene 5954 on gene 4710). The left column shows the
  observational data,
while the second column shows the interventional data (that are
neither using interventions on the target variable itself nor using
interventions on the examined possible causal predictors of the target
variable); these two data sets are used as two environments for training the invariant prediction model. The regression line for a joint model of observational and
interventional data, as proposed in Method II, is shown in both plots;
we cannot reject the hypothesis that the regression is different
for observational and interventional data here. The third column
finally shows the test data (with the 1\%-99\% quantile-range of the observational data
shown as a shaded box as in the first column). There, we use the
intervention data point on the
chosen gene and look at the effect on the target variable. The first
two predicted causal effects can be seen to be correct (true
positives) in the following sense: after successfully intervening on
the predicted cause, the target gene shows reduced activity; the third
suggested pair is unsuccessful (false positive) since the intervention reduces the activity of the cause but the target gene remains as active as in the observational data.}
\label{fig:genes}
\end{center}
\end{figure}
The proposed method (Method II) outputs eight gene pairs that can be checked because the corresponding interventional experiments are available.
There are in total eight causal effects that are significant at level $0.01$ after a Bonferroni correction.
Out of these eight pairs, six are correct (random guessing has a success probability of $9.2\%$). 
Figure~\ref{fig:genes} shows the three pairs that obtained the highest
rank, i.e. smallest $p$-values. The rows in the figure therefore
correspond to the three causal effects in the data set that were
regarded as most significant by our method.  One note regarding the
plot: we plot all available data even though only two-thirds of it was
effectively used for training due to the discussed cross-validation
scheme. Many outlying points in the interventional training data of the
false
positive (second column of third row in
Figure~\ref{fig:genes}) are in particular not part of the training
data and the method might have performed better with a more
computationally-intensive validation scheme that would split the data into
$B$ blocks with $B$ larger than the currently used $B=3$.

In order to compare with other methods (none of which provide a measure of significance), we always consider the eight highest-ranked pairs. Table~\ref{tab:genes} summarises the results. In this data set, the alternative methods were not able to exceed random guessing. 

\begin{table}
\caption{\label{tab:genes} The number of true effects among the
  strongest $8$ effects that have  been found in the interventional
  test data (the number $8$ has been chosen to correspond to the
  number of significant effects under the proposed Method II). Method
  I is based on $1000$ samples and required roughly $10$ times more
  computational time than Method II.}
\fbox{%
\begin{tabular}{*{8}{c}}
\multirow{2}{*}{method} & \multirow{2}{*}{Method I} & \multirow{2}{*}{Method II}  & \multirow{2}{*}{GIES} & \multirow{2}{*}{IDA} & \multicolumn{2}{c}{marginal corr.} & random\\ 
\cline{6-7} &  & &  & &  observ. & pooled &  guessing\\ 
\cline{1-8}  \# of true & \multirow{3}{*}{6}  & \multirow{3}{*}{6} & \multirow{3}{*}{2} & \multirow{3}{*}{2} & \multirow{3}{*}{1} & \multirow{3}{*}{2}& 2 ($95\%$ quantile) \\ positives& &&&&&&3 ($99\%$ quantile)\\
(out of $8$)&&&&&&&4 ($99.9\%$ quantile)
\end{tabular}}
\end{table}
To test sensitivity of the results to the chosen implementation
details of the method, 
the variable pre-selection, the goodness-of-fit cutoff have also all
been varied (for example using Lasso instead of boosting as
pre-selection and using a cutoff of 0.1 instead of 0.01).
For Method II, variable selection with Lasso instead of boosting leads to a true positive rate of 
$0.63$ ($5$ out of $8$).
Choosing the goodness-of-fit cutoff at $0.01$ rather than $0.1$ leads to true positive rates of 
$0.43$ ($9$ out of $21$) for boosting and
$0.47$ ($8$ out of $17$) for Lasso. 
Method I without forcing eight decisions leads to a true positive rate of
$0.75$ (3 out of 4) for boosting and 
$1.00$ (1 out of 1) for Lasso.
Choosing the goodness-of-fit cutoff at $0.01$ rather than $0.1$ leads to true positive rates of 
$0.86$ (6 out of 7) for boosting and
$0.75$ (3 out of 4) for Lasso.
(Using $500$ instead of $1000$ subsamples for Method~I leads to increased speed and worse performance.)
We regard it as encouraging that the true positive rate is always larger than random guessing, irrespective of
the precise implementation of the method.

Among the reasons for false
positives (e.g. $2$ out of $8$ for Method II in Table~\ref{tab:genes}, 
there are at least
the following options: (a) noise fluctuations, (b) nonlinearities, (c) hidden variables, (d) issues with
the experiment (for example the intervention might have changed other
parts of the network) and~(e) the pair is a true positive but is -by chance- classified as a false positive by our criterion (see ``Data set'' above). Missing causal variables in the pre-screening by
boosting or Lasso falls under category (c).  We control (a) and have provided
arguments why (b) and (c) will lead to rejection of the whole model
rather than lead to false positives.
Lowering the goodness-of-fit-threshold seemed indeed to lead to more
spurious results, as expected from the discussion in the previous
Section~\ref{sec:modelmismain}. Validating a potential issue with the experiment as in
reason (d) is beyond our possibilities. We could address (e) if we had access to multiple repetitions of the intervention experiments.

\subsection{Educational attainment}\label{sec:educ}
\begin{figure}
\begin{center}
\includegraphics[width=0.85\textwidth]{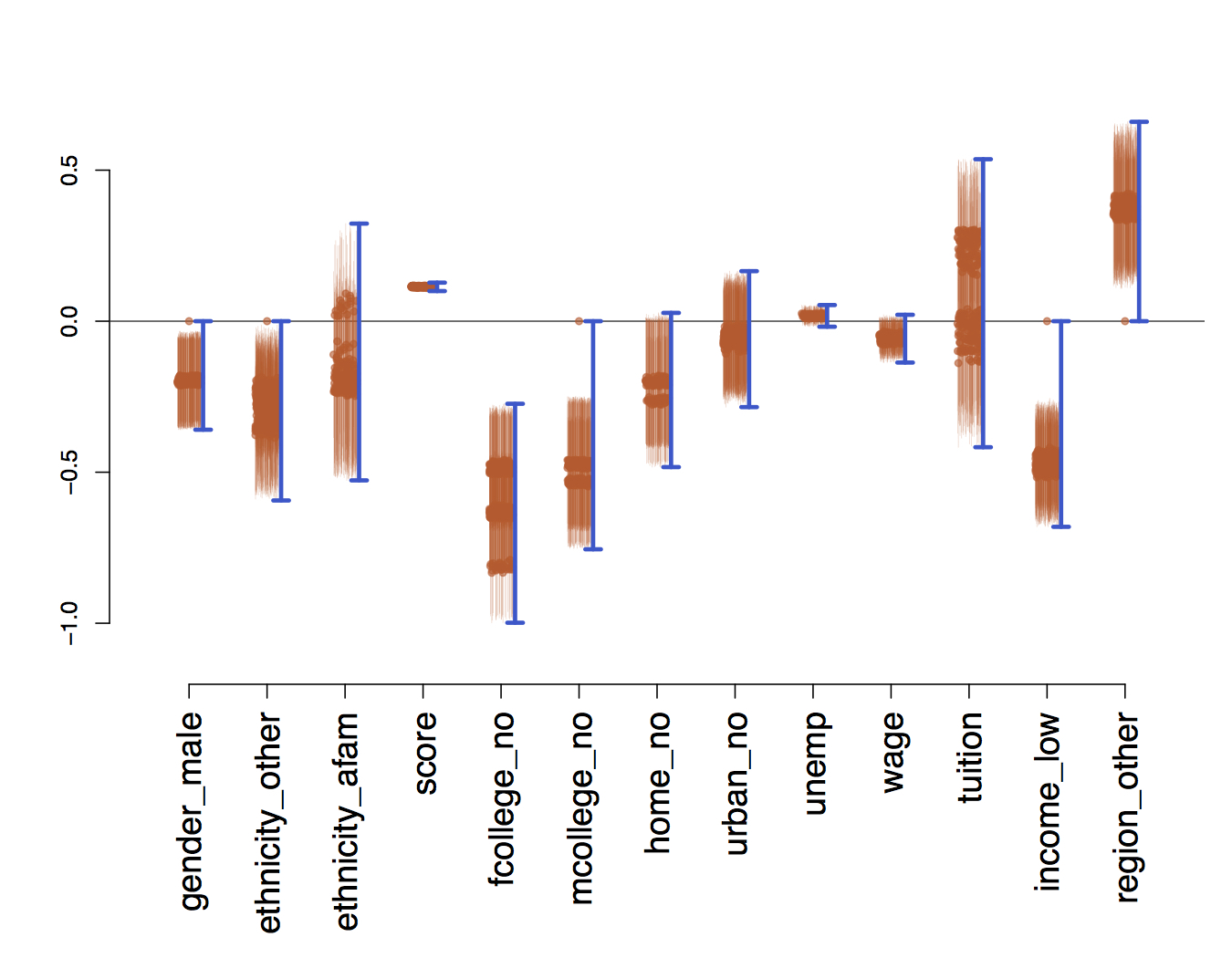}
\caption{\label{fig:edu}  The 90\% confidence intervals for the influence of various variables
on the probability of receiving a BA degree (or higher) are shown in
blue. Of all 8192 possible sets $S$, we accept 1565 sets (the empty
set is not accepted as the probability of receiving a degree is
sufficiently different for people within a close distance to a
4-year college and further away). The 
point-estimates for the coefficients are shown for these 1565 
sets as 
red dots and the corresponding confidence intervals as vertical red
bars. The blue confidence intervals are then the union of all 1565
confidence intervals, as in our proposed procedure. The variables
\emph{score} (test score) and 
\emph{fcollege\_no} (active if father did not receive a college
degree) show significant effects.
}
\end{center}
\end{figure}

We look at a data set about educational attainment of
teenagers \citep{rouse1995democratization}. For 4739 pupils from approximately
1100 US high schools, 13 attributes are recorded, including gender,
race, scores on relevant achievement tests, whether the parents
are college graduates, or family income. 
Here we work with the data as provided in \citet{stock2003introduction},
where we can see the length of education pupils received. We
make a binary distinction into whether pupils received a BA degree or
higher (equivalent to at least 16 years of education in the
classification used in \citet{stock2003introduction}) and ask whether
we can identify a causal predictive model that allows to forecast
whether students will receive a BA degree or not and this forms a
binary target $Y$. 

 The distance to the nearest
4-year college is recorded in the data and we use it to split the dataset into two
parts in the sense of~\eqref{eq:datasplittingworks}; we assume that this variable has no \emph{direct} influence on the target variable.  
As discussed, this variable does not have to
satisfy the usual assumptions about instrumental variables for our
analysis but just has to be independent of the noise in the outcome
variable (it must be a non-descendant of the target), which seems satisfied in this dataset as the distance to the
4-year college precedes the educational attainment chronologically. One set of observations are thus all pupils who live closer
to a 4-year college than the median
distance of 10 miles. The second set are all other pupils, who live at
least 10 miles from the nearest 4-year college. We ask for a
classification that is invariant in both cases in the sense that the conditional
distribution of $Y$, given $X$, is identical for both groups, where $X$
are the set of collected attributes and $Y$ is the binary outcome of
whether they attained a BA degree or higher. We use the fast
approximate Method II of
Section~\ref{sect:invariant}, with the suggested extension to logistic
regression.

Figure~\ref{fig:edu} shows the outcome of the analysis, which is also
included as an example in the \texttt{R}-package \texttt{InvariantCausalPrediction}. Factors were split into dummy variables 
so that ``ethnicity\_afam'' is~1 if the ethinicity is african-american
and~0 otherwise, ``fcollege\_no'' is~1 if the father did not receive a
college degree and so forth. We provide
90\% confidence intervals. All of them include~0 except
for the confidence interval for the influence of the test score
(positive effect) and
the indicator that the father did not receive a college degree
(negative effect). A high score on the
achievement test thus seems to have a positive causal influence on the
probability of obtaining a BA degree, which seems plausible. 

As it is difficult to verify the ground truth in this case, we refrain
from comparisons with other possible approaches to the same data set
and just want to use it as an example of a possible practical application.
The example shows that we can use instrumental-variable-type variables
to split the data set into different ``experimental'' groups. If the
distributions of 
the outcome are sufficiently different in the created groups, we
can potentially have power to detect invariant causal prediction
effects. 

\section{Discussion and Future Work}\label{sec:discussion}

An advantage of causal predictors compared to non-causal ones is that
their influence on the target variable remains invariant under
different changes of the environment (which arise for example through interventions). 
We have described this invariance and exploit it
for the identification of the causal predictors. Confidence sets for the causal predictors and confidence intervals for relevant parameters follow naturally in this framework.  In the
special case of Gaussian structural equation models with interventions we
have proved identifiability guarantees for the set of causal predictors.  
We discussed some of the questions that require more work: suitable tests
for equality of conditional distributions for nonlinear models, feedback
models and increased computational efficiency both in the absence and 
presence of hidden variables.  

The approach of invariant prediction provides new concepts and
  methods for causal inference, and also relates to many known concepts 
but considers them
from a different angle. It constitutes a new understanding of 
causality that 
opens the way to a novel class of theory and methodology in causal inference.

\section*{Acknowledgements}
The research
leading to these results has received funding from the People Programme (Marie Curie
Actions) of the European Union's Seventh Framework Programme (FP7/ 2007--2013)
under REA grant agreement no 326496.
The authors would like to thank seven anonymous referees for their helpful comments on an earlier version of the manuscript and would like to thank Alain Hauser, Thomas Richardson, Bernhard Sch\"olkopf and Kun Zhang for helpful discussions.

\bibliographystyle{plainnat}
\bibliography{bibliography}

\appendix

\section{An Example} \label{app:example}
\ifgamma
\else
We illustrate here in Figure~\ref{fig:examplesSEM} the concepts and methodology
which have been developed in Sections~\ref{sec:21new},~\ref{sec:22new} 
and~\ref{sec:conf}.
\fi
The figure shows an example of two environments whose data were generated from observational and interventional structural equation models.

\definecolor{newblue}{RGB}{62,76,209}
\definecolor{newblueLight}{RGB}{184,189,238}

\if0
\hfsetbordercolor{gray!80}
\hfsetfillcolor{newblue!30}
\hfsetfillcolor{newblueLight!100}

\begin{figure}
\begin{center}
\tikzmarkin{glw3}(0.1,-2.1)(0.2,2.2)
\hfsetbordercolor{gray!80}
\hfsetfillcolor{gray!40}
\begin{minipage}[][][c]{0.33\textwidth}
{\small
environment $e=1$:
\begin{alignat*}{2}
 X_2^1 &= 0.3\varepsilon_2^1 \\
X_3^1 &= X_2^1 + 0.2\varepsilon_3^1\\
 \tikzmarkin{glw}(0.1,-0.2)(-0.15,0.4) Y^1 &= -0.7X_2^1 + 0.6X_3^1 + 0.1\varepsilon_Y^1 \tikzmarkend{glw}\\
X_4^1 &= -0.5Y^1 + 0.5X_3^1 + 0.1\varepsilon_4^1
\end{alignat*}
$\varepsilon_Y^1, \varepsilon_2^1, \varepsilon_3^1, \varepsilon_4^1 \overset{iid}{\sim} \mathcal{N}(0,1)$
}
\end{minipage}
\tikzmarkend{glw3}
\hfill
\begin{minipage}[][][c]{0.14\textwidth}
\begin{tikzpicture}[xscale = 0.74, yscale=1, line width=0.5pt, minimum size=0.58cm, inner sep=0.3mm, shorten >=1pt, shorten <=1pt]
    \small
    \draw (0,2) node(x2) [circle, draw] {$X_2$};
    \draw (2,2) node(x3) [circle, draw] {$X_3$};
    \draw (1,1) node(y) [circle, draw] {$Y$};
    \draw (1,0) node(x4) [circle, draw] {$X_4$};
    \draw[-arcsq, dashed] (x2) -- (x3);
    \draw[-arcsq] (x2) -- (y);
    \draw[-arcsq] (x3) -- (y);
    \draw[-arcsq] (y) -- (x4);
   \end{tikzpicture}
\end{minipage}
\hspace{0.02cm}
\definecolor{neworange}{RGB}{221,120,63}
\definecolor{neworangeLight}{RGB}{242,206,185}
\hfsetbordercolor{gray!80}
\hfsetfillcolor{neworange!30}
\hfsetfillcolor{neworangeLight!100}
\tikzmarkin{glw4}(0.2,-2.1)(0.15,2.2)
\hfsetbordercolor{gray!80}
\hfsetfillcolor{gray!40}
\begin{minipage}[][][c]{0.32\textwidth}
{\small
environment $e=2$:
\begin{alignat*}{2}
X_2^2 &= 0.3\varepsilon_2^2\\
X_3^2 &= 0.4\varepsilon_3^2\\
\tikzmarkin{glw2}(0.1,-0.2)(-0.15,0.4) Y^2 &= -0.7X_2^2 + 0.6X_3^2 + 0.1\varepsilon_Y^2 \tikzmarkend{glw2}\\
X_4^2 &= -0.5Y^2 + 0.5X_3^2 + 0.1\varepsilon_4^2
\end{alignat*}
$\varepsilon_Y^2, \varepsilon_2^2, \varepsilon_3^2, \varepsilon_4^2 \overset{iid}{\sim} \mathcal{N}(0,1)$
}
\end{minipage}
\tikzmarkend{glw4}
\\
\includegraphics[width=0.48\textwidth]{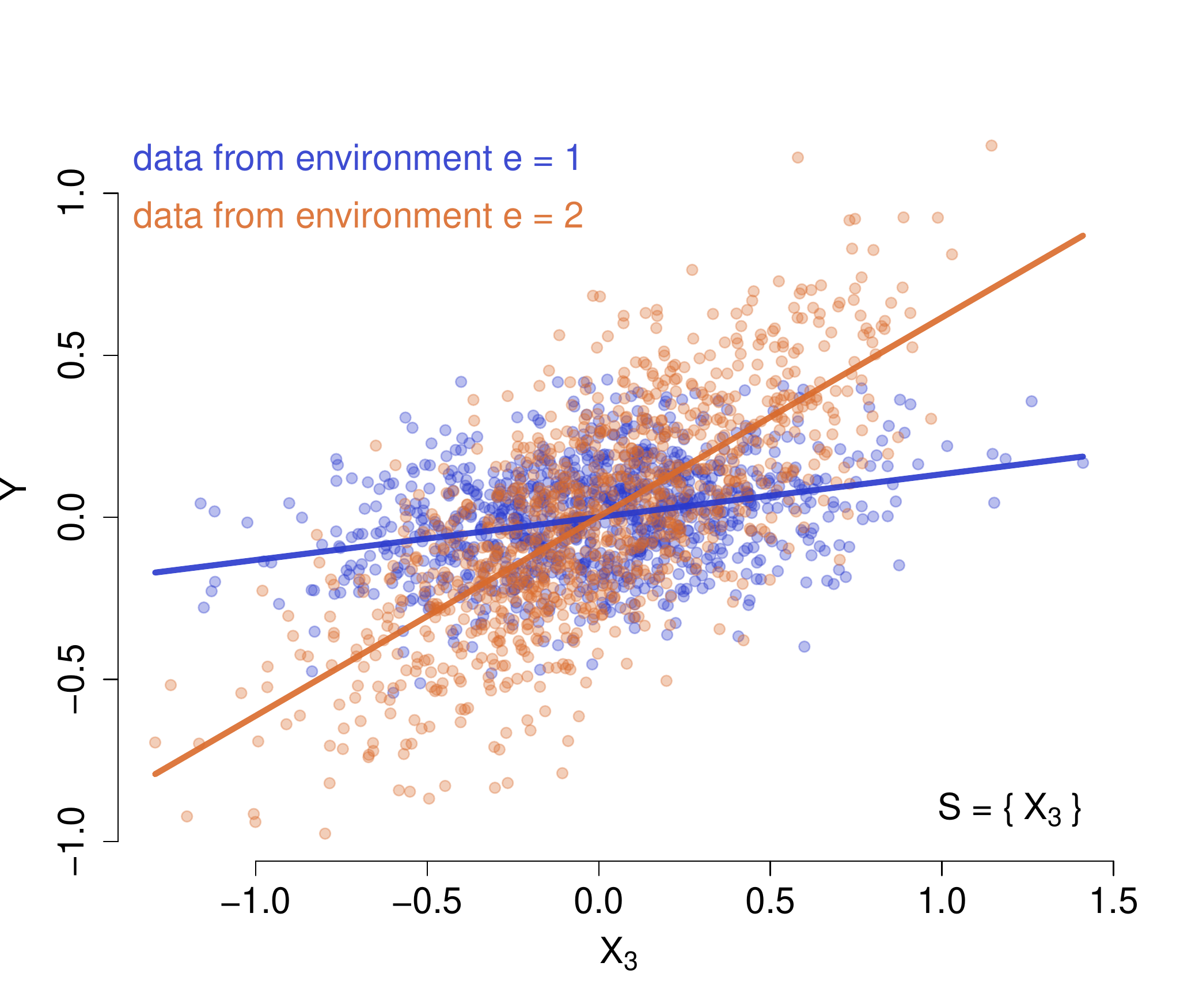}
\hfill
\includegraphics[width=0.48\textwidth]{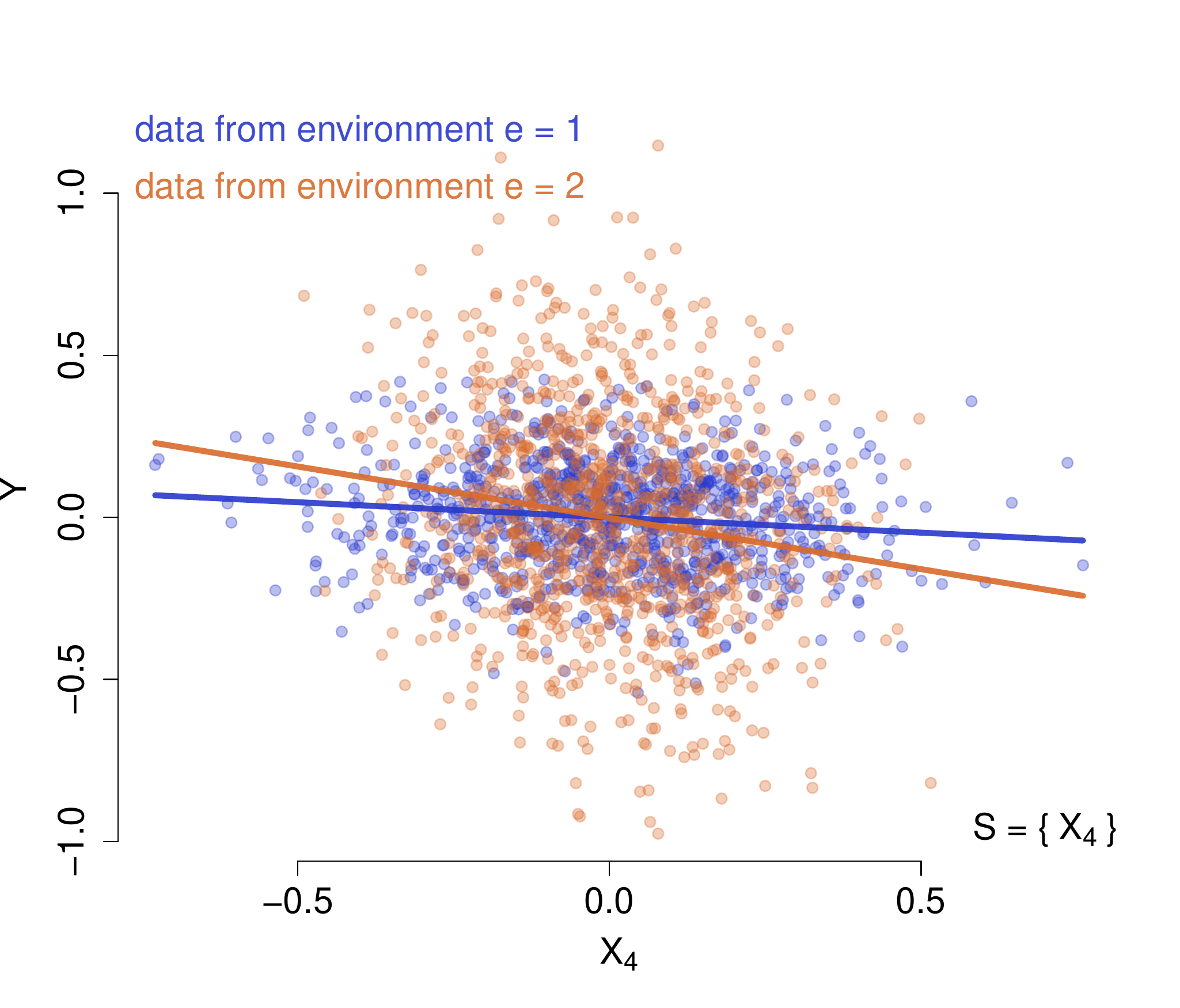}
\caption{ The top row shows the example of two structural equation models (SEMs) entailing the two distributions corresponding to two environments $e=1$ and $e=2$. Here, the first environment corresponds to the graph including the dashed edge, the second environment corresponds to an intervention on $X_3$, the graph excluding the dashed edge. Since the structural equation for $Y$ is unchanged, the set $S^\caus = \{X_2, X_3\} = \PA{1}$ satisfies Assumption~\ref{assum:invariant}, see Proposition~\ref{propos:sem}.
We consider the setup where we  know 
neither $S^\caus$ nor the SEMs (we do not even require the existence of such a SEM). Instead, we are given two finite samples (one from each environment) and provide an estimator $\hat{S}$ for $S^\caus$. 
In the above example, the null hypothesis of invariant prediction gets
rejected for any set $S$ of variables except for $S
  = \{X_2, X_3\}$ and $S= \{X_2, X_3, X_4\}$ (using the methodology described in Section~\ref{sect:invariant}). 
The bottom row shows that for $S = \{X_3\}$, for example, the linear
regression coefficients differ in the two environments. For $S = \{X_4\}$,
the regression coefficients seem similar but the set is rejected because of
varying variances of the residuals. 
We then propose to consider the intersection of the sets of
  variables for which the hypothesis of invariance is not rejected; this leads to the (conservative) estimate 
$\hat{S}$ for the set of identifiable
    predictors $S^*$: $\hat{S} = \{X_2, X_3\} \cap \{X_2, X_3, X_4\} =
    \{X_2, X_3\}$. We thus have for this case $\hat{S} = S^*$, see also Theorem~\ref{prop:onlyone} with $k_0 = 3$.}
\label{fig:examplesSEM}
\end{center}
\end{figure}
\else

\begin{figure}
\begin{center}
\includegraphics[width=0.99\textwidth]{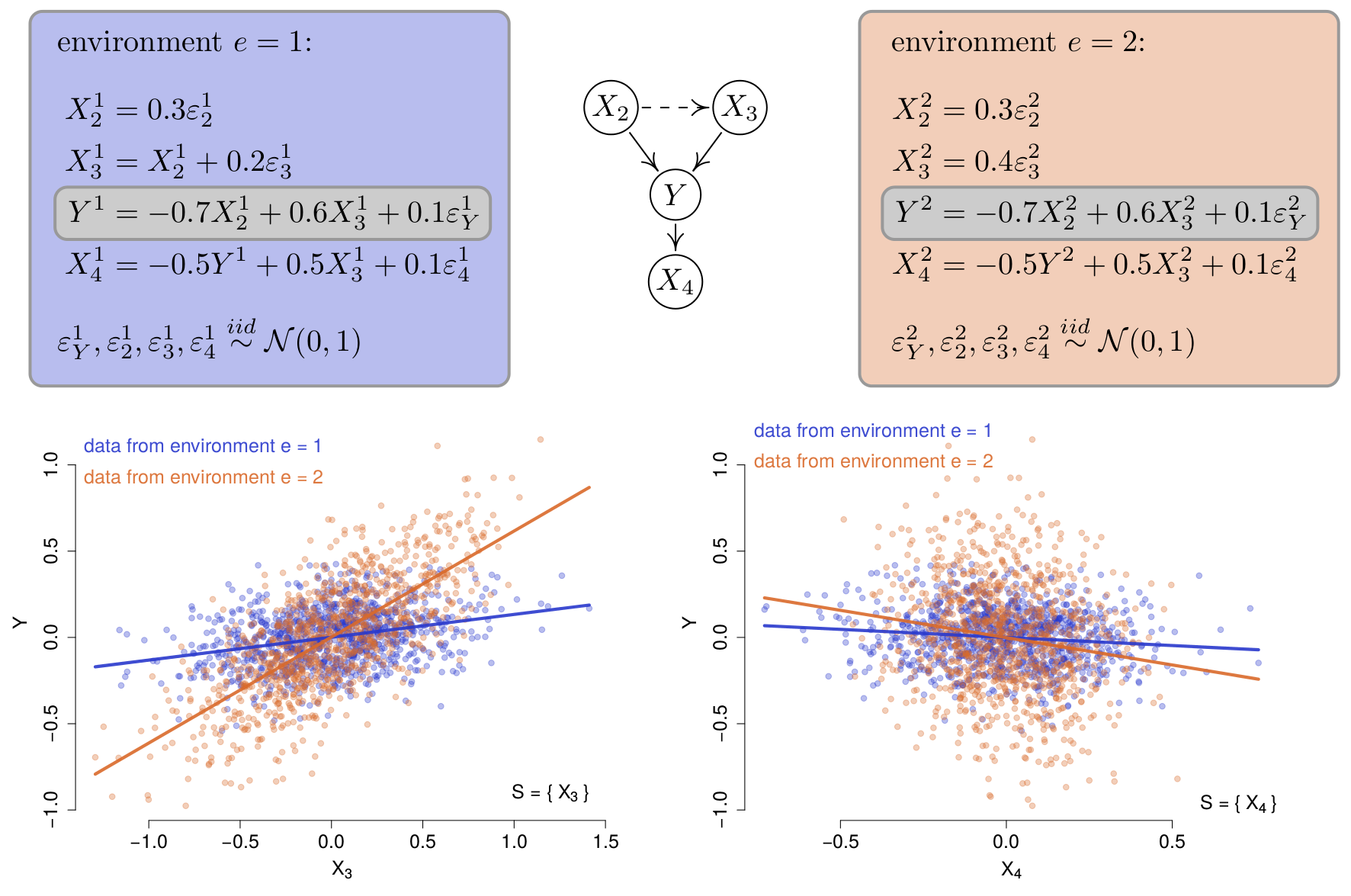}
\caption{ The top row shows the example of two structural equation models (SEMs) entailing the two distributions corresponding to two environments $e=1$ and $e=2$. Here, the first environment corresponds to the graph including the dashed edge, the second environment corresponds to an intervention on $X_3$, the graph excluding the dashed edge. Since the structural equation for $Y$ is unchanged, the set $S^\caus = \{X_2, X_3\} = \PA{1}$ satisfies Assumption~\ref{assum:invariant}, see Proposition~\ref{propos:sem}.
We consider the setup where we  know 
neither $S^\caus$ nor the SEMs (we do not even require the existence of such a SEM). Instead, we are given two finite samples (one from each environment) and provide an estimator $\hat{S}$ for $S^\caus$. 
In the above example, the null hypothesis of invariant prediction gets
rejected for any set $S$ of variables except for $S
  = \{X_2, X_3\}$ and $S= \{X_2, X_3, X_4\}$ (using the methodology described in Section~\ref{sect:invariant}). 
The bottom row shows that for $S = \{X_3\}$, for example, the linear
regression coefficients differ in the two environments. For $S = \{X_4\}$,
the regression coefficients seem similar but the set is rejected because of
varying variances of the residuals. 
We then propose to consider the intersection of the sets of
  variables for which the hypothesis of invariance is not rejected; this leads to the (conservative) estimate 
$\hat{S}$ for the set of identifiable
    predictors $S^*$: $\hat{S} = \{X_2, X_3\} \cap \{X_2, X_3, X_4\} =
    \{X_2, X_3\}$. We thus have for this case $\hat{S} = S^*$, see also Theorem~\ref{prop:onlyone} with $k_0 = 3$.}
\label{fig:examplesSEM}
\end{center}
\end{figure}

\fi

\section{Hidden variables without confounding} \label{app:propsemb}
We discuss first a generalisation of Proposition \ref{propos:sem},
allowing for some hidden variables but excluding confounding between
the observable causal variables and the target variable. Another
setting allowing for such confounding is presented in Section
\ref{subsec.instrvar}. Consider the structural equation model with variables
$X_1=Y,X_2,\ldots ,X_p, X_{p+1} , H_1,\ldots ,H_q$, where the latter
$H_1,\ldots ,H_q$ are unobserved, hidden variables with mean zero.

\begin{prop}\label{propos:semb}
Consider a linear structural equation model including variables \[ (X_1=Y,X_2, \ldots, X_p, X_{p+1} ,H_1,\ldots ,H_q),\]  whose structure is given by a directed acyclic
graph. Denote by 
\begin{eqnarray*}
S^0 := \PA{1} \cap \{2,\ldots ,p+1\}
\end{eqnarray*}
the indices of the observable direct causal variables for $Y$ and by $S_H^0$ the set of indices having a 
directed edge from the hidden variables $H_1,\ldots ,H_q$ to $Y$, i.e., 
${S_H^0}= \PA{1} \setminus S^0$. The 
structural equation for $Y$ is  
\begin{eqnarray*}
Y = \sum_{j \in S^0} \beta_{Y,j} X_j + \sum_{k \in S_H^0} \kappa_{Y,k} H_k
  + \varepsilon_Y,
\end{eqnarray*} 
where $\varepsilon_Y$ is independent of $X_{S^0}$ and $H_{S^0_H}$.\\ 
Then, by choosing $\gamma^* = \{\beta_{Y,j},\ j \in S^0\}$ and $S^* = S^0$,
Assumption 1 holds if one of the following conditions (i) or (ii) is
satisfied.
\begin{enumerate}
\item[(i)] There are no direct causal effects from the hidden variables
  $H_1,\ldots ,H_q$ to the target variable~$Y$, i.e., $S^0_H = \emptyset$,
  and it holds that 
\begin{eqnarray}\label{addPB2}
& &Y^e = \sum_{j \in S^0} \beta_{Y,j} X_j^e + \varepsilon_Y^\e\ \mbox{for all}\ \e
  \in \E,
\end{eqnarray}
where $\varepsilon_Y^\e$ is independent of
$X^\e_{S^0}$ and has the same distribution for all $\e \in \E$.
In particular, this
holds under do- or soft-interventions on the
  variables $\{X_2,\ldots ,X_{p+1}\} \cup \{H_1,\ldots ,H_q\}$ given that~$S^0_H = \emptyset$.
\item[(ii)] There are hidden variables which have a direct effect on the
  target variable $Y$, i.e., $S^0_H \neq \emptyset$. It holds that 
\begin{eqnarray}\label{addPB3}
Y^e = \sum_{j \in S^0} \beta_{Y,j} X_j^e  + \sum_{k \in S^0_H}
    \kappa_{Y,k} H_k^\e + \varepsilon_Y^\e\ \mbox{for all}\ \e
  \in \E,
\end{eqnarray}
where $\sum_{k \in S^0_H}
    \kappa_{Y,k} H_k^\e + \varepsilon_Y^\e$ is
    independent of $X_{S^0}^\e$ and has the same distribution with mean
    zero for all $\e \in \E$. This 
    holds under the
    following conditions (a)-(c):  
\begin{enumerate}
\item[(a)] the experiments $\e \in \mathcal{E}$ arise as do- or soft-interventions;
\item[(b)] there are no interventions on $Y$, on nodes in $S^0_{H}$ or on any ancestor of $S^0_{H}$;
\item[(c)] there is no $d$-connecting path between any node in $S^0$ and $S^0_H$.
\end{enumerate}   
\end{enumerate}

\end{prop}
\begin{proof}
Assumption 1 follows immediately from~\eqref{addPB2}
or~\eqref{addPB3}, respectively. From the definition of the interventions,
as described in Section \ref{sec:intervv}, the justification
for~\eqref{addPB2} follows and hence the claim assuming condition (i). 
When invoking condition (ii), we show now that (a)-(c)
imply~\eqref{addPB3} and the required conditions. Due to (a) and (b), we 
have Equation~\eqref{addPB3} and we
know that the distribution of
$$
\eta^e := \sum_{k \in S^0_H}
  \kappa_{Y,k} H_k^\e + \varepsilon_Y^\e
  $$
is the same for all $\e \in \E$.
Furthermore, $\eta^\e$ is independent of $X_{S^0}^\e$ because of (c).
\end{proof}

\section{Model Misspecification} \label{app:proofmodelmis}
Under model misspecification $S(\E)$ may not be a subset of the direct causes of $Y$ anymore. 
The following proposition shows that in most cases it is still a subset of the ancestors of $Y$ (and is therefore a subset of possibly indirect causes of $Y$). 
The proposition is formulated in the general case, see Section~\ref{sec:nonlin}.
In order to formulate the required faithfulness assumption, we consider an environment variable $E$.
\begin{prop} \label{prop:modelmis}
Consider a SEM over nodes $(Y, X_2, \ldots, X_{p+1}, H_1, \ldots, H_q)$ with hidden variables $H_1, \ldots, H_q$. We now augment the corresponding graph by a discrete environment variable $E \in \E$ \citep[e.g.][]{Pearl2009} 
that satisfies $P(E = e) > 0$ for all $e \in \E$ and has a directed
edge to any node that is do- or soft-intervened on. Let us assume that the joint distribution over $(Y, X_2, \ldots, X_{p+1}, H_1, \ldots, H_q, E)$ is faithful w.r.t.\ the augmented graph. 
Then 
$$
S(\E) := \bigcap_{S\,:\, H_{0,S,nonlin}(\E) \text{ is true }} \! \!S \; \; \subseteq \; \mathbf{AN}(Y) \cap \{ X_2, \ldots, X_{p+1}\}.
$$
\end{prop}
In particular, this proposition still holds under 
model misspecification when for some do-interventions, for example,
$S^0 = \mathbf{PA}(Y) \cap \{X_2,\ldots, X_{p+1}\}$ does not
satisfy 
$H_{0,S,nonlin}(\E)$~\eqref{eq:H0S-nonlin}; Figure~\ref{fig:exmisspef} shows an example. 
The following proof 
also shows that there are model
misspecifications where we expect  
$S(\E) = \emptyset$. If $Y$ is directly intervened on, for example,
under the assumption of Proposition~\ref{prop:modelmis}, 
we will not be able
to find any set $S$ that satisfies~\eqref{eq:H0S-nonlin}.
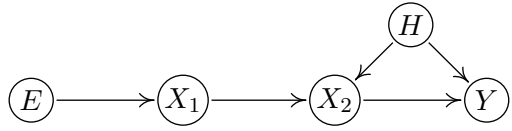
\begin{figure}
\begin{center}
\begin{tikzpicture}[xscale=2, yscale=1, line width=0.5pt, minimum size=0.58cm, inner sep=0.3mm, shorten >=1pt, shorten <=1pt]
    \normalsize
    \draw (0,0) node(i) [circle, draw] {$E$};
    \draw (1,0) node(x1) [circle, draw] {$X_1$};
    \draw (2,0) node(x2) [circle, draw] {$X_2$};
    \draw (3,0) node(y) [circle, draw] {$Y$};
    \draw (2.5,1) node(h) [circle, draw] {$H$};
    \draw[-arcsq] (i) -- (x1);
    \draw[-arcsq] (x1) -- (x2);
    \draw[-arcsq] (x2) -- (y);
    \draw[-arcsq] (h) -- (y);
    \draw[-arcsq] (h) -- (x2);
\end{tikzpicture}
\end{center}
\caption{
This graph corresponds to a model misspecification in the sense that the assumptions of Proposition~\ref{propos:sem}  and assumption (ii) c) of Proposition~\ref{propos:semb} are not satisfied. Indeed, we find that $H_{0,S}$ is violated for $S = S^0 := \{X_2\}$. And since $H_{0,S}$ is satisfied for 
both $S = \{X_1, X_2\}$ and $S = \{X_1\}$, we obtain
$S(\E) = \{X_1\}$. Therefore, $S(\E)$ is not a subset of $S^0$ but it is still a subset of the ancestors $\mathbf{AN}(Y)$ of $Y$, see Proposition~\ref{prop:modelmis}.
}
\label{fig:exmisspef}
\end{figure}

\begin{proof}
We first note that $H_{0,S,nonlin}(\E)$~\eqref{eq:H0S-nonlin-cond} holds if and only if $Y \independent E \given X_S$. 
Because of faithfulness this is the same as $Y$ and $E$ being $d$-separated given $X_S$ in the augmented graph.
Assume now that the latter holds for some set $S \subseteq \{X_2, \ldots, X_{p+1}\}$. (Such a set $S$ does not exist if $Y$ is directly intervened on.) The proposition follows if we can construct a set $\tilde S \subseteq \mathbf{AN}(Y) \cap \{X_2, \ldots, X_{p+1} \}$ that satisfies
$Y$ and $E$ being $d$-separated given $X_{\tilde S}$.

Assume that not all nodes in $S$ are ancestors of $Y$. Define then $W \in S$ to be one ``youngest'' non-ancestor in $S$, that is, $W \not\in \mathbf{AN}(Y)$ and there is no directed path from $W$ to any other node in $S$. (Such a node must exist since otherwise all youngest nodes of $S$ are in $\mathbf{AN}(Y)$, which implies $S \subseteq \mathbf{AN}(Y)$.) We now prove that for 
$$\tilde S := S \setminus \{W\}$$
we have
$Y$ and $E$ are $d$-separated given $X_{\tilde S}$.
To see this, consider any path from $E$ to $Y$. If this path does not go through $W$, the path is blocked by $\tilde S$ because it was blocked by $S = \tilde S \cup \{W\}$ (removing nodes outside a path can -if anything- only block it). Consider now a path that passes $W$ and the two edges connected to $W$ that are involved in this path. If both edges are into $W$, we are done because removing $W$ does not open the path. If one of these edges goes out of $W$, there must be a collider on this path which is a descendant of $W$ ($E$ does not have incoming edges and $W$ is not an ancestor of $Y$). But because $W$ is the youngest node in $S$ neither the collider nor any of its descendants is in $S$. We can therefore remove $W$ and the path is still blocked. 
\end{proof}

\section{Potential Outcomes and Invariant Prediction} \label{app:potentialoutcomes}
We now sketch that 
the assumption of invariant prediction can also be satisfied in a potential outcome framework \citep[e.g.][]{rubin2005causal}: 
as long as we do not intervene on the target variable $Y$,
the conditional distributions of $Y$ given the of causal predictors remains invariant. (Here, we discuss the nonlinear setting and therefore develop a result that corresponds to Remark~\ref{rem:extpropos:sem} rather than Proposition~\ref{propos:sem}.)
Although other formulations may be possible, too, we adopt the counterfactual language introduced by \citet{richardson2013single} who refer to finest fully randomised causally interpretable structured tree graphs (FFR-CISTG) \citep{Robins1986}.
We further consider the nonlinear version~\eqref{eq:H0S-nonlin-cond} of invariant prediction, see also Remark~\ref{rem:extpropos:sem}.

Similar as in \citep[][Definition~1]{richardson2013single}, we consider random variables~$\mathbf{V}:=(X_1 = Y, X_2, \ldots, X_p, X_{p+1})$ and assume the existence of counterfactual variables 
$X_j(\tilde{\mathbf{r}})$, for any assignment $\tilde{\mathbf{r}}$ to a subset $\mathbf{R} \subseteq \mathbf{V}$ and for all $j \in \{1, \ldots, p+1\}$. 
We further assume 
\begin{itemize}
\item[$\quad$(C1)] ``consistency and recursive substitution'' \citep[][equation~(14)]{richardson2013single} and
\item[$\quad$(C2)] ``FFR-CISTG independence'' 
\citep[][equation~(17)]{richardson2013single}\,.
\end{itemize}
To ease notation, we require $X_j(x_j = \tilde r) = \tilde r$ rather than $X_j(x_j = \tilde r) = X_j$
\citep[][p. 21]{richardson2013single}. 
\begin{prop} \label{prop:potentialoutcomes}
Consider random variables~$\mathbf{V}:=(X_1 = Y, X_2, \ldots, X_p, X_{p+1})$
and denote the causes of $Y$ by $\B{P}:=\pa{1}$.
For each environment $e \in \mathcal{E}$
consider a set $\mathbf{R}^e \subseteq \mathbf{V} \setminus \{Y\}$ of treatment variables and an assignment $\tilde{\mathbf{r}}^e$, that is $X_j^e := X_j(\tilde{\mathbf{r}}^e)$.
Assuming (C1) and (C2), i.e. an FFR-CISTG model, we have that
\begin{equation} \label{eq:potent}
Y(\tilde{\mathbf{r}}^e) \given \B{P}(\tilde{\mathbf{r}}^e) = \mathbf{q}
\quad \overset{d}{=} \quad
Y(\tilde{\mathbf{r}}^f) \given \B{P}(\tilde{\mathbf{r}}^f) = \mathbf{q}
\end{equation}
for all $e,f \in \mathcal{E}$ and for all $\B{q}$ such that both sides of~\eqref{eq:potent} are well-defined. Therefore, the set $\B{P}$ of parents satisfies~\eqref{eq:H0S-nonlin-cond}.
\end{prop}
We have already seen in Appendix~\ref{app:propsemb}, that we can allow for some hidden variables, i.e., the assumption (C2) can be relaxed further.

\begin{proof}
We have for all $e \in \E$
\begin{align*}
Y(\tilde{\mathbf{r}}^e) \;\big|\; \B{P}(\tilde{\mathbf{r}}^e) = \mathbf{q}
\quad & \overset{}{=} \quad  
Y(\tilde{\mathbf{r}}^e) \;\big|\; (\B{P}\setminus \B{R})(\tilde{\mathbf{r}}^e) = \mathbf{q}_{\B{P} \setminus \B{R}}, (\B{P}\cap \B{R})(\tilde{\B{r}}) = \tilde{\B{r}}_{\B{P} \cap \B{R}}\\
& \overset{(*)}{=} \quad  
Y(\tilde{\mathbf{r}}^e) \;\big|\; (\B{P}\setminus \B{R})(\tilde{\mathbf{r}}^e) = \mathbf{q}_{\B{P} \setminus \B{R}}\\
& \overset{(+)}{=} \quad  
Y \;\big|\; (\B{P}\setminus \B{R}) = \mathbf{q}_{\B{P} \setminus \B{R}} , (\B{P}\cap \B{R}) = \tilde{\B{r}}_{\B{P} \cap \B{R}},
\end{align*}
where we have used  
$(\B{P}\cap \B{R})(\tilde{\B{r}}) = \tilde{\B{r}}_{\B{P} \cap \B{R}}$
in $(*)$
and both (C1) and the modularity property
\citep[][Proposition~16]{richardson2013single} in $(+)$. 
This proves the statement because the latter expression is an observational distribution.
All equality signs should be understood as holding in distribution. 
\end{proof}

\section{Proof of Proposition~\ref{prop:popIV}} \label{app:popIV}
\begin{proof}
The residuals $Y - X\gamma$ for $\gamma\in \mathbb{R}^p$ are
given by 
$g(H,\varepsilon) + (\gamma^*-\gamma) f(H,\eta) +  Z 1_{I=1} (\gamma^*-\gamma)$. 
The two environments $\E$ are equivalent to conditioning on $I=0$ for the
first environment and $I=1$ for the second environment. 
Since $I, H,\varepsilon,\eta,Z$ are
independent and $Z$ has a full-rank covariance matrix, the
distribution of the residuals can only be invariant
between the two environments if $\gamma - \gamma^*\equiv 0$. Hence the
test of $H_{0,S,\mathit{hidden}}(\E)$ will be rejected for 
$S\neq S^*$,
whereas the true null
$H_{0,S^*,\mathit{hidden}}(\E)$ is accepted with probability at least
$1-\alpha$ by construction of the test and the result follows by the
definition of $\hat{S}$ in~\eqref{eq:hatScausalIV}.
\end{proof}

\section{Proofs of Section~\ref{sec:idres}} \label{app:proofs}

\subsection{Proof of Theorem~\ref{prop:1} (i)}
\begin{proof}
As shown in Proposition~\ref{propos:sem} we have $S(\E) \subseteq \PA{Y}$ because the null hypothesis~\eqref{eq:H0S} is correct for $S^{\caus}=\PA{Y}$. 
We assume that $S(\E) \neq \PA{Y}$ and deduce a contradiction.

As in~\eqref{eq:betapred} we define the regression coefficient
$$
\beta^{\pred,\e}(S) := \mbox{argmin}_{\beta\in \mathbb{R}^{p}: \beta_k=0 \mbox{ if } k \notin S} \;  E (Y^\e-X^\e \beta)^2.
$$
We then look for sets $S \subseteq \{1, \ldots, p\}$ such that for all $e_1,e_2 \in \E$
$$
\beta^{\pred, e_1}(S) = \beta^{\pred, e_2}(S) \quad \text{ and } \quad R^{e_1}(S) \equald R^{e_2}(S),
$$
with $R^{e_1}(S) := Y^{e_1} - X^{e_1} \beta^{\pred, e_1}(S)$ and $R^{e_2}(S) := Y^{e_2} - X^{e_2} \beta^{\pred, e_2}(S)$
(``constant beta'' and ``same error distribution'').
If $S(\E) \neq \PA{Y}$, then there must be a set 
$S \nsupseteq \PA{Y}$ whose null hypothesis is correct and that 
satisfies $\beta^{\pred, e}(S) \neq \beta^{\pred, e}(S^{\caus}) = \gamma^*$.
This set $S$ leads to the following residuals for $e=1$:
$$
R^1(S) = Y^1 - \sum_{k=2}^{p+1} \beta^{\pred,1}(S)_{k} X_k^1 = \sum_{k=2}^{p+1} \alpha_{k} X_k^1 + \varepsilon_{1}^1,
$$
with $\alpha_{k} := \gamma^{\caus}_{k} - \beta^{\pred,1}(S)_{k} =
\gamma^{\caus}_{k} - \beta^{\pred,\e}(S)_{k}$ 
for any $\e\in \E$ and $\alpha_k \neq 0$ for some (possibly more than one) $k \in \{2, \ldots, p+1\}$.

Among the set of {\it all} nodes (or variables) $X_{k}^1$ that have
non-zero $\alpha_k$, we consider a ``youngest'' node $X_{k_0}^1$ with the
property that there is no directed path from this node to any other node
with non-zero $\alpha_k$. We further consider experiment $e_0$ with
$\C{A}^{\e_0} = \{k_0\}$. This yields 
\begin{align}
R^1(S) &= \alpha_{k_0} X_{k_0}^1 + \sum_{k=2, k \neq k_0}^{p+1} \alpha_{k} X_k^1  + \varepsilon_{1}^1\quad \text{ and} \label{eq:help1}\\
R^{\e_0}(S) &= \alpha_{k_0} a_{k_0}^{\e_0} + \sum_{k=2, k \neq k_0}^{p+1} \alpha_{k} X_k^1  + \varepsilon_{1}^1, \label{eq:help2}
\end{align}
Since $\mean( X_{k_0}^1) \neq a_{k_0}^{\e_0}$,  
$R^{\e_0}(S)$ and $R^1(S)$ cannot have the same distribution.
This yields a contradiction.
\end{proof}

\subsection{Proof of Theorem~\ref{prop:1} (ii)}
\begin{proof}
As before we obtain equations~\eqref{eq:help1} and~\eqref{eq:help2}
for a ``youngest'' node $X_{k_0}^1$ among all nodes with non-zero $\alpha_{k_0}$ and an experiment $\e_0$ with $\C{A}^{\e_0} = \{k_0\}$. 
We now iteratively use the structural equations in order to obtain
\begin{align}
R^1(S) &= \alpha_{k_0} \varepsilon_{k_0}^1 + \sum_{k=1, k \neq k_0}^{p+1} \tilde \alpha_{k} \varepsilon_k^1  \quad \text{ and} \label{eq:52a}\\
R^{\e_0}(S) &= \alpha_{k_0} A_{k_0}^\e \varepsilon_{k_0}^1 + \sum_{k=1, k \neq k_0}^{p+1} \tilde \alpha_{k} \varepsilon_k^1 . \label{eq:52b}
\end{align}
Since all $\varepsilon_{k}^\e$ are jointly independent and $\mean (A_{k_0}^{\e_0})^2 \neq 1$, $R^1(S)$ and $R^{\e_0}(S)$ cannot have the same distribution.
This contradicts the fact that the null hypothesis~\eqref{eq:H0S} is correct for $S$. The proof works analogously for the shifted noise distributions.
\end{proof}

\subsection{Proof of Theorem~\ref{prop:1} (iii)}
\begin{proof}
We start as before and obtain analogously to equations~\eqref{eq:52a} and~\eqref{eq:52b} the equations
\begin{align*}
R^1(S) &= \alpha_{k_0} \varepsilon_{k_0}^1 + \sum_{k=1, k \neq k_0}^{p+1} \tilde \alpha_{k} \varepsilon_k^1  \quad \text{ and} \\
R^2(S) &= \alpha_{k_0} A_{k_0} \varepsilon_{k_0}^1 + \sum_{k=1, k \neq k_0}^{p+1} \tilde D_{k} \varepsilon_k^1 ,
\end{align*}
where the $\tilde D_k$ are continuous functions of the random variables 
$A_{s}, s \in \{2, \ldots, p+1\} \setminus \{k_0\}$ and
$\beta_{j,s}^{e=2}, j,s \in \{2, \ldots, p+1\}$ (and therefore random variables themselves).
$R^1(S)$ and $R^2(S)$ are supposed to have the same distribution. It follows from Cram\'er's theorem \citep{cramer} that $A_{k_0} \varepsilon_{k_0}^1$ must be normally distributed. 
But then it follows that
\begin{align*}
\mean[(A_{k_0})^4] \,\mean[(\varepsilon_{k_0}^1)^4] 
&= \mean[(A_{k_0} \varepsilon_{k_0}^1)^4] 
= 3 \mean[(A_{k_0} \varepsilon_{k_0}^1)^2]^2\\
&= 3 \mean[(A_{k_0})^2]^2 \,\mean[(\varepsilon_{k_0}^1)^2]^2  
= \mean[(A_{k_0})^2]^2 \,\mean[(\varepsilon_{k_0}^1)^4] \,
\end{align*}
and therefore
$$
\var(A_{k_0}^2) = 0
$$
which means $P[A_{k_0} \in \{-c,c\}] = 1$ for some constant 
$c \geq 0$. This contradicts the assumption that $A_{k_0}$ has a density.
\end{proof}

\subsection{Proof of Theorem~\ref{prop:onlyone}}
\begin{proof}
The proof follows directly from Lemma~\ref{lem:2} (see below) and the
fact that faithfulness is satisfied with probability one \citep[][Theorem~3.2]{Spirtes2000}.
Assume that the null hypothesis~\eqref{eq:H0Sregr} is accepted for $S$
with $S^{\caus} \setminus S \neq \emptyset$. Lemma~\ref{lem:2} implies
that with probability one, we have $\alpha_{k_0} \neq 0$, where
$\alpha$ is defined as in~\eqref{eq:resy}. (Otherwise, we 
construct a new SEM by replacing the equation for $Y$ with $Y_{k_0} := \sum_{k \in S^* \setminus \{k_0\}} \gamma^*_k X_k +\varepsilon_{1}$
and removing all equations for the descendants of $Y$. 
Equation~\eqref{eq:lem2} then reads a violation of faithfulness since
there is a path between $k_0$ and $Y_{k_0}$ via nodes in $S^{\caus}
\setminus S$ that is unblocked given $S \setminus \{k_0\}$.) 
But if $\alpha_{k_0} \neq 0$, we can use exactly the same arguments as in the proof of Theorem~\ref{prop:1}.
\end{proof}

\begin{lemma} \label{lem:2}
Assume that the joint distribution of $(X_1, \ldots, X_{p+1})$ is generated by a
structural equation model~\eqref{eq:semmmmm} with all non-zero parameters
$\beta_{j,k}$ and $\sigma_j^2$ drawn from a joint density
w.r.t.\ Lebesgue measure. Let
$X_{k_0}$ denote a youngest parent of target variable $Y=X_{1}$. Let $S$
be a set with $S^* \setminus S \neq \emptyset$, that is, some of the true
causal parents are missing in the set $S$. Consider the residuals 
\begin{align} \label{eq:resy}
Res(Y) &= \sum_{k \in S^*} \gamma^{*}_k X_k - \sum_{k \in S} \beta^{\pred, 1}(S)_k X_k + \varepsilon_{1}^1\\
&= \sum_{k \in S^*} \alpha_k X_k + \sum_{k \notin S^*} \alpha_k \varepsilon_k^1 \nonumber
\end{align}
where the second equation is obtained by iteratively using the structural equations except the ones for the parents $S^*$ of $Y$. \\
Then for almost all parameter values, we have: $\alpha_{k_0} = 0$
implies
${k_0} \in S$
   and
\begin{equation} \label{eq:lem2}
X_{k_0} \perp Y_{k_0} \given X_{\tilde S\setminus \{k_0\}},
\end{equation}
where $Y_{k_0} := \sum_{k \in S^* \setminus \{k_0\}} \gamma^*_k X_k +\varepsilon_{1}$
and $\tilde S := S \cap \ND{k_0}$
with $\ND{k_0}$ being the non-descendants of $k_0$.
\end{lemma}
\begin{proof}
With probability one, we have $\gamma_{k_0}^* \neq 0$. Hence, $\alpha_{k_0} = 0$ can happen only if ${k_0} \in S$ or $S$ contains a descendant of $X_{k_0}$ (otherwise $\alpha_{k_0} = \gamma_{k_0}^* \neq 0$).
We will now show that in fact ${k_0} \in S$ must be true.
Let the random vector $X_S$ contain all variables $X_k$ with $k \in S$
and
let it be topologically ordered such that if $X_{k_2}$ is a descendant of $X_{k_1}$, it appears after $X_{k_1}$ in the vector $X_S$. 
Assume now that $S$ contains a descendant of $X_{k_0}$. 
W.l.o.g., we can assume that the $|S|$-entry of $X_S$ 
(i.e. its last component) is a ``youngest'' descendant $X_s$ of $X_{k_0}$ in $S$, that is, there is no directed path from $X_s$ to any other descendant of $X_{k_0}$ in $S$.
The entry $(|S|,|S|)$ of the matrix $\big(\mean {X_{S}^1}^t X_{S}^1\big)$ is the only entry depending (additively) on the parameter $\sigma_{s}^2$, we call this entry $d$. With 
$$
\big(\mean {X_{S}^1}^t X_{S}^1\big) =: \left(
\begin{array}{cc}
A & b\\
b^T & d
\end{array} \right)
$$
it follows 
\begin{align*}
\big(\mean {X_{S}^1}^t X_{S}^1\big)^{-1} &= \left(
\begin{array}{cc}
A^{-1} + \frac{A^{-1}bb^TA^{-1}}{d-b^TA^{-1}b} & \frac{A^{-1}b}{d-b^TA^{-1}b}\\
\frac{b^TA^{-1}}{d-b^TA^{-1}b} & \frac{1}{d-b^TA^{-1}b}
\end{array} \right)
=: 
\left( \begin{array}{cc}
A^{-1} & 0\\
0 & 0
\end{array} \right)
+ \frac{1}{d-b^TA^{-1}b} \; C
\end{align*}
Observe that $\big(\mean {X_{S}^1}^t X_{S}^1\big)$ is non-singular
with probability one (if the matrix is non-singular, the full covariance matrix over $(X_2, \ldots, X_{p+1})$ is non-singular, too) 
and
$$
\beta^{\pred,1}(S) = 
\big(\mean {X_{S}^1}^t X_{S}^1\big)^{-1}
\;
\xi
$$
for $\xi := {\mean X_S^1}^t Y^1 \neq 0$ 
(otherwise $\beta^{\pred,1}(S)$ would be zero and thus $\alpha_{k_0} = \gamma_{k_0}^{\caus} \neq 0$).

According to formula~\eqref{eq:resy} and $\alpha_{k_0}=0$, computing the linear coefficients $\beta^{\pred,1}(S)$
and subsequently using the
true structural equations, 
leads to the following
relationship between the true coefficients $\beta_{j,k}$ and
$\gamma^{\caus}$: 
$$
\gamma^*_{k_0} = 
\eta^t_S \; \beta^{\pred,1}(S),
$$
where $\eta_S$ depends on the true coefficients $\beta_{j,k}$ and is constructed in the following way: the $i$-th component of $\eta_S$ is obtained 
by multiplying the path coefficients between $X_{k_0}$ and $X_i$. For example, the two directed paths 
$X_{k_0} \rightarrow X_5 \rightarrow X_3  \rightarrow X_i$ and $X_{k_0} \rightarrow X_5 \rightarrow X_i$, lead to the corresponding $i$th entry 
$\eta_{S,i} = \beta_{5,k_0}^1 \beta_{3,5}^1 \beta_{i,3}^1 + \beta_{5,k_0}^1
\beta_{i,5}^1$. All non-descendants of $k_0$ have a zero entry in $\eta_S$, $k_0$ itself has the entry one in $\eta_S$ if $k_0 \in S$ (we will see below that this must be the case).
But then, we have:
\begin{equation} \label{eq:ggg}
\gamma^*_{k_0} = 
\eta^t_S \; \beta^{\pred,1}(S)
=
\eta^t_S
\;
\big(\mean {X_{S}^1}^t X_{S}^1\big)^{-1}
\;
\xi
=
\eta^t_S \;
\left( \begin{array}{cc}
A^{-1} & 0\\
0 & 0
\end{array} \right)
\;
\xi 
+
\frac{1}{d-b^TA^{-1}b}\;
\eta^t_S
\;
 C
\;
\xi.
\end{equation}
If $X_s \neq X_{k_0}$ then $\xi$ does not depend on $\sigma_s^2$ (it does if $X_s = X_{k_0}$). We must then have that $\eta^t_S C \xi = 0$ since otherwise
it follows from~\eqref{eq:ggg} that
$$
d
= b^TA^{-1}b \; + \;
\frac{\eta^t_S
\;
 C
\;
\xi}
{ 
\gamma^*_{k_0} - \eta_S^t \;
\left( \begin{array}{cc}
A^{-1} & 0\\
0 & 0
\end{array} \right)
\;
\xi 
},
$$
which can happen only with probability zero (it requires a ``fine-tuning''
of the parameter $\sigma_{s}^2$; note that $d$ is depending on
  $\sigma_s^2$).

But if $\eta^t_S C \xi = 0$ then $\gamma_{k_0}^* = (\eta_1 \cdots
\eta_{|S|-1}) A^{-1} (\xi_1, \cdots \xi_{|S|-1}) = \eta_{\tilde S_1}^t \; \beta^{\pred,1}(\tilde S_1)$ 
with $\tilde S_1 := S \setminus \{s\}$, an equation analogue to the first part of~\eqref{eq:ggg}.
We can now repeat the same argument for $\tilde S_1$ (assume that $\tilde S_1$ contains a descendant of $k_0$, then consider the youngest descendant of $k_0$ in $\tilde S_1$\ldots) and obtain $\tilde{S}_2$.
After $\ell$ iterations, we obtain 
$\gamma_{k_0}^* = \eta_{\tilde S}^t \; \beta^{\pred,1}(\tilde S)$, 
  where $\tilde S := \tilde S_{\ell}$ does not contain any descendant of
  $k_0$. The only non-zero entry of $\eta_{\tilde S}$ is the one for $k_0$
  (otherwise all remaining $\eta_{\tilde S}$ entries would be zero which
  implies $\gamma_{k_0}^{\caus}=0$).  

We have thus shown that $k_0 \in S$ and that 
$\beta^{\pred,1}(\tilde S)_{k_0} = \gamma^{\caus}_{k_0}$
  with 
$\tilde S := S \cap \ND{k_0}$.
We obtain~\eqref{eq:lem2} with the following argument: regressing $Y$ on $\tilde S$ yields a regression coefficient $\gamma^*_{k_0}$ for $X_{k_0}$; thus, regressing $Y_{k_0} = Y - \gamma^*_{k_0} X_{k_0}$ on $\tilde S$ yields a regression coefficient zero for $X_{k_0}$.
\end{proof}

\section{Experimental settings for numerical studies} \label{app:pars}
We sample $n_{obs}$ data points from an observational and $n_{int}$ data points from an interventional setting $(|\mathcal{E}| = 2)$. We first sample a directed acyclic graph with $p$ nodes that is common to both scenarios. In order to do so, we choose a random topological order and then connect two nodes with a probability of $k/(p-1)$. This leads to an average degree of $k$. Given the graph structure, we then sample non-zero linear coefficients with a random sign and a random absolute value between a lower bound $lb^{e=1}$ and an upper bound $ub^{e=1} = lb^{e=1} + \Delta_b^{e=1}$. We consider normally distributed noise variables with a random variance between $\sigma^2_{\text{min}}$ and $\sigma^2_{\text{max}}$. We can then sample the observational data set ($e=1$).

For the interventional setting ($e=2$), we choose simultaneous noise
interventions (Section~\ref{sec:det}) with the extension of changing
linear coefficients, that is for $j \in \mathcal{A}$ (where even
$\mathcal{A}$ is random and can include the later target of interest $Y$), we have 
$\varepsilon_j^{e=2} = A_j \varepsilon_j^{e=1}$ and (possibly) $\beta_{j,s}^{e=2} \neq \beta_{j,s}^{e=1}$. 
The set $\mathcal{A}$ of intervened nodes contains either a single node or a fraction $\theta$ of nodes.
We chose $A_j$ to be uniformly distributed random variables that take values between $a_{min}$ and $a_{min} + \Delta_a$. The linear coefficients $\beta_{j,s}^{e=2}$ are chosen either equal to $\beta_{j,s}^{e=1}$ or according the same procedure with corresponding bounds $lb^{e=2}$ and $ub^{e=2}$.

All parameters were sampled independently for each of the scenarios,
uniformly in a given range that is shown below in brackets (or with
given probability for discrete parameters).
(1) The number $n_{obs}$ of samples in the observational data is chosen uniformly from $\{100,200,300,400,500\}$.
(2)  The number $n_{int}$ of samples in intervention data is chosen uniformly from $\{100,200,300,400,500\}$.
(3) The number $p$ of nodes in the graph is chosen uniformly from $\{5, 6, 7, \ldots, 40\}$.
(4) The average degree $k$ of the graph is chosen uniformly from $\{1,2,3,4\}$.
(5) The lower bound $lb^{\e=1}$ is chosen uniformly from $\{0.1,0.2, \ldots, 2\}$.
(6) The maximal difference $\Delta_b^{\e=1}$ between largest and smallest coefficients is chosen uniformly from $\{0.1,0.2, \ldots, 1\}$.
(7) The minimal noise variance $\sigma^2_{\text{min}}$ is chosen uniformly from $\{0.1,0.2, \ldots, 2\}$ and
(8) the maximal noise variance $\sigma^2_{\text{max}}$ uniformly from $\{0.1,0.2, \ldots, 2\}$, yet at least equal
  to $\sigma^2_{\text{max}}$.
(9) The lower bound $a_{j,min}$ for the noise multiplication is chosen uniformly from $\{0.1, 0.2, \ldots, 4\}$.
(10) The difference $\Delta_a$ between upper and lower bound $a_{j,min}$ for noise
  multiplication is chosen to be zero with probability $1/3$ (which results in fixed coefficients) and otherwise uniformly from $\{0.1, 0.2, \ldots, 2\}$.
 (11) The interventional coefficients are chosen to be identical ($\beta_{j,s}^{e=2} = \beta_{j,s}^{e=1}$) with probability $2/3$, otherwise they are chosen uniformly between $lb^{e=2}$ and $ub^{e=2}$.
 (12) The lower bound $lb^{e=2}$ for new coefficients under interventions is chosen as the smaller
  value of two uniform values in $\{0.1, 0.2, \ldots, 2\}$ and
(13) the upper bound $ub^{e=2}$ for new coefficients under interventions as the corresponding larger value.
(14) With probability $1/6$ we intervene only on one (randomly chosen) variable, that is $|\mathcal{A}| = 1$.
(15) Otherwise, the inverse fraction $1/\theta$ is chosen uniformly from $\{1.1, 1.2, \ldots, 3\}$, that is the fraction of intervened nodes varies between $\theta = 1/3$ and $\theta = 1/1.1$.

\end{document}